\documentclass{article}
\usepackage{fullpage}
\usepackage[utf8]{inputenc}

\usepackage{nicematrix}
\usepackage{physics}
\usepackage{enumitem}
\usepackage{amsthm,amssymb}
\usepackage{mathtools}
\usepackage{algorithm}
\usepackage[noend]{algpseudocode}

\usepackage{booktabs}

\usepackage{xspace}
\usepackage{comment}

\usepackage{hyperref}
\hypersetup{
    colorlinks=true,       
    linkcolor=blue,        
    citecolor=magenta,         
    filecolor=magenta,     
    urlcolor=cyan,         
    linktoc=all
}
\usepackage[capitalise]{cleveref}
\usepackage{upref}

\usepackage{thm-restate}
\usepackage{thmtools}

\theoremstyle{plain}
\newtheorem{theorem}{Theorem}
\numberwithin{theorem}{section}
\newtheorem{lemma}[theorem]{Lemma}
\newtheorem{corollary}[theorem]{Corollary}
\newtheorem{proposition}[theorem]{Proposition}
\newtheorem{conjecture}[theorem]{Conjecture}
\newtheoremstyle{claimstyle}{\topsep}{\topsep}{}{0pt}{\sffamily}{. }{5pt plus 1pt minus 1pt}%
  {$\vartriangleright$ \thmname{#1}\thmnumber{ #2}\thmnote{ (#3)}}
\theoremstyle{claimstyle}
\newtheorem{claim}[theorem]{Claim}
\newtheorem{case0}[theorem]{Case}
\theoremstyle{definition}

\theoremstyle{remark}
\newtheorem{remark}[theorem]{Remark}

\newenvironment{claimproof}{\begin{proof}}{\end{proof}}

\newlength{\bibitemsep}
\newlength{\bibparskip}
\let\oldthebibliography\thebibliography
\renewcommand\thebibliography[1]{%
  \oldthebibliography{#1}%
  \setlength{\parskip}{\bibitemsep}%
  \setlength{\itemsep}{\bibparskip}%
}

\DeclareMathOperator{\lat}{lat}

\newcommand{\textscup}[1]{\textsc{#1}}
\newcommand{\textsfup}[1]{\textsf{\textup{#1}}}
\renewcommand{\P}{\textsfup{P}\xspace}
\newcommand{\NP}{\textsfup{NP}\xspace}
\newcommand{\NPhard}{\textsfup{NP-hard}\xspace}

\newcommand{\N}{\mathbb{N}}
\newcommand{\Z}{\mathbb{Z}}
\newcommand{\Zp}{\Z_{\ge 0}}
\newcommand{\Q}{\mathbb{Q}}
\newcommand{\R}{\mathbb{R}}

\newcommand{\cB}{\mathcal{B}}
\newcommand{\cH}{\mathcal{H}}
\newcommand{\ones}{\mathbf{1}}

\newcommand{\GF}{\mathrm{GF}}

\DeclareMathOperator*{\argmin}{arg\,min}

\newcommand{\symdif}{\mathbin{\triangle}}
\DeclarePairedDelimiter{\set}{\{}{\}}
\DeclarePairedDelimiterX{\Set}[2]{\{}{\}}{\,#1\mathclose{}\nonscript\;\delimsize|\nonscript\;\mathopen{}#2\,}

\def\final{0}  
\def\iflong{\iffalse}
\ifnum\final=0  
\newcommand{\anote}[1]{{\color{orange}[{\tiny \textbf{András:} \bf #1}]\marginpar{\color{orange}*}}}
\newcommand{\fnote}[1]{{\color{red}[{\tiny \textbf{Florian:} \bf #1}]\marginpar{\color{red}*}}}
\newcommand{\rnote}[1]{{\color{teal}[{\tiny \textbf{Ryuhei:} \bf #1}]\marginpar{\color{teal}*}}}
\newcommand{\onote}[1]{{\color{purple}[{\tiny \textbf{Taihei:} \bf #1}]\marginpar{\color{purple}*}}}
\newcommand{\snote}[1]{{\color{cyan}[{\tiny \textbf{Tamás:} \bf #1}]\marginpar{\color{cyan}*}}}

\else 
\newcommand{\anote}[1]{}
\newcommand{\fnote}[1]{}
\newcommand{\rnote}[1]{}
\newcommand{\onote}[1]{}
\newcommand{\snote}[1]{}

\fi

\title{Problems on Group-labeled Matroid Bases}

\author{
Florian Hörsch\thanks{CISPA, Saarbrücken, Germany. E-mail: \texttt{florian.hoersch@cispa.de}.}
\and
András Imolay\thanks{MTA-ELTE Momentum Matroid Optimization Research Group, Department of Operations Research, ELTE Eötvös Loránd University, Budapest, Hungary. E-mail: \texttt{\{andras.imolay, tamas.schwarcz\}@ttk.elte.hu}.}
\and
Ryuhei Mizutani\thanks{Department of Mathematical Informatics, Graduate School of Information Science and Technology, The University of Tokyo, Tokyo, Japan. E-mail: \texttt{\{ryuhei\_mizutani, oki\}@mist.i.u-tokyo.ac.jp}.}
\and
Taihei Oki\footnotemark[3]
\and 
Tamás Schwarcz\footnotemark[2]
}

\begin{document}
\maketitle

\begin{abstract}

Consider a matroid equipped with a labeling of its ground set to an abelian group. We define the label of a subset of the ground set as the sum of the labels of its elements. We study a collection of problems on finding bases and common bases of matroids with restrictions on their labels. For zero bases and zero common bases, the results are mostly negative. While finding a non-zero basis of a matroid is not difficult, it turns out that the complexity of finding a non-zero common basis depends on the group. Namely, we show that the problem is hard for a fixed group if it contains an element of order two, otherwise it is polynomially solvable. 

As a generalization of both zero and non-zero constraints, we further study $F$-avoiding constraints where we seek a basis or common basis whose label is not in a given set $F$ of forbidden labels. Using algebraic techniques, we give a randomized algorithm for finding an $F$-avoiding common basis of two matroids represented over the same field for finite groups given as operation tables. The study of $F$-avoiding bases with groups given as oracles leads to a conjecture stating that whenever an $F$-avoiding basis exists, an $F$-avoiding basis can be obtained from an arbitrary basis by exchanging at most $|F|$ elements. We prove the conjecture for the special cases when $|F|\le 2$ or the group is ordered. By relying on structural observations on matroids representable over fixed, finite fields, we verify a relaxed version of the conjecture for these matroids. As a consequence, we obtain a polynomial-time algorithm in these special cases for finding an $F$-avoiding basis when $|F|$ is fixed.

\medskip


\end{abstract}

\section{Introduction}

Several combinatorial optimization problems involve additional constraints, such as parity, congruency, and exact-weight constraints~\cite{liu2023congruency,nagele2023ccc,nagele2019submodular,nagele2020new,papadimitriou1982exact}.
These constraints are subsumed by \emph{group-label constraints} defined as follows: the ground set $E$ is equipped with a labeling $\psi\colon E \to \Gamma$ to an abelian group $\Gamma$ and a solution $X \subseteq E$ must ensure that the sum of the labels of its entries is not in a prescribed forbidden set $F \subseteq \Gamma$, i.e., $\psi(X) \coloneqq \sum_{e \in X} \psi(e) \notin F$.
We call such a solution \emph{$F$-avoiding}.

Particularly important special cases of group-label constraints include the \emph{non-zero} ($F = \set{0}$) and \emph{zero} ($F = \Gamma \setminus \set{0}$) constraints, where $F$-avoiding sets are referred to as \emph{non-zero} and \emph{zero}, respectively.
The non-zero constraint has been extensively studied for path problems on graphs as it generalizes constraints on parity and topology.
This line of research includes packing non-zero $A$-paths~\cite{chudnovsky2006apath} as well as finding a shortest odd $s$--$t$ path~\cite[Section~29.11e]{schrijver2003combinatorial}, a shortest non-zero $s$--$t$ path~\cite{iwata2022finding}, and an $F$-avoiding $s$--$t$ path with $|F| \le 2$~\cite{kawase2020twoforbidpath}.
For these problems, some literature allows $\Gamma$ to be non-abelian since the order of operations can be naturally defined for paths.
Problems related to non-zero perfect bipartite matchings in $\Z_2$ have also been dealt with, see~\cite{Artmann2017,elmaalouly2023exact,jia2023exact}. 
The zero constraint, or, slightly more generally, the group-label constraint with $|\Gamma \setminus F| = 1$, can encode the congruency and exact-weight constraints by setting $\Gamma$ to be a cyclic group $\Z_m$ and the integers $\Z$, respectively.
Examples of congruency-constrained problems include submodular function minimization~\cite{nagele2019submodular}, minimum cut~\cite{nagele2020new}, and integer linear programming with totally unimodular coefficients~\cite{nagele2023ccc}.
For the last problem, N\"{a}gele, Santiago, and Zenklusen~\cite{nagele2023ccc} gave a randomized strongly polynomial-time algorithm to test the existence of an $F$-avoiding feasible solution, where the group is $\Z_m$ with prime $m$ and $|F| \le 2$, implying the congruency constraint if $m = 3$.
%
The exact-weight constraint was first considered for the matching problem by Papadimitriou and Yannakakis~\cite{papadimitriou1982exact}.
Mulmuley, Vazirani, and Vazirani~\cite{mulmuley1987matching} gave a randomized polynomial-time algorithm for solving the problem using an algebraic technique. Derandomizing this algorithm is a major open problem and there is a collection of partial results for it, see e.g.~\cite{bhatnagar,elmaalouly2023exact,Galluccio1998,svensson2017,yuster}. Other exact problems include arborescences, matchings, cycles~\cite{barahona1987exact}, and independent sets or bases in a matroid~\cite{camerini1992exact,doronarad2023tight,rieder1991lattices}.

In this work, we explore group-label constraints for matroid bases and matroid intersection. Throughout the paper, we assume that any group $\Gamma$ is abelian without mentioning it. 
In the problems \textsc{Non-Zero Basis} and \textsc{Zero Basis}, we are given a matroid $M$ on a ground set $E$ and a group labeling $\psi\colon E \to \Gamma$, and we are to find a non-zero or zero basis, respectively, or to correctly report that no such basis exists. In \textsc{$F$-avoiding Basis} along with the matroid $M$ and labeling $\psi$, we are also given a forbidden label set $F \subseteq \Gamma$, and we need to find an $F$-avoiding basis, that is a basis $B$ with $\psi(B) \not \in F$. In \textsc{Non-Zero Common Basis}, \textsc{Zero Common Basis}, and \textsc{$F$-avoiding Common Basis}, instead of a single matroid, we are given two matroids $M_1$ and $M_2$ on a common ground set $E$ and seek a non-zero, zero, and $F$-avoiding common basis, respectively. We also tackle the weighted variants of these problems, referred to as \textsc{Weighted Non-Zero Basis} for example, where we are to find a feasible solution minimizing a given weight function $w\colon E \to \R$.

We note that the target label $0$ in the non-zero and zero problems can be changed to an arbitrary group element $g \in \Gamma$ by appending a coloop to the ground set with label $-g$. Regarding the input of the group, we consider the following three types: (i) operation and zero-test oracle, (ii) operation table of a finite group, and (iii) a fixed finite group.
Unless otherwise stated, we assume that a group is given as the oracles and the matroids are given as independence oracles. In this case, by a polynomial-time algorithm, we mean an algorithm making polynomially many elementary steps, group operations, and independence oracle calls. If the group is finite and is given by its operation table, then the running time of a polynomial-time algorithm can depend polynomially on the group size.

Our research follows the recent initiative by Liu and Xu~\cite{liu2023congruency}, who addressed \textsc{Zero Basis}\footnote{Called \emph{Group-Constrained Matroid Base} (GCMB) in~\cite{liu2023congruency}.}.
They conjectured that, given the existence of a zero basis, for any non-zero basis $B$, there is a zero basis $B^*$ such that $|B^* \setminus B| \le D(\Gamma) - 1$, where $D(\Gamma)$ denotes the Davenport constant of $\Gamma$ (see \cref{sec:counter} for definition) and is upper-bounded by $|\Gamma|$. 
Liu and Xu proved the conjecture for cyclic groups $\Gamma = \Z_m$ with the order $m$ being prime power or the product of two primes, with the aid of an additive combinatorics result by Schrijver and Seymour~\cite{schrijver1990spanning}, deriving an FPT algorithm parameterized by $|\Gamma| = m$ for \textsc{Zero Basis}. In \cref{thmcounter}, we give a counterexample to this conjecture.

The non-zero constraint is closely related to \emph{lattices} studied by Lovász~\cite{lovasz1985some}.
The \emph{lattice} generated by vectors $\set{v_1, \dotsc, v_n} \subseteq \R^n$ is the set $\Set{\sum_{i=1}^n \lambda_i v_i}{\lambda_1, \dotsc, \lambda_n \in \Z}$.
For a set family $\mathcal{F} \subseteq 2^E$, let $\lat(\mathcal{F})$ denote the lattice generated by the characteristic vectors of $\mathcal{F}$.
Every lattice has a \emph{lattice basis} $A = \set{a_1, \dotsc, a_n} \subseteq \Z^E$, which is a set of linearly independent vectors generating it.
Since $\mathcal{F}$ and its lattice basis $A$ generate the same lattice, $\mathcal{X}$ has a non-zero member if and only if $A$ has a non-zero member, i.e., $\psi(a_i) \coloneqq \sum_{e \in E} a_i(e) \psi(e) \ne 0$ for some $i$.
This implies that if a basis of $\lat(\mathcal{F})$ can be computed in polynomial time, then the existence of a non-zero member of $\mathcal{F}$ can be decided in polynomial time.
Such set families $\mathcal{F}$ include matroid bases~\cite{rieder1991lattices}, common bases of a matroid and a partition matroid having two classes~\cite{rieder1991lattices}, and perfect matchings~\cite{lovasz1985some,lovasz1987matching}. 

We below summarize existing and our results for each problem.

\vspace{-0.5em}

\begin{description}
    \item[Non-Zero Basis]
        The tractability of \textsc{Non-Zero Basis} can be derived from the above lattice argument together with a characterization of base lattices~\cite{rieder1991lattices}.
        We observe that for any zero basis $B$, there exists a non-zero basis $B^*$ such that $|B^* \setminus B| \le 1$, provided that at least one non-zero basis exists.
        A weighted variant of this statement is shown in the same way.
        This result generalizes an algorithm for \textsc{Weighted Zero Basis} with $\Gamma = \Z_2$ by Liu and Xu~\cite{liu2023congruency}.
        
    \item[Non-Zero Common Basis]
        We show in \cref{thm:non-zero-common-basis-hard,thm:non-zero-common-basis-poly} that \textsc{Non-Zero Common Basis} is polynomially solvable if and only if $\Gamma$ does not contain $\Z_2$ as a subgroup.
        Our hardness proof for $\Gamma = \Z_2$ is based on an information-theoretic argument using sparse paving matroids, which is independent of the $\P \ne \NP$ conjecture.
        The polynomial-time algorithm for $\Gamma \not\ge \Z_2$ is a modification of the negatively directed cycle elimination algorithm for weighted matroid intersection~\cite{brezovec1986two}.
        In \cref{thm:2approx}, we also give a 2-approximation algorithm for \textsc{Weighted Non-Zero Common Basis} if $\Gamma \not\ge \Z_2$ and the weight function is nonnegative. 
        Finally, in \cref{thm:non-zero-matching,thm:non-zero-exact-base}, we solve \textsc{Weighted Non-Zero Common Basis} for an arbitrary group when both matroids are partition matroids or one of the matroids is a partition matroid defined by a partition having two classes.
        
    \item[$F$-avoiding Basis and Common Basis] 
    If the group is fixed and finite, \textsc{(Weighted) $F$-avoiding Basis} reduces to polynomially many instances of (weighted) matroid intersection~\cite{liu2023congruency}.
    On the other hand, it follows from the results of \cite{doronarad2023tight} that \textsc{$F$-avoiding Basis} requires exponentially many independence oracle queries if $F$ is part of the input and the group is finite and is given as an operation table, while the same hardness of \textsc{$F$-avoiding Common Basis} follows from our \cref{thm:non-zero-common-basis-hard} even if the group is fixed and finite. 
    Regarding positive results for $F$-avoiding problems, our contribution is twofold. First, using similar ideas as in \cite{camerini1992exact, webb2004paffian}, in \cref{thm:representable-zero-common-basis} we give a randomized algebraic algorithm for \textsc{$F$-avoiding Common Basis} in case where the matroids are  represented over the same field and the group is finite and is given as an operation table. We observe in \cref{thm:pfaffian} that the algorithm can be derandomized in certain cases, including \textsc{$F$-avoiding Basis} for graphic matroids. Second, we turn to the study of \textsc{$F$-avoiding Basis} for cases where $|F|$ is fixed and the group is given by an oracle. Motivated by the work of Liu and Xu~\cite{liu2023congruency}, we propose \cref{conj:k} stating that if at least one $F$-avoiding basis exists, then each basis can be turned into an $F$-avoiding basis by exchanging at most $|F|$ elements. The validity of the conjecture follows from the results of \cite{liu2023congruency} for groups of prime order.
    We show that the conjecture also holds if $\Gamma$ is an ordered group (\cref{thm:ordered}) 
    or if $|F|\le 2$ (\cref{thm:two-forbidden}). By introducing a novel relaxation of strong base orderability, in \cref{thm:reprweak} we show that a relaxation of the conjecture holds for $\GF(q)$-representable matroids for a fixed prime power $q$. In \cref{thm:graphicwbo}, we prove a somewhat stronger version of this result for graphic matroids. In each of these special cases, we obtain the polynomial solvability of \textsc{$F$-avoiding Basis} for fixed $F$.

    \item[Zero Basis and Zero Common Basis] 
        The zero constraint for $\Gamma = \Z$ corresponds to the exact-weight constraint, implying that many problems are \NPhard, in particular, \textsc{Zero Basis} is \NPhard even for uniform matroids (\cref{thm:zerobasishard}). 
        It follows from the results of \cite{doronarad2023tight} that \textsc{Zero Basis} requires exponentially many independence oracle queries for a finite group given by operation table. We show the same hardness of \textsc{Zero Basis} for any fixed infinite group (\cref{thm:zero-basis-infinite}), and of \textsc{Zero Common Basis} for any fixed nontrivial group (\cref{thm:zero-common-base-fixed}).
        On the other hand, we obtain positive results from the aforementioned results on $F$-avoiding problems. In particular, \textsc{Zero Basis} is polynomially solvable if the group is fixed and finite~\cite{liu2023congruency}, \cref{thm:representable-zero-common-basis} implies a randomized polynomial-time algorithm for \textsc{Zero Common Basis} for matroids represented over the same field if $\Gamma$ is finite and is given as an operation table, and using the results of \cite{liu2023congruency}, \cref{thm:reprweak} implies an FPT algorithm for \textsc{Zero Basis} when parameterized by $|\Gamma|$ if the matroids are representable over a fixed, finite field (\cref{cor:zero-basis-fpt}).  
    \end{description}

\paragraph*{Other work related to group-labeled matroids}
Bérczi and Schwarcz~\cite{berczi2021complexity} showed hardness of partitioning into common bases, see also \cite{berczi2023complexity, horsch2022rainbow} for later results. A natural relaxation of that problem gives rise to problems related to \textsc{Non-Zero Common Basis} for the group $\R/\Z$. This relation is explained in \cref{hard1}.

It is straightforward to verify that the family of non-zero subsets of a set satisfies the axiom of \emph{delta-matroids}, which are a generalization of matroids introduced by Bouchet~\cite{bouchet1987greedy}.
From this viewpoint, \textsc{Non-Zero Basis} offers a tractable special case of the intersection of a matroid and a delta-matroid.
We note that the intersection of a matroid and a delta-matroid is intractable in general, as it encompasses matroid parity~\cite{jensen1982complexity,lovasz1980matroid}.
Kim, Lee, and Oum~\cite{kim2023gamma} defined a delta-matroid, called a \emph{$\Gamma$-graphic delta-matroid}, from a graph equipped with a labeling of vertices to an abelian group $\Gamma$.
In the definition, they employ a constraint similar to but different from non-zero.
Exploring the relationship between $\Gamma$-graphic delta-matroids and our findings is left for future work.

\paragraph*{Organization}
The rest of this paper is organized as follows.
\Cref{sec:preliminaries} provides preliminaries on groups and matroids. \Cref{sec:non-zero} deals with non-zero problems. \Cref{sec:algebraic} provides an algebraic algorithm for \textsc{$F$-avoiding Common Basis} for represented matroids. \Cref{sec:fixedF} studies \textsc{$F$-avoiding Basis} if $|F|$ is fixed. \Cref{sec:hard} includes our hardness results for each of the problems.
Finally, \cref{sec:conclusion} concludes this paper enumerating open problems.

\section{Preliminaries}\label{sec:preliminaries}
Let $\N$, $\Zp$, $\Z$, $\Q$, $\R_{\ge 0}$, and $\R$ denote the set of positive integers, nonnegative integers, integers, rationals, nonnegative reals, and reals, respectively.
We let $[n] \coloneqq \set{1, \dotsc, n}$ for $n \in \Zp$.
For a set $S$, we simply write $S \setminus \set{x}$ as $S - x$ for $x \in S$ and $S \cup \set{y}$ as $S + y$ for $y \notin S$.
For a set $E$ and $r \in \Z_{\ge 0}$, we let $\binom{E}{r} \coloneqq \Set{S \subseteq E}{|S| = r}$.

In this paper, all groups are implicitly assumed to be abelian.
We use the additive notation for the operations of groups except in \cref{sec:algebraic}.
For $m \in \N$, let $\Z_m = \set{0, \dotsc, m-1}$ denote the cyclic group of order $m$.
For groups $\Gamma_1$ and $\Gamma_2$, we mean by $\Gamma_1 \le \Gamma_2$ that $\Gamma_1$ is a subgroup of $\Gamma_2$.
A group $\Gamma$ is said to be \emph{ordered} if $\Gamma$ is equipped with a total order $\le$ compatible with the operation of $\Gamma$ in the sense that $a \le b$ implies $a+c \le b+c$ for all $a,b,c \in \Gamma$.
A \emph{labeling} is a function $\psi\colon E \to \Gamma$ from a set $E$ to a group $\Gamma$, and we let $\psi(S) \coloneqq \sum_{x \in S} \psi(x)$ for $S \subseteq E$.
Let $\GF(q)$ denote the finite field of $q$ elements for a prime power $q$.

We follow~\cite{diestel2017graph} for basic terminologies on graphs such as paths and cycles.
The vertex and edge sets of a graph $G$ are denoted by $V(G)$ and $E(G)$, respectively.  Similarly, $V(D)$ denotes the vertex set of a directed graph $D$ and $A(D)$ denotes its arc set.
Given a weight function $w\colon A(D)\to \R$ and a subgraph $C$ of $D$, the \emph{weight} of $C$ is defined as $w(C)\coloneqq w(A(C))$.  A weight function $w$ is said to be \emph{conservative} if $D$ does not contain a directed cycle of negative weight. 

We refer the reader to \cite{frank2011connections,schrijver2003combinatorial} for basics on matroid optimization. A \emph{matroid} $M$ consists of a finite ground set $E(M)$ and a nonempty set family $\cB(M) \subseteq 2^{E(M)}$ such that for any $B_1, B_2 \in \cB(M)$ and $x \in B_1 \setminus B_2$, there exists $y \in B_2 \setminus B_1$ such that $B_1 - x + y \in \cB(M)$.
Elements in $\cB(M)$ are called \emph{bases}.
The next basis exchange property was proved by Brualdi~\cite{brualdi1969comments}, see also \cite[Theorem~39.12]{schrijver2003combinatorial}.

\begin{lemma}[Brualdi~\cite{brualdi1969comments}] \label{lem:bijection}
    If $B$ and $B'$ are bases of a matroid $M$, then there exists a bijection $\phi\colon B \setminus B' \to B' \setminus B$ such that $B-e+\phi(e)$ is a basis for each $e \in B \setminus B'$.
\end{lemma}

Following \cite{frank2011connections}, we define a \emph{partition matroid} as a direct sum of uniform matroids, and a \emph{unitary partition matroid} as a direct sum of rank-1 uniform matroids. We note that several authors refer to the latter object as partition matroids. Given a matrix $A$ over some field, we denote by $M(A)$ the matroid defined on the column indices of $A$ where a set is a basis of $M(A)$ if the corresponding columns form a basis of the vector space spanned by the columns of $A$. Given a connected graph $G$, its \emph{cycle matroid} $M(G)$ is the matroid whose ground set is $E(G)$ and whose bases are the spanning trees of $G$. 
If $M=M(A)$ for a matrix $A$ over a field $\mathbb{F}$ or a graph $A$, we say that $M$ is \emph{$\mathbb{F}$-representable} or \emph{graphic}, respectively.

\section{Non-zero Basis and Common Basis}\label{sec:non-zero}
This section deals with the problem of finding non-zero bases and non-zero common bases. In \cref{subsec:non-zero-basis}, we give some rather simple algorithmic results on \textsc{Non-Zero Basis}. In \cref{subsec:non-zero-common-basis}, we give a collection of results on the more complex algorithmic landscape of \textsc{Non-Zero Common Basis}.

\subsection{Non-zero Basis}\label{subsec:non-zero-basis}

In this section, we consider (\textsc{Weighted}) \textsc{Non-Zero Basis}. The following theorem can be derived from the description of the lattices of matroid bases by Rieder~\cite{rieder1991lattices}. In what follows we give a direct proof of the result.

\begin{theorem} \label{thm:components}
    Let $M$ be a matroid and $\psi\colon E(M)\to \Gamma$ a group labeling. The following are equivalent:
    \begin{enumerate}[label={\textup{(\roman*)}}]\itemsep0em
        \item all bases of $M$ have the same label, \label{it:same}
        \item $M$ has a basis $B$ such that $\psi(B')=\psi(B)$ holds for each basis $B'$ with $|B\setminus B'|\le 1$, and\label{it:same1}
        \item $\psi$ is constant on each component of $M$. \label{it:constant}
    \end{enumerate}
\end{theorem}
\begin{proof}
    It is clear that \ref{it:same} implies \ref{it:same1} and \ref{it:constant} implies \ref{it:same}. In what follows we show that \ref{it:same1} implies \ref{it:constant}.
    Let $B$ be a basis such that $\psi(B')=\psi(B)$ holds for each basis $B'$ with $|B\setminus B'| \le 1$. 
    Let $G_B$ denote the bipartite graph with bipartition $(B, E(M) \setminus B)$ and edge set $\Set{xy}{x \in B, y \in E(M) \setminus B, B-x+y\in \mathcal{B}(M)}$. By the assumption on $B$, it follows that $\psi(x)=\psi(y)$ for each edge $xy$ of $G_B$. Then, $\psi$ is constant on each connected component of the graph $G_B$, and thus \ref{it:constant} follows by using that the connected components of the graph $G_B$ coincide with the components of the matroid $M$~\cite{Krogdahl1977dependence}.
\end{proof}

Note that \cref{thm:components}\ref{it:constant} provides a characterization for ``NO'' instances of \textsc{Non-Zero Basis}, while \cref{thm:components}\ref{it:same1} provides a simple algorithm for this problem. Liu and Xu~\cite{liu2023congruency} gave the following extension of this algorithm for \textsc{Weighted Zero Basis} with $\Gamma = \Z_2$, for which the zero and non-zero constraints are equivalent. First, find a minimum weight basis $B \in \cB(M)$, and if $\psi(B) \neq 0$, then we are done. 
Otherwise, consider all the bases of the form $B-x+y$ with $x \in B$ and $y \in E(M) \setminus B$. 
Among these bases, if there is none with a non-zero label, then there does not exist a non-zero basis, otherwise, choose a non-zero basis of minimum weight among the considered ones.
Actually, this algorithm works correctly for \textsc{Weighted Non-Zero Basis} for any group and the proof of~\cite[Proposition~1]{liu2023congruency} can be modified to show its correctness. In what follows, we give a different proof of this fact.

\begin{lemma}\label{lem:weighted-non-zero-base}
    Let $M$ be a matroid, $\psi\colon E(M) \to \Gamma$ a group labeling, and $w\colon E(M) \to \R$ a weight function.
    Then, for any minimum-weight basis $B$, there exists a minimum-weight non-zero basis $B^*$ such that $|B \setminus B^*| \leq 1$, provided that at least one non-zero basis exists.
\end{lemma}
\begin{proof}
    Let $B'$ be a minimum-weight non-zero basis with $|B \setminus B'|$ being minimal.
    If $B=B'$ then we are done, otherwise $\psi(B)=0$.
    According to the symmetric exchange axiom, we can choose $x \in B \setminus B'$ and $y \in B' \setminus B$ such that $B-x+y$ and $B'+x-y$ are both bases.
    As $0 \neq \psi(B)+\psi(B')=\psi(B-x+y)+\psi(B'+x-y)$, one of $B-x+y$ and $B'+x-y$ must be non-zero.
    Suppose $\psi(B-x+y) \ne 0$.
    Since $w(B)+w(B')=w(B-x+y)+w(B'+x-y)$ and $w(B)$ has the minimum weight, we have $w(B-x+y) \le w(B')$, which implies $w(B-x+y) = w(B')$ by $\psi(B-x+y) \ne 0$.
    Thus, we can take $B^* = B-x+y$.
    If $\psi(B'+x-y) \ne 0$, then it can be shown in the same way that $B'+x-y$ is a minimum-weight non-zero basis, contradicting the assumption that $B'$ is a minimum-weight non-zero basis closest to $B$.
\end{proof}

We obtain the following from \cref{lem:weighted-non-zero-base}.

\begin{theorem}\label{thm:weighted-non-zero-basis}
   \textscup{Weighted Non-Zero Basis} can be solved in polynomial time.
\end{theorem}

\subsection{Non-zero Common Basis} \label{subsec:non-zero-common-basis}
This section is dedicated to \textsc{Non-Zero Common Basis}. For the algorithmic tractability, it turns out that it is decisive whether $\Gamma$ contains $\Z_2$ or not. In \cref{sec:algz2}, for the case that $\Gamma$ does not contain $\Z_2$, we give a polynomial-time algorithm for \textsc{Non-Zero Common Basis} and a polynomial-time approximation algorithm for the weighted version of the problem. While the hardness proof for the case that $\Gamma$ contains $\Z_2$ is postponed to \cref{sec:hard}, we give a polynomial-time algorithm solving the weighted version for arbitrary groups in case that both matroids are partition matroids or one of the two matroids is a partition matroid with two partition classes in \cref{sec:partmat}. We further give a characterization of the no-instances if $\Gamma$ does not contain $\Z_2$ in \cref{sec:cert} and we deal with a connection of \textsc{Non-Zero Common Basis} and reconfiguration problems in \cref{sec:reconfig}.
\subsubsection{Polynomial-time Algorithm with \texorpdfstring{$\Z_2 \not\le \Gamma$}{Z2 </= Γ}}\label{sec:algz2}
In this section, we show the polynomial solvability of \textsc{Non-Zero Common Basis} when $\Z_2 \not \le \Gamma$, that is, $\Gamma$ does not contain any element of order two. Later, we will show in \cref{thm:non-zero-common-basis-hard} that the problem is hard if $\Z_2 \le \Gamma$. Our algorithm is a modification of the weighted matroid intersection algorithm given by Krogdahl~\cite{krogdahl1974combinatorial, krogdahl1976combinatorial, Krogdahl1977dependence} and independently by Brezovec, Cornu\'{e}jols, and Glover~\cite{brezovec1986two}.

We will use the following result on directed graphs. While several works concentrated on finding non-zero paths and cycles in group-labeled graphs~\cite{iwata2022finding}, their setting does not seem to include group-labeled digraphs. Therefore, we give a proof of the next result in \cref{sec:appendix}.  While this result may be of independent interest, it will later be applied as a subroutine.

\begin{restatable}{theorem}{nzdircycle} \label{thm:graphalg}
    Let $D$ be a digraph, $\psi\colon A(D)\rightarrow \Gamma$ a group labeling, and $w\colon A(D)\rightarrow \R$ a conservative  weight function. Then, there is a polynomial-time algorithm that returns a non-zero directed cycle in $D$ which is shortest with respect to $w$ or correctly reports that $D$ contains no non-zero directed cycle. 
\end{restatable}

Let $M_1$ and $M_2$ be matroids on a common ground set $E$ and $\psi\colon E \to \Gamma$ a group labeling.
Given a common basis $B$, we define the digraph $D_{M_1, M_2}(B)$ on vertex set $E$ and the labeling $\psi'$ on its arcs as follows. 
For each $x \in B$ and $y \in E\setminus B$ such that $B-x+y \in \cB(M_1)$, we add an arc $xy$ to $D_{M_1, M_2}(B)$ with label $\psi'(xy)\coloneqq \psi(y)$. 
Similarly, for each $x \in B$ and $y \in E\setminus B$ such that $B-x+y\in \cB(M_2)$, we add an arc $yx$ and with label $\psi'(yx) \coloneqq -\psi(x)$. 

\begin{lemma} \label{lem:non-zero-cycle}
    Let $M_1$ and $M_2$ be matroids on a common ground set $E$ and $\psi\colon E \to \Gamma$ a group labeling. 
    Let $B$ and $B'$ be common bases of  $M_1$ and $M_2$ such that $\psi(B)=0$ and $\psi(B') \ne 0$.
    Then, $D_{M_1, M_2}(B)$ contains a non-zero directed cycle $C$ with $V(C) \subseteq B \symdif B'$.
\end{lemma}
\begin{proof}
    By \cref{lem:bijection}, $D_{M_1, M_2}(B)$ contains a collection $P_1$ of $|B\setminus B'|$ pairwise vertex-disjoint arcs from $B\setminus B'$ to $B'\setminus B$ and a collection $P_2$ of $|B\setminus B'|$ pairwise vertex-disjoint arcs from $B'\setminus B$ to $B\setminus B'$. The union of $P_1$ and $P_2$ has label $\psi(B'\setminus B) - \psi(B\setminus B') = \psi(B')-\psi(B) = \psi(B) \ne 0$ and consists of pairwise vertex-disjoint directed cycles in $D_{M_1, M_2}(B)$, hence it contains a non-zero directed cycle. 
\end{proof}

We will rely on the following two lemmas by Krogdahl~\cite{krogdahl1974combinatorial, krogdahl1976combinatorial, Krogdahl1977dependence} and Brezovec, Cornu\'{e}jols, and Glover~\cite{brezovec1986two}. For the proof of \cref{lem:unique}, see also \cite[Theorem~39.13]{schrijver2003combinatorial}.

\begin{lemma}[Krogdahl~\cite{krogdahl1974combinatorial, krogdahl1976combinatorial, Krogdahl1977dependence}] \label{lem:unique}
    Let $B$ be a basis of a matroid $M$ and consider $X \subseteq B$ and $X' \subseteq E(M) \setminus B$ with $|X| = |X'|$. If there exists a unique bijection $\phi\colon X \to X'$ such that $B-e+\phi(e)$ is a basis for each $e \in X$, then $(B \setminus X) \cup X'$ is a basis.
\end{lemma}

\begin{lemma}[Brezovec--Cornu\'{e}jols--Glover~\cite{brezovec1986two}] \label{lem:2cycles} 
    Let $C$ and $C'$ be two different directed cycles in a digraph $D$ with $V(C)=V(C')$. Then, there exist directed cycles $C_1,\dots, C_t$ for some $t \ge 3$ such that $V(C_i) \subsetneq V(C)$ for $i \in [t]$ and each arc of $D$ appears in $C_1, \dots, C_t$ in total as many times as it appears in $C$ and $C'$ in total.
\end{lemma}

The following result and proof are analogous to that of \cite[Theorem~2]{brezovec1986two}. In that result, a weight function is given instead of a labeling, and the constraint ``non-zero'' is replaced by ``negative''.

\begin{restatable}{lemma}{goodcycle} \label{lem:goodcycle} 
    Let $\Gamma$ be a group such that $\Z_2 \not \le \Gamma$.
    Let $M_1$ and $M_2$ be matroids on a common ground set $E$, $\psi\colon E \to \Gamma$ a group labeling, and $B$ a common basis.
    If $C$ is a non-zero directed cycle of $D_{M_1, M_2}(B)$ whose vertex set is inclusion-wise minimal among non-zero directed cycles, then $B\symdif V(C)$ is a common basis.
\end{restatable}

\begin{proof}
    Let $C$ be a non-zero directed cycle such that there exists no non-zero directed cycle $C'$ with $V(C') \subsetneq V(C)$.
    Let $X \coloneqq V(C) \cap B$ and $X' \coloneqq V(C) \setminus B$.
    Assume for contradiction that $(B \setminus X) \cup X'$ is not a common basis. We may assume by symmetry that $(B \setminus X) \cup X' \not \in \cB(M_1)$. 
    Let $P_1$ and $P_2$ denote the sets of arcs of $C$ directed from $B$ to $E\setminus B$ and from $E\setminus B$ to $B$, respectively.
    Note that $P_1$ and $P_2$ are matchings between $X$ and $X'$ directed from $X$ to $X'$ and from $X'$ to $X$, respectively.    
    As $(B \setminus X) \cup X' \not \in \cB(M_1)$, by \cref{lem:unique} there exists a matching $P'_1$ between $X$ and $X'$ directed from $X$ to $X'$ in $D_{M_1, M_2}(B)$ which is different from $P_1$.
    Let $C'$ be defined by $A(C')=A(C)-P_1+P'_1$.
    Observe that $\psi'(C')=\psi'(C)$ and $C'$ is either the union of at least two pairwise vertex-disjoint directed cycles or a single directed cycle.
    In the first case, $C'$ is the union of vertex-disjoint directed cycles $C_1,\dots, C_t$ such that $V(C_i) \subsetneq V(C)$ for $i \in [t]$,  thus $0 \ne \psi(C') = \psi'(C_1)+\dots + \psi'(C_t)$ implies that $\psi'(C_i) \ne 0$ for some $i \in [t]$.
    In the second case, let $C_1,\dots, C_t$ be the directed cycles provided by \cref{lem:2cycles}, then $0 \ne 2 \psi'(C) = \psi'(C_1)+\dots + \psi'(C_t)$ implies that $\psi'(C_i) \ne 0$ for some $i \in [t]$, using that $\Z_2 \not \le \Gamma$. In both cases, $\psi'(C_i) \ne 0$ holds for some directed cycle $C_i$ with $V(C_i) \subsetneq V(C)$, contradicting the choice of $C$.
\end{proof}

Combining \cref{lem:non-zero-cycle,lem:goodcycle}, we obtain the main result of the section.

\begin{theorem}\label{thm:non-zero-common-basis-poly}
    Let $\Gamma$ be a group such that $\Z_2 \not \le \Gamma$. Let $M_1$ and $M_2$ be matroids on a common ground set $E$, $\psi\colon E \to \Gamma$ a group labeling, and $B_0$ a zero common basis. Then, $M_1$ and $M_2$ have a non-zero common basis if and only if $D_{M_1, M_2}(B_0)$ contains a non-zero directed cycle.
    Moreover, \textscup{Non-Zero Common Basis} is polynomially solvable.
\end{theorem}
\begin{proof}
    If there exists a non-zero common basis $B^*$, then $D_{M_1, M_2}(B_0)$ contains a non-zero directed cycle by \cref{lem:non-zero-cycle}. Conversely, if $D_{M_1, M_2}(B_0)$ contains a directed cycle, then let $C$ be a minimum length non-zero directed cycle. Then, \cref{lem:goodcycle} implies that $B^* \coloneqq B \symdif V(C)$ is a common basis, and we have $\psi(B^*)=\psi(B_0)+\psi'(C) = \psi'(C) \ne 0$.

    This provides the following algorithm for \textsc{Non-Zero Common Basis}. First, we find a common basis $B_0$. If no common basis exists or $B_0$ is non-zero, we are done. Otherwise, we find a minimum length non-zero directed cycle $C$ in $D_{M_1, M_2}(B_0)$ using \cref{thm:graphalg}.
    If no non-zero directed cycle exists then we report that there is no non-zero common basis, otherwise we output $B_0 \symdif V(C)$.
\end{proof}

We turn to the study of \textsc{Weighted Non-Zero Common Basis}. Given two matroids $M_1$ and $M_2$ on a common ground set $E$, a common basis $B$, and a weight function $w\colon E \to \mathbb{R}$, we define the weight function $w'$ on the arcs of $D_{M_1, M_2}(B)$ as follows. For each arc $xy$ such that  $x \in B$, $y \in E\setminus B$ and $B-x+y \in \cB(M_1)$ we define $w'(xy) \coloneqq w(y)$, and for each arc $yx$ such that $x \in B$, $y \in E \setminus B$ and $B-x+y \in \cB(M_2)$ we define $w'(yx) \coloneqq -w(x)$.
Then, $B$ is a minimum-weight common basis if and only if $w'$ is conservative~\cite{krogdahl1976combinatorial, fujishige1977primal, brezovec1986two}, see also~\cite[Theorem~41.5]{schrijver2003combinatorial}.
We observe the following relationship between the weight of a shortest non-zero directed cycle in $D_{M_1, M_2}(B)$ and the weights of non-zero common bases of $M_1$ and $M_2$. 

\begin{restatable}{lemma}{lembla}
 \label{lem:weighted-non-zero-common-base}
    Let $M_1$ and $M_2$ be matroids on a common ground set $E$, $\psi\colon E \to \Gamma$ a group labeling, and $w\colon E \to \R$ a weight function. Let $B_0$ be a minimum-weight common basis and assume that $\psi(B_0)=0$. Let $C$ be a shortest non-zero directed cycle in $D_{M_1, M_2}(B_0)$ with respect to $w'$. Then, $w(B_0\symdif V(C)) \le w(B^*)$ holds for each non-zero common basis $B^*$. 
\end{restatable}

\begin{proof}
    Let $B^*$ a non-zero common basis. As in the proof of \cref{lem:non-zero-cycle}, \cref{lem:bijection} implies that $D_{M_1, M_2}(B_0)$ contains vertex-disjoint directed cycles $C_1,\dots, C_q$ such that $V(C_1)\cup \dots \cup V(C_q) = B_0 \symdif B^*$ and one of the directed cycles is non-zero, say, $\psi(C_1) \ne 0$. Using the definition of $w'$ and the facts that $C$ is a shortest non-zero directed cycle and $w'$ is conservative, we obtain
    \[w(B_0\symdif V(C)) = w(B_0) + w'(C) \le w(B_0) + w'(C_1) \le w(B_0) + \sum_{i=1}^q w'(C_i) = w(B^*). \qedhere\]
\end{proof}

We note that \cref{lem:weighted-non-zero-common-base} generalizes \cref{lem:weighted-non-zero-base}, as in the special case $M_1=M_2$ each arc of $D_{M_1, M_2}(B_0)$ is bidirectional, thus a shortest non-zero directed cycle consists of two vertices.

In \cref{lem:weighted-non-zero-common-base}, the weight of $C$ is measured by $w'$ (which takes negative values on some arcs) so $V(C)$ is not necessarily inclusion-wise minimal among the vertex sets of non-zero directed cycles. Thus, it does not yield an algorithm for \textsc{Weighted Non-Zero Common Basis}. In fact, the complexity of the problem remains open for a group $\Gamma$ with $\Z_2 \not\le \Gamma$. Nevertheless, we use \cref{lem:weighted-non-zero-common-base} to give a 2-approximation if the weight function $w$ is nonnegative. 

\begin{restatable}{theorem}{approx} \label{thm:2approx}
    Let $\Gamma$ be a group with $\Z_2 \not \le \Gamma$.
    Let $M_1$ and $M_2$ be matroids on a common ground set $E$, $\psi\colon E \to \Gamma$ a group labeling, and $w\colon E \to \R_{\geq 0}$ a weight function. Then, there exists a polynomial-time algorithm that computes a non-zero common basis $B$ that satisfies $w(B)\leq 2 w(B^*)$ for every non-zero common basis $B^*$ or correctly outputs that there exists no non-zero common basis.
\end{restatable}

\begin{proof}
    First, we compute a minimum-weight common basis $B_0$. This is possible due to the algorithm of Edmonds~\cite{EDMONDS197939}, see also~\cite[Section 13.2.2]{frank2011connections} for details. If no common basis exists or $B_0$ is non-zero, we are done. Otherwise, using \cref{thm:graphalg} we compute a shortest non-zero directed cycle $C$ in $D_{M_1, M_2}(B_0)$ with weight function $w'$. If no non-zero directed cycle exists then we obtain by \cref{thm:non-zero-common-basis-poly} that no non-zero common basis exists. Otherwise, using \cref{thm:graphalg} we compute a shortest non-zero directed cycle $C'$ in $D_{M_1, M_2}(B_0)[V(C)]$ with unit weights. We let the algorithm return $B\coloneqq B_0\symdif V(C')$, which is a non-zero common basis by \cref{lem:goodcycle}.  Using $B \subseteq B_0 \cup V(C)$, $w \ge 0$, \cref{lem:weighted-non-zero-common-base} and the fact that $B_0$ is a minimum-weight common basis, we get that
    \[w(B) \le w(B_0 \cup V(C)) \le w(B_0) + w(B_0 \symdif V(C)) \le 2 w(B^*)\]
    holds for each non-zero common basis $B^*$.
\end{proof}

\subsubsection{Certificate for All Common Bases Being Zero}\label{sec:cert}

In this section we give a characterization for all all directed cycles of $D_{M_1, M_2}(B)$ having label zero, analogous to the weight-splitting theorem of Frank~\cite{frank1981weighted}. By \cref{thm:non-zero-common-basis-poly}, this provides a certificate for the no-instances of \textsc{Non-Zero Common Basis} if $\Z_2 \not \le \Gamma$. 

For the characterization, we will use the following result on strongly connected digraphs.
Its special case for cyclic groups appears in \cite[Theorem~10.5.1]{BangJensen2009digraphs} where the authors note that its origin is unclear. The proof given there shows that the result holds for any abelian group $\Gamma$ as well.

\begin{theorem} \label{thm:period} Let $D$ be a strongly connected digraph and $\psi\colon A(D) \to \Gamma$ a group labeling. Then, all directed cycles of $D$ are zero if and only if there exists a labeling $\varphi\colon V(D) \to \Gamma$ such that $\psi(uv) = \varphi(v)-\varphi(u)$ for each arc $uv \in A(D)$.  
\end{theorem}

To apply \cref{thm:period} to $D_{M_1,M_2}(B)$, we need to ensure its strong connectivity, which is described in the next statement.

\begin{lemma} \label{lem:strongly-connected} 
    Let $M_1$ and $M_2$ be matroids with the common ground set $E$ and the same rank $r$, and let $B$ be a common basis of $M_1$ and $M_2$.
    Then, $D_{M_1, M_2}(B)$ is strongly connected if and only if $r_{M_1}(X) + r_{M_2}(E\setminus X) > r$ holds for $\emptyset \ne X \subsetneq E$.
\end{lemma}
\begin{proof}
    Since $M_1$ and $M_2$ have a common basis, $r_{M_1}(X)+r_{M_2}(E\setminus X) \ge r$ holds for $X \subseteq E$. We show that a set $X\subseteq E$ has in-degree zero in $D_{M_1, M_2}(B)$ if and only if $r_{M_1}(X) + r_{M_2}(E\setminus X) = r$. By definition, there is no arc from $B \setminus X$ to $X \setminus B$ if and only if $B-x+y$ is not a basis of $M_1$ for any $x \in B \setminus X$ and $y \in X \setminus B$. This is equivalent to that $C_{M_1}(B,y) \subseteq (B \cap X) + y$ holds for every $y \in X \setminus B$, which happens if and only if $B \cap X$ spans $X$ in $M_1$, that is, $r_{M_1}(X) = |B \cap X|$. Similarly, there is no arc from $E \setminus (B \cap X)$ to $B \cap X$ if and only if $r_{M_2}(E\setminus X) = |B \setminus X|$. Therefore, no arc enters $X$ if and only if $r_{M_1}(X) = |B \cap X|$ and $r_{M_2}(E\setminus X) = |B \setminus X|$ both hold, which is equivalent to $r_{M_1}(X) + r_{M_2}(E \setminus X) = r$. 
\end{proof}

We note that if $r_{M_1}(X) + r_{M_2}(E \setminus X) = r$ holds for a set $\emptyset \ne X \subsetneq E$, then the common bases of $M_1$ and $M_2$ are exactly the sets that can be obtained as the union of a common basis of $M_1\mathbin{|}X$ and $M_2 \mathbin{/} (E \setminus X)$ and a common basis of $M_1 \mathbin{/} (E \setminus X)$ and $M_2\mathbin{|}X$. Therefore, $M_1$ and $M_2$ have two common bases of different labels if and only if one of the smaller matroid pairs $M_1 \mathbin{|} X$, $M_2 \mathbin{/} (E \setminus X)$ and $M_1 \mathbin{/} (E \setminus X)$, $M_2\mathbin{|}X$ has two common bases of different labels. This implies that when solving \textsc{Non-Zero Common Basis}, we may assume that no such $X$ exists. Note that if such a set exists then we can find it in polynomial time, e.g., by minimizing the the submodular function $f(Z)\coloneqq r_{M_1}(Z) + r_{M_2}(E \setminus Z)$ over the sets $\emptyset \ne Z \ne E$ which can be performed in polynomial time~\cite{cunningham1984testing}.

We are ready to give the main result of this section. Our proof is similar to that of the weight-splitting theorem of Frank~\cite{frank1981weighted} given in~\cite[Section~41.3a]{schrijver2003combinatorial}.

\begin{restatable}{theorem}{splitting} \label{thm:splitting}
    Let $M_1$ and $M_2$ be matroids with the common ground set $E$ and the same rank $r$. Assume that $r_{M_1}(X) + r_{M_2}(E \setminus X) > r$ holds for every $\emptyset \ne X \subsetneq E$. 
    Let $B$ be a common basis of  $M_1$ and $M_2$, and $\psi\colon E \to \Gamma$ a group labeling. Then, $D_{M_1, M_2}(B)$ contains no non-zero directed cycle if and only if there exist labelings $\psi_1, \psi_2 \colon E \to \Gamma$ such that $\psi = \psi_1 + \psi_2$ and $\psi_i$ is constant on each connected component of $M_i$ for $i=1,2$.
\end{restatable}

\begin{proof}

By the assumption and \cref{lem:strongly-connected}, we get that $D_{M_1, M_2}(B)$ is strongly connected. 

First suppose that there exist labelings $\psi_1, \psi_2 \colon E \to \Gamma$ such that $\psi = \psi_1 + \psi_2$ and $\psi_i$ is constant on each connected component of $M_i$ for $i=1,2$, and let $C=x_1,y_1,x_2,\ldots, x_t,y_t,x_1$ be a directed cycle in $D_{M_1, M_2}(B)$. Let $X\coloneqq\set{x_1,\ldots,x_t}$ and $Y\coloneqq\set{y_1,\ldots,y_t}$. By symmetry and construction, we may suppose that $X \subseteq B$ and $Y \subseteq E \setminus B$. By construction, for $i \in [t]$, we have that $x_i$ and $y_i$ are in the same component of $M_1$, hence $\psi_1(x_i)=\psi_1(y_i)$. Similarly, we have $\psi_2(y_i)=\psi_2(x_{i+1})$ for $i \in [t-1]$ and $\psi_2(y_t)=\psi_2(x_1)$.
This yields 
\begin{align*}
    \psi'(C)&=\psi(Y)-\psi(X)\\
    &=\psi_1(Y)-\psi_1(X)+\psi_2(Y)-\psi_2(X)\\
    &=\sum_{i=1}^t(\psi_1(y_i)-\psi_1(x_i))+\sum_{i=1}^{t-1}(\psi_2(y_i)-\psi_2(x_{i+1}))+(\psi_2(y_t)-\psi_2(x_1))\\
    &=0.
\end{align*}
It follows that $D_{M_1, M_2}(B)$ does not contain any non-zero directed cycle.

Now suppose that $D_{M_1, M_2}(B)$ does not contain any non-zero directed cycle. Then, by \cref{thm:period}, there exists a labeling $\varphi\colon E \to \Gamma$ such that $\psi'(xy) = \varphi(y)-\varphi(x)$ holds for each arc $xy$ of $D_{M_1, M_2}(B)$.
Let $\psi_1, \psi_2 \colon E \to \Gamma$ be defined by 
\[\psi_1(e) = \begin{cases} -\varphi(e) & \text{if $e \in B$}, \\ -\varphi(e)+\psi(e) & \text{if $e \in E \setminus B$}, \end{cases} \quad \psi_2(e) = \begin{cases} \varphi(e) + \psi(e) & \text{if $e \in B$}, \\ \varphi(e) & \text{if $e \in E \setminus B$}. \end{cases} \]
Then, $\psi = \psi_1 + \psi_2$ holds by construction. Now let $x \in B$ and $y \in E \setminus B$ such that $B-x+y$ is a basis of $M_1$. We obtain 
\begin{align*}
    \psi_1(B-x+y)-\psi_1(B)&=\psi_1(y)-\psi_1(x)\\
    &=\psi(y)-\varphi(y)+\varphi(x)\\&=
    \psi'(xy)-\psi'(xy)\\&=0.
\end{align*}
It follows by \cref{thm:components} that $\psi_1$ is constant on every component of $M_1$. A similar argument shows that $\psi_2$ is constant on every component of $M_2$. 
\end{proof}

\subsubsection{Partition Matroids}\label{sec:partmat}

When both matroids are partition matroids, we can drop the assumption $\Z_2 \not\le \Gamma$ from \cref{thm:non-zero-common-basis-poly} and extend it to the weighted setting. 

\begin{restatable}{theorem}{nonzeromatching} \label{thm:non-zero-matching}
    \textscup{Weighted Non-Zero Common Basis} is polynomially solvable if $M_1$ and $M_2$ are partition matroids. 
\end{restatable}

\begin{proof}
    The proof relies on the following property of partition matroids.
    
    \begin{claim} \label{cl:partition}
        Let $B$ be a common basis and $C$ a directed cycle of $D_{M_1, M_2}(B)$. Then, $B\symdif V(C)$ is a common basis.
    \end{claim}
    \begin{claimproof}
        For $i=1,2$, let $M_i$ be the direct sum of uniform matroids on $E^i_1\cup \dots \cup E^i_{t_i}$ having rank $g^i_1,\dots, g^i_{t_i}$, respectively.
        Then, $\cB(M_i) = \Set{B \subseteq E}{|B \cap E^i_j| = 1 \text{ for each $j \in [t_i]$}}$, thus $B-x+y$ is a basis of $M_i$ for $x \in B$, $y \in E \setminus B$ if and only if $\set{x,y}\subseteq E^i_j$ for some $j \in [t_i]$. By the definition of $D_{M_1,M_2}(B)$, this implies that $|(V(C)\cap B)\cap E^i_j| = |(V(C)\setminus B)\cap E^i_j|$ for each $j \in [t_i]$ and $i \in [2]$, thus $B\symdif V(C)$ is a common basis. 
    \end{claimproof}

    We turn to the proof of the theorem.
    As in the proof of \cref{thm:2approx}, we may assume to have a minimum weight common basis $B_0$ with $\psi(B_0)=0$, and a shortest non-zero directed cycle $C$ in $D_{M_1, M_2}(B_0)$. Then, $B_0\symdif V(C)$ is a minimum weight non-zero common basis by \cref{cl:partition} and \cref{lem:weighted-non-zero-common-base}.
\end{proof}

Given a graph $G$ and a function $b\colon V(G) \to \Zp$, a \emph{perfect $b$-matching} is an edge set $F \subseteq E(G)$ such that $d_F(v) = b(v)$ for each $v\in V$. If $G$ is bipartite, then its perfect $b$-matchings form the family of common bases of two partition matroids. Therefore, \cref{thm:non-zero-matching} can be formulated as having a polynomial-time algorithm for finding a minimum weight non-zero perfect $b$-matching in a bipartite graph with weights and labels on its edges. For perfect matchings and the group $\Z_2$, the idea of essentially the same algorithm as ours was briefly mentioned in \cite{jia2023exact}, where the authors noted that it can also be derived from results of \cite{Artmann2017}. A formal description of the algorithm and a proof of its 
correctness was given in \cite{elmaalouly2023exact} for a special weight function.

Next, we consider the special case of \textsc{Weighted Non-Zero Common Basis} when we only assume that one of the matroids is a partition matroid. Without further assumptions, this problem is not easier than the general one: a construction similar to that of Harvey, Kir\'aly and Lau~\cite{harvey2011disjoint} shows that the general problem can be reduced to a special case when one of the matroids is a unitary partition matroid and all partition classes have size two. In what follows, we will consider the case when the partition matroid is defined by a partition having two classes. In this special case, the polynomial solvability of \textsc{Non-Zero Common Basis} follows from the corresponding lattice basis characterization of Rieder~\cite{rieder1991lattices}. We extend this result by solving the weighted version of the problem.

\begin{restatable}{theorem}{nonzeroexactbase} \label{thm:non-zero-exact-base}
    \textscup{Weighted Non-Zero Common Basis} is polynomially solvable if $M_2$ is a partition matroid defined by a partition having two classes.
\end{restatable}

In order to prove \cref{thm:non-zero-exact-base} we will use the following lemma to study the directed cycles of $D_{M_1, M_2}(B)$.

\begin{lemma} \label{lem:cyclepartition}
    Let $m, k \in \N$, $A$ and $A'$ disjoint $m$-element sets, $V\coloneqq A \cup A'$,  $c\colon V \to [k]$ a function, and $V_i \coloneqq \Set{v \in V}{c(v)=i}$ for $i \in [k]$. Consider a digraph $D$ on vertex set $V$ such that its arc set consists of $m$ vertex-disjoint arcs from $A$ to $A'$ and the arc set $\Set{a'a}{i,j\in[m], c(a)=c(a')}$.
    Then, there exist vertex-disjoint directed cycles $C_1, \dots, C_t$ in $D$ such that $V(C_1)\cup \dots \cup V(C_t) = V$ and $|V(C_i)| \le 2k$ for $i \in [t]$. 
\end{lemma}
\begin{proof}
    We use induction on $m$. Observe that each vertex of $D$ has outdegree at least one. Indeed, the outdegree of each vertex $a \in A$ is one due to the arcs from $A$ to $A'$, and the outdegee of each $a' \in A'$ is at least one since $|A\cap V_{c(a')}| = |A' \cap V_{c(a')}| \ge 1$.
    This implies that $D$ contains a directed cycle, let $C$ be one of minimum length. It is not difficult to see that $C$ being of minimum length implies that $c(a'_1) \ne c(a'_2)$ for all distinct $a'_1,a'_2 \in A$. Then, $|V(C)| = 2|V(C) \cap A'| \le 2k$. Since $|V(C)\cap A \cap V_i| = |V(C) \cap A' \cap V_i|$ for each $i \in [k]$, it follows by induction that $D-V(C)$ contains a collection of vertex-disjoint directed cycles $C_1, \dots, C_t$ each of which is of length at most $2k$ such that $V(C_1)\cup \dots \cup V(C_t)=V\setminus V(C)$.  The statement then follows by adding $C$ to the collection $C_1, \dots, C_t$.
\end{proof}

As a consequence, we derive the following property of $D_{M_1, M_2}(B)$ when $M_2$ is a partition matroid and $B$ is a minimum-weight common basis. In the following, the length of a directed cycle is its number of arcs while ``shortest'' refers to the weight function.

\begin{lemma} \label{lem:shortcycle}
    Let $k \ge 1$, $M_1$ be a matroid and $M_2$ a partition matroid on ground set $E$ such that $M_2$ is defined by a partition having $k$ classes. Let $\psi\colon E \to \Gamma$ be a group labeling, $w\colon E \to \R$ a weight function and $B$ a minimum-weight common basis of $M_1$ and $M_2$. Then, the shortest non-zero directed cycle in $D_{M_1, M_2}(B)$ has length at most $2k$, provided that it has a non-zero directed cycle.
\end{lemma}
\begin{proof}
    Let $E_1,\dots, E_k$ denote the partition defining $M_2$.
    Let $C$ be a shortest non-zero directed cycle in $D_{M_1, M_2}(B)$ with respect to $\psi'$ and $w'$. 
    Using the definition of $D_{M_1, M_2}(B)$, it follows that $|(V(C) \cap B)\cap E_i| = |(V(C)\setminus B) \cap E_i|$ for each $i \in [k]$.  Applying \cref{lem:cyclepartition} to $A \coloneqq V(C) \cap B$, $A'\coloneqq V(C) \setminus B$, the digraph induced by $A \cup A'$ in $D_{M_1, M_2}(B)$, and the function $c\colon V(C) \to [k]$ satisfying $E_i \cap V(C) = \Set{v \in V(C)}{c(v) = i}$ for each $i \in [k]$, we get that $D_{M_1, M_2}(B)$ contains vertex-disjoint directed cycles $C_1, \dots, C_t$ such that $V(C_1)\cup \dots \cup V(C_t) = V(C)$ and $|V(C_i)| \le 2k$ for each $i \in [t]$. As $B$ is a minimum-weight basis, $w'$ is conservative, this $w(C_i) \le w(C)$ for each $i \in [t]$.
    Since $C$ is non-zero and $\psi(C)=\psi(C_1)+\dots + \psi(C_t)$, it follows that $\psi(C_i) \ne 0$ for some $i \in [t]$. As $C$ is a shortest non-zero directed cycle, we get $i=t=1$, and the statement follows by $|V(C_1)| \le 2k$.
\end{proof}

We will use \cref{lem:shortcycle} for $k=2$ combined with the following observation.

\begin{lemma} \label{lem:4cycle}
    Let $M_1$ and $M_2$ be matroids on a common ground set $E$, $\psi\colon E \to \Gamma$ a group labeling, $w\colon E\to \R$ a weight function, and $B$ a minimum-weight common basis.
    Suppose that $C$ is a shortest non-zero directed cycle in $D_{M_1, M_2}(B)$ and $|V(C)|\le 4$. Then, $B\symdif V(C)$ is a common basis. 
\end{lemma}
\begin{proof}
    Using the definition of $D_{M_1, M_2}(B)$ we have that $|V(C)| \in \set{2,4}$ and $B\symdif V(C)$ is a common basis for $|V(C)| = 2$. It remains to consider the case $|V(C)|=4$. Let $b_1b'_1$, $b'_1b_2$, $b_2b'_2$, $b'_2b_1$ denote the arcs of $C$ such that $b_1,b_2 \in B$ and $b'_1, b'_2 \notin B$. 
    Suppose that $B \symdif V(C)$ is not a common basis. By symmetry, we may assume that it is not a basis of $M_1$. Then, \cref{lem:unique} implies that $B-b_1+b'_2$ and $B-b_2+b'_1$ are bases of $M_1$, that is, $b_1b'_2$ and $b_2b'_1$ are arcs of $D_{M_1, M_2}(B)$. Let $C_1$ and $C_2$ denote the directed cycles with arc sets $\set{b_1b'_2, b'_2b_1}$ and $\set{b_2b'_1, b'_1b_2}$, respectively. Then, $C_1$ and $C_2$ are directed cycles in $D_{M_1, M_2}(B)$ such that $w(C_i) \le w(C)$ for $i=1,2$ and $\psi(C_1)+\psi(C_2) = \psi(C) \ne 0$, thus $\psi(C_i)\ne 0$ for some $i \in [2]$. This contradicts $C$ being a shortest non-zero directed cycle.
\end{proof}

We are ready to prove \cref{thm:non-zero-exact-base}.

\begin{proof}[Proof of \cref{thm:non-zero-exact-base}]
    As in the proof of \cref{thm:2approx}, we may assume to have a minimum weight common basis $B_0$ with $\psi(B_0)=0$ and a shortest non-zero directed cycle $C$ in $D_{M_1, M_2}(B_0)$. Then, $B_0\symdif V(C)$ is a common basis by \cref{lem:shortcycle,lem:4cycle}, thus it is a minimum weight non-zero common basis by \cref{lem:weighted-non-zero-common-base}.
\end{proof}

\subsubsection{Relation to Reconfiguration Problems}\label{sec:reconfig}

We discuss the relation of \textsc{Non-Zero Common Basis} to a reconfiguration problem on matroids. 
For common bases $B$ and $B'$ of  matroids $M_1$ and $M_2$, we define a \emph{reconfiguration sequence} from $B$ to $B'$ as a sequence of common bases $B_0,B_1,\dots, B_\ell$ such that $B_0=B$, $B_\ell = B'$ and $|B_{i-1}\setminus B_i| = |B_i \setminus B_{i-1}| = 1$ for each $i \in [\ell]$.
In general, a reconfiguration sequence might not exist, in fact, it is hard to decide if a reconfiguration sequence exists between given bases $B$ and $B'$~\cite{kobayashi2023rcb}. Nevertheless, it is known for some matroid pairs that a reconfiguration sequence exists and can be found in polynomial time between any pair of common bases.  Such matroid pairs include a regular matroid and its dual~\cite{berczi2023reconfiguration}, the family of arborescences (or $r$-arborescences for a fixed root $r$) in a digraph~\cite{ito2023reconfiguring}, or more generally, arc sets that can be partitioned into $k$ arborescences (or $r$-arborescences for a fixed root $r$) for a fixed $k \ge 1$~\cite{kobayashi2023rcb}. Note that these arc sets can be represented as common bases of two matroids, see~\cite[Corollary 53.1c]{schrijver1990spanning}. 
We show that \textsc{Non-Zero Common Basis} is solvable for such matroid pairs using an idea similar to the solution of the exact arborescence problem for $(0,1)$-valued weight functions by Barahona and Pulleyblank~\cite{barahona1987exact}.

\begin{theorem} 
    Let $M_1$ and $M_2$ be matroids such that a reconfiguration sequence exists and can be found in polynomial time between any pair of common bases. Then, \textscup{Non-Zero Common Basis} is polynomially solvable for $M_1$ and $M_2$.
\end{theorem}
\begin{proof}
We may assume that $M_1$ and $M_2$ have a common basis.
Let $g_1, \dots, g_m$ be an arbitrary ordering of the set $\Set{\psi(e)}{e \in E}$. Let us define $s\colon E \to \Z^m$ such that $s(e)=(s_1(e),\dots, s_m(e))$ for each $e \in E$ where for each $i \in [m]$, $s_i(e)=1$ if $\psi(e)=g_i$ and $s_i(e)=0$ otherwise. 
Compute a minimum and maximum $s$-weight common basis $B_-$ and $B_+$, respectively, with respect to the lexicographic ordering on $\Z^m$. If $s(B_-)=s(B_+)$, then the $s$-weight of each common basis is the same, implying that the label of each common basis is the same. Thus, in this case we can output $B_-$ if $\psi(B_-) \ne 0$, otherwise we correctly report that no non-zero common basis exists. It remains to consider the case $s(B_-) \ne s(B_+)$. By the assumption of the theorem, we can find a sequence of common bases $B_0,\dots, B_\ell$ such that $B_0 = B_-$, $B_\ell = B_+$ and $|B_i \setminus B_{i-1}| = 1$ for each $i \in [\ell]$. Since $s(B_0)\ne s(B_\ell)$, there exists $i \in [\ell]$ such that $s(B_{i-1}) \ne s(B_i)$. As $|B_i \setminus B_{i-1}|=1$, this implies that $\psi(B_i) - \psi(B_{i-1}) = g_j-g_{j'}$ for some distinct $j,j' \in [m]$, in particular, $\psi(B_i) \ne \psi(B_{i-1})$. Therefore, at least one of $B_{i-1}$ and $B_i$ provides a solution to \textsc{Non-Zero Common Basis}.
\end{proof}

\section{Algebraic Algorithm for \texorpdfstring{$F$}{F}-avoiding Basis and Common Basis} \label{sec:algebraic}

\newcommand{\F}{\mathbb{F}}

We present a randomized polynomial-time algorithm for \textsc{$F$-avoiding Common Basis} for representable matroids given as matrices over a field $\F$ and a finite group $\Gamma$ given as an operation table.
Our algorithm is a generalization of the exact-weight matroid intersection algorithm for representable matroids by Camerini, Galtiati, and Maffioli~\cite{camerini1992exact}.
A similar algebraic algorithm has also been considered by Webb~\cite[Section~3.7]{webb2004paffian}.

Before describing the algorithm, we introduce needed algebraic notions and results.
We assume that the arithmetic operations and the zero test over $\F$ can be performed in constant time.
In this section, we use the multiplicative notation for the operation of $\Gamma$, and let $e$ denote the group unit (zero) of $\Gamma$ instead of $0$.
The \emph{group ring} $\F[\Gamma]$ of $\Gamma$ over $\F$ is the set of formal $\F$-coefficient linear combinations of the elements in $\Gamma$, i.e., $\F[\Gamma] \coloneqq \Set{\sum_{g \in \Gamma} a_g g}{ a_g \in \F \; (g \in \Gamma)}$.
The addition and multiplication of $f = \sum_{g \in \Gamma} a_g g \in \F[\Gamma]$ and $h = \sum_{g \in \Gamma} b_g g \in \F[\Gamma]$ are naturally defined as $f+h = \sum_{g \in \Gamma} (a_g + b_g) g$ and $fh = \sum_{g,g' \in \Gamma} a_g b_{g'} gg'$.
With these operations, $\F[\Gamma]$ forms a commutative ring, containing $\F$ as a subring under the natural identification $\F \ni a \mapsto ae \in \F[\Gamma]$.
Note that the operations of $\F[\Gamma]$ and the zero test can be performed in polynomially many operations of $\F$ and $\Gamma$.

Our algorithm will require the field $\F$ to be sufficiently large.
If $\F$ is a small finite field $\GF(q)$, where $q$ is a prime power, we use the following result from algebraic computation to extend $\F$ to $\GF(q^d)$ by finding an irreducible polynomial $p \in \F[x]$ of degree $d$ as $\GF(q^d) \simeq \F[x]/(p)$.

\begin{theorem}[{see~\cite[Corollarly~4.6, Theorem~14.42]{vondergathen2013modern}}]\label{thm:field-extension}
    Let $\F = \GF(q)$ for some  prime power $q$ and $d \in \N$.
    There is a randomized algorithm to find an irreducible polynomial of degree $d$ over $\F$ in time polynomial in $\log q$ and $d$ in expectation.
    After this preprocessing, we can represent every element in $\GF(q^d)$ as a $d$-tuple of elements in $\F$ and emulate the arithmetic operations and the zero test over $\GF(q^d)$ in time polynomial in $d$.
\end{theorem}

For finite sets $R$ and $C$, we mean by an \emph{$R \times C$ matrix} a matrix of size $|R| \times |C|$ whose rows and columns are identified with $R$ and $C$, respectively.
We simply write $[r] \times C$ as $r \times C$ for $r \in \Zp$.
Given a ground set $E$ and a labeling $\psi\colon E \to \Gamma$, we define an $E \times E$ diagonal matrix $D_{\psi}$ as follows: for every $j \in E$, we set the $(j,j)$ diagonal entry of $D_{\psi}$ as $x_j \psi(j)$, where $x_j$ is an indeterminate (variable) whose actual value comes from $\F$.
Then, $D_{\psi}$ is regarded as a matrix over the group ring $\F(\set{x_e}_{e \in E})[\Gamma]$, where $\F(\set{x_e}_{e \in E})$ denotes the rational function field over $\F$ in $|E|$ indeterminates $\set{x_e}_{e \in E}$.

The following is a modification of a claim of Tomizawa and Iri~\cite{tomizawa1974axb}, who first used the Cauchy--Binet formula in the context of linear matroid intersection.

\begin{lemma}\label{lem:cauchy-binet}
    Let $\F$ be a field, $M_1$ and $M_2$ $\F$-representable matroids with the common ground set $E$ and the same rank $r$, $A_k$ an $r \times E$ matrix representing $M_k$ for $k = 1,2$, and $\psi\colon E \to \Gamma$ a labeling.
    Let $\Xi = A_1 D_{\psi} A_2^\top$.
    Then, the coefficient of $g \in \Gamma$ in $\det(\Xi)$ is a non-zero polynomial in $\set{x_j}_{j \in E}$ if and only if a common basis with label $g$ exists.
\end{lemma}
\begin{proof}
    By the Cauchy--Binet formula, we can expand $\det (\Xi)$ as
    \begin{align}\label{eq:det-xi}
        \det(\Xi)
        = \sum_{B \in \binom{E}{r}} \det(A_1[B]) \det(A_2[B]) \prod_{j \in B} x_j \cdot \psi(B),
    \end{align}
    where $A_k[B]$ denotes the submatrix of $A_k$ obtained by extracting the columns in $B$ for $k=1,2$.
    Observe that $\det(A_1[B]) \det(A_2[B])$ is non-zero if and only if $B$ is a common basis of $M_1$ and $M_2$, and the terms coming from different common bases do not cancel out thanks to the factor $\prod_{j \in B} x_j$, proving the claim.
\end{proof}

\Cref{lem:cauchy-binet} together with the Schwartz--Zippel lemma~\cite{lovasz1979rand,schwartz1980pit,zippel1979pit}, division-free determinant algorithm~\cite{kaltofen2005determinant}, search-to-decision reduction, and the field extension for small fields give rise to a randomized algebraic algorithm for \textsc{$F$-avoiding Common Basis}. 

\begin{restatable}{theorem}{representablezerocommonbasis} 
\label{thm:representable-zero-common-basis}
    Let $\F$ be a field and $M_1$ and $M_2$ $\F$-representable matroids with the common ground set $E$.
    There is a randomized algorithm that, given matrices $A_1$ and $A_2$ over $\F$ representing $M_1$ and $M_2$, respectively, the operation table of a finite abelian group $\Gamma$, a group labeling $\psi\colon E \to \Gamma$, and a forbidden label set $F \subseteq \Gamma$, solves \textscup{$F$-avoiding Common Basis} in expected polynomial time.
\end{restatable}

\begin{proof}
    We may assume that $M_1$ and $M_2$ have the same rank $r$ and the matrices $A_1$ and $A_2$ are of size $r \times E$.
    If $\F$ is a finite field $\GF(q)$ and $q \le 4nr$ with $n = |E|$, we first extend $\F$ to $\GF(q^d)$ so that it satisfies $|\GF(q^d)| = q^d \ge 4nr$ by applying \cref{thm:field-extension} with $d = \lceil \log_q 2nr \rceil$.
    In the following, we assume that $|\F| \ge 2nr$ holds.
    
    Let $\Xi = A_1 D_{\psi} A_2^\top$ be the matrix in \cref{lem:cauchy-binet} and $H \subseteq \F$ an arbitrary finite set with $|H| \ge 2nr$.
    We first build an algorithm for the decision version as follows: for every $j \in E$, draw $a_j \in H$ from the uniform distribution over $H$, substitute $a_j$ for $x_j$ in $\Xi$, compute $\det(\Xi) \in \F[\Gamma]$, and report that an $F$-avoiding common basis exists if the coefficient of $g$ in $\det(\Xi)$ is non-zero for some $g \in \Gamma \setminus F$; otherwise, the algorithm concludes that it does not.
    The complexity of this algorithm is polynomial in $n$ and $|\Gamma|$ because the determinant over $\F[\Gamma]$ can be calculated with polynomially many operations of $\F[\Gamma]$ by using a division-free determinant algorithm, e.g.,~\cite{kaltofen2005determinant}.
    From \cref{lem:cauchy-binet}, this algorithm reports correctly if there is no $F$-avoiding common basis, and otherwise the probability of success can be lower-bounded by $1 - \frac{r}{|H|} \ge 1 - \frac{1}{2n}$ via the Schwartz--Zippel lemma since every non-vanishing term in $\det(\Xi)$ is of total degree $r$ as a polynomial in $\set{x_j}_{j \in E}$.

    To construct an actual $F$-avoiding common basis, we use the so-called search-to-decision reduction based on the self-reducibility as follows.
    Fix an arbitrary ordering $j_1, \dotsc, j_n$ of $E$.
    From $k = 1$ to $n$, we use the above decision algorithm to test whether removing $j_k$ from the ground set retains the existence of an $F$-avoiding common basis.
    If so, we actually remove $j_k$ from $E$; otherwise, we do nothing.
    The removal of an element is represented by deleting corresponding columns from $A_1$ and $A_2$.
    The final ground set $E$ is an $F$-avoiding common basis if all invocations of the decision algorithm have reported correctly.
    By the union bound, the probability of success is at least $1 - n \cdot \frac{1}{2n} = \frac12$.
    We thus check whether the final $E$ is an $F$-avoiding common basis or not and repeat the algorithm until we get an $F$-avoiding common basis.
    Since the probability of failing the first $k$ attempts is at most $\frac{1}{2^k}$ for any $k \in \N$, the expected running time is upper-bounded by $\sum_{k=1}^\infty \frac{k}{2^k}T = 2T$, where $T$ is the running time of each attempt.
    Thus, the expected running time is polynomially bounded.
\end{proof}

A \emph{Pfaffian pair} is a pair of $r \times n$ matrices $A_1, A_2$ such that $\det(A_1[B]) \det(A_2[B])$ is a non-zero constant for any common basis $B$~\cite{webb2004paffian}.
This property implies that, if $\F = \Q$ and matroids are given as a Pfaffian pair, then no cancel-out occurs in the summands of $\det(\Xi)$ in the equation~\eqref{eq:det-xi} even if we substitute $1$ for all $x_i$.
Therefore, we can derandomize the algorithm given in \cref{thm:representable-zero-common-basis}. 
Examples of common bases of matroid pairs representable by Pfaffian pairs include spanning trees, regular matroid bases, arborescences, perfect matchings in Pfaffian-orientable bipartite graphs, and node-disjoint $S$--$T$ paths in planar graphs~\cite{matoya2022pfaffian}.

\begin{theorem} \label{thm:pfaffian}
    \textscup{$F$-avoiding Common Basis} is polynomially solvable for $\Q$-representable matroids if matroids are given as a Pfaffian pair and a group is given as the operation table.
\end{theorem}

We can also generalize \cref{thm:representable-zero-common-basis} to a randomized pseudo-polynomial-time algorithm for the weighted problem as follows.

\begin{restatable}{theorem}{representableweightedzerocommonbasis}\label{thm:representable-weighted-zero-common-basis}
    Let $\F$ be a field and $M_1$ and $M_2$ $\F$-representable matroids on a common ground set $E$.
    There is a randomized algorithm that, given matrices $A_1$ and $A_2$ over $\F$ representing $M_1$ and $M_2$, respectively, the operation table of a finite abelian group $\Gamma$, a group labeling $\psi\colon E \to \Gamma$, a forbidden label set $F \subseteq \Gamma$, and a weight function $w\colon E \to \Z$, solves \textscup{Weighted $F$-avoiding Common Basis} in pseudo-polynomial time in expectation.
    If $\F = \Q$ and $(A_1, A_2)$ is a Pfaffian pair, the algorithm can be derandomized.
\end{restatable}

\begin{proof}
    We modify each $(j,j)$ diagonal entry of $D_{\psi}$ as $x_j\psi(j)\theta^{w(j)}$, where $\theta$ is an additional indeterminate.
    Then, by the same argument as \cref{lem:cauchy-binet}, the degree of the coefficient of $e$ in $\det (\Xi)$ with respect to $\theta$ is the maximum weight of an $F$-avoiding common basis; see also~\cite[Proposition~2.5]{matoya2022pfaffian}.
    Therefore, via random substitution, division-free determinant computation, and search-to-decision reduction, we obtain a randomized pseudo-polynomial-time algorithm, as given in the proof of \cref{thm:representable-zero-common-basis}.
\end{proof}

\section{\texorpdfstring{$F$}{F}-avoiding Basis with Fixed \texorpdfstring{$|F|$}{|F|}} \label{sec:fixedF}

As we will see in \cref{thm:zero-base-table}, \textsc{Zero Basis} is hard for groups given by operation tables. This implies the hardness of \textsc{$F$-avoiding Basis} if the set $F$ of forbidden labels is part of the input. 
In this section, we study the problem when $F$ has fixed size.
Note that in contrast to the setting of \cref{sec:algebraic}, we assume that $\Gamma$ is given as an operation oracle (and it is not necessarily finite).

Related to the notion of $k$-closeness recently introduced by Liu and Xu~\cite{liu2023congruency}, we propose the following conjecture.

\begin{conjecture} \label{conj:k}
    Let $M$ be a matroid, $\psi\colon E(M) \to \Gamma$ a labeling, and $F \subseteq \Gamma$ a finite collection of forbidden labels. Then, for any basis $B$, there exists an $F$-avoiding basis $B^*$ with $|B \setminus B^*| \le |F|$, provided that at least one $F$-avoiding basis exists.
\end{conjecture}

Note that \cref{thm:components} implies that \cref{conj:k} holds for $|F|=1$.

If \cref{conj:k} is true, then it is tight in the sense that for each $k\ge 1$ there exists a group-labeled matroid and a set $F$ of $k$ forbidden labels such that there exists an $F$-avoiding basis but $|B \setminus B^*| \ge k$ holds for some basis $B$ and any $F$-avoiding basis $B^*$. Indeed, consider the rank-$k$ uniform matroid on ground set $E=\set{e_1, \dots, e_{2k}}$, the labeling $\psi\colon E \to \Z_{k+1}$ given by $\psi(e_i) = 0$ for $i \in [k]$ and $\psi(e_i) = 1$ for $i \in [2k] \setminus [k]$, and forbidden label set $F = \Z_{k+1} \setminus \set{0}$. Then, the only $F$-avoiding basis is $B^*=\set{e_1, \dots, e_k}$ and we have $|B \setminus B^*|$ for the basis $B=\set{e_{k+1}, \dots, e_{2k}}$. 

We can relax \cref{conj:k} as follows.

\begin{conjecture} \label{conj:fk}
    There exists a function $f\colon \N \to \N$ with the following property: If $M$ is a matroid, $\psi\colon E(M) \to \Gamma$ is a group labeling, and $F \subseteq \Gamma$ is a finite collection of forbidden labels, then, for any basis $B$, there exists an $F$-avoiding basis $B^*$ with $|B \setminus B^*| \le f(|F|)$, provided that at least one $F$-avoiding basis exists.
\end{conjecture}
Conjectures~\ref{conj:k} and \ref{conj:fk} have algorithmic implications due  to the following simple observation.

\begin{lemma}\label{lem:cons}
    Let $\alpha$ be a fixed positive integer. Further, let $M$ be a matroid, $\psi\colon E(M) \to \Gamma$ a group labeling, and $F \subseteq \Gamma$ a finite collection of forbidden labels, such that, for any basis $B$, there exists an $F$-avoiding basis $B^*$ with $|B \setminus B^*| \le \alpha$, provided that at least one $F$-avoiding basis exists. Then, an $F$-avoiding basis of $M$ can be found in polynomial time, if one exists.
\end{lemma}
\begin{proof}
    We first compute an arbitrary basis $B$ of $M$. Then, for every $X \subseteq B$ and every $Y \subseteq E(M)\setminus B$ with $|X|=|Y|\leq \alpha$, we test whether $(B\setminus X)\cup Y$ is an $F$-avoiding basis of $M$. As there are at most $\alpha n^{\alpha}$ choices for each of $X$ and $Y$, the desired running time follows. If we find an $F$-avoiding basis during this procedure, we return it. Otherwise, no $F$-avoiding basis exists by assumption.
\end{proof}
The following is an immediate consequence of \cref{lem:cons}.
\begin{corollary} \label{cor:kalg} 
    If \cref{conj:fk} holds, then \textscup{$F$-avoiding Basis} is solvable in  polynomial time if $|F|$ is fixed. 
\end{corollary}

Liu and Xu~\cite{liu2023congruency} defined a finite group $\Gamma$ to be \emph{$k$-close} for an integer $k \ge 1$, if for any matroid $M$, group labeling $\psi \colon E(M) \to \Gamma$, element $g \in \Gamma$ and basis $B$, there exists a basis $B^*$ with $|B \setminus B^*| \le k$ and $\psi(B^*)=g$, provided that $M$ has at least one basis with label $g$. Observe that \cref{conj:k} would imply $(|\Gamma|-1)$-closeness, and \cref{conj:fk} would imply $f(|\Gamma|-1)$-closeness of each finite group $\Gamma$ for some function $f\colon \N\to \N$. This would imply an FPT algorithm for \textsc{Zero Basis} when parameterized with $|\Gamma|$ due to the following result which is a consequence of \cite[Theorem~1]{liu2023congruency}.

\begin{theorem}[see Liu--Xu~\cite{liu2023congruency}] \label{thm:liuxu} 
Assume that for each finite group $\Gamma$, there exists an integer $k$ such that $\Gamma$ is $k$-close. Then, \textscup{Zero Basis} is FPT for finite groups when parameterized by $|\Gamma|$.
\end{theorem}

Liu and Xu~\cite{liu2023congruency} observed that if all subgroups of $\Gamma$ satisfy a conjecture by Schrijver and Seymour~\cite{schrijver1990spanning}, then $\Gamma$ is $(|\Gamma|-1)$-close. By the results of DeVos, Goddyn, and Mohar~\cite{devos2009generalization}, this implies that any cyclic group $\Gamma$ is $(|\Gamma|-1)$-close whose order is a prime power or the product of two primes.
The proof of \cite[Theorem~4]{liu2023congruency} does not seem to generalize to our setting, thus, it is not clear whether the conjecture of Schrijver and Seymour implies \cref{conj:k}.
If $\Gamma$ has prime order, then Liu and Xu~\cite[Theorem~3]{liu2023congruency} gave a simpler proof of $(|\Gamma|-1)$-closeness. That proof also generalizes to show that \cref{conj:k} holds for such groups. 

Using results of Lemos~\cite{lemos2006weight}, we can also prove that \cref{conj:k} holds for ordered groups. The result is restated in the following theorem.

 \begin{restatable}{theorem}{ordered} \label{thm:ordered}
    Let $M$ be a matroid, $\psi\colon E(M) \to \Gamma$ a labeling to an ordered group $\Gamma$, $F \subseteq \Gamma$ a finite collection of forbidden labels, $B$ a basis of $M$, and suppose that $M$ has an $F$-avoiding basis. Then, there exists an $F$-avoiding basis $B^*$ of $M$ with $|B^* \setminus B|\leq |F|.$
\end{restatable}

To prove the theorem, we will rely on the following results due to Lemos~\cite{lemos2006weight}.

\begin{theorem}[Lemos~\cite{lemos2006weight}] \label{thm:lemos}
    Let $M$ be a matroid and $\psi\colon E(M) \to \Gamma$ a labeling, where $\Gamma$ is an ordered group.
    \begin{enumerate}[label={(\alph*)}] 
        \item \label{it:decrease} If $B$ and $B'$ are bases such that $\psi(B') < \psi(B)$, then there exist elements $e \in B\setminus B'$ and $f \in B'\setminus B$ such that $B-e+f$ is a basis with $\psi(B-e+f) < \psi(B)$.
        \item \label{it:kano} Let $g_0,\ldots,g_m$ be the unique ordering of $\Set{\psi(B)}{ B \text{ is a basis of $M$}}$ such that $g_0 < \dots < g_m$, and let $B_0$ be a basis with $\psi(B_0) = g_0$. Then, for any $k \in \set{0,\dots, m}$, there exists a basis $B_k$ such that $\psi(B_k) = g_k$ and $|B_0 \setminus B_k| \le k$. 
    \end{enumerate}
\end{theorem}
Note that while Lemos only considered the case $\Gamma = \R$, the same proofs work for any ordered group $\Gamma$. We remark that statement~\ref{it:decrease} follows from \cref{lem:bijection}. Motivated by results of Kawamoto, Kajitani, and Shinoda~\cite{kawamoto1978second}, Kano~\cite{kano1987maximum} stated four conjectures on spanning trees whose weight is $k$-th largest among the weights of spanning trees. The first of these conjectures was verified by Mayr and Plaxton~\cite{mayr1989spanning}, and statement~\ref{it:kano} is its extension to matroids by Lemos~\cite{lemos2006weight}. We show that \cref{thm:lemos} implies \cref{conj:k} for ordered groups.

We first prove the following preliminary result as a conclusion from \cref{thm:lemos}\ref{it:decrease}.

\begin{proposition}\label{prop:seq}
    Let $M$ be a matroid, $\psi\colon E(M) \to \Gamma$ a labeling to an ordered group $\Gamma$, and let $B$ be a basis of $M$. Then, there exists a sequence $(B_0,\ldots,B_\ell)$ of bases of $M$ such that $B_0=B$, $B_\ell$ is a minimum-label basis of $M$, and $\psi(B_i)<\psi(B_{i-1})$ and $|B_i \setminus B_{i-1}|=1 $ hold for all $i \in [\ell]$.
\end{proposition}
\begin{proof}
    Suppose for contradiction that the desired sequence does not exist and that $B$ is chosen so that $\psi(B)$ is minimum among all bases for which the desired sequence does not exist. Let $B^-$ be a minimum-label basis. By \cref{thm:lemos}\ref{it:decrease} applied to $B$ and $B^-$, there exists a basis $B'$ with $\psi(B')<\psi(B)$ and $|B \setminus B'|=1$. Hence, by the choice of $B$, we obtain that there exists a sequence $(B_0,\ldots,B_\ell)$ of bases such that $B_0=B'$, $B_\ell$ is a minimum-label basis, and $\psi(B_i)<\psi(B_{i-1})$ and $|B_i \setminus B_{i-1}|=1 $ hold for all $i \in [\ell]$. It follows that $(B,B_0,\ldots,B_\ell)$ is a sequence with the desired properties, a contradiction.
\end{proof}

We are now ready to prove that \cref{conj:k} holds when $\Gamma$ is an ordered group.

\begin{proof}[Proof of \cref{thm:ordered}]
Let $g_0,\ldots,g_m$ be the unique ordering of $\Set{\psi(B)}{B \text{ is a basis}}$ such that $g_0 < \dots < g_m$. 

\begin{claim} \label{cl:seq}
    There exists a sequence $(B_0,\ldots,B_\ell)$ of bases with the following properties: $g_m = \psi(B_\ell) < \psi(B_1) < \dots < \psi(B_\ell) = g_0$, $|B_i \setminus B_{i-1}|=1$ for all $i \in [\ell]$, and there exists some $i_0 \in \set{0,\ldots,\ell}$ with $B_{i_0}=B$.
\end{claim}
\begin{claimproof}
    By \cref{prop:seq}, there exists a sequence $(B'_0,\ldots,B'_{\ell'})$ of bases of $M$ such that $B'_0=B$, $\psi(B'_{\ell'})=g_0$, and $\psi(B'_i)<\psi(B'_{i-1})$ and $|B'_i \setminus B'_{i-1}|=1 $ hold for all $i \in [\ell']$. We now apply \cref{prop:seq} to $B$ and $-\psi$. We obtain that there exists a sequence $(B''_0,\ldots,B''_{\ell''})$ of bases such that $B''_0=B$, $\psi(B''_{\ell''})=g_m$, and $\psi(B''_i)>\psi(B''_{i-1})$ and $|B''_i \setminus B''_{i-1}|=1 $ hold for all $i \in [\ell'']$. We obtain that $(B''_{\ell''},\ldots,B''_1,B,B'_1,\ldots,B'_{\ell'})$ is a sequence with the desired properties.
\end{claimproof}

Let $(B_0,\dots, B_\ell)$ be the sequence and $i_0$ the index provided by \cref{cl:seq}.
We finish the proof through a distinction of two cases.

\begin{case0}
    There exists some $i_1 \in \set{0,\ldots,\ell}$ with $\psi(B_{i_1})\notin F$.
\end{case0}
    We may suppose that $i_1$ is chosen with this property so that $|i_1-i_0|$ is minimum. Further, by, symmetry, we may suppose that $i_1>i_0$. 
    As $\set{\psi(B_{i_0}),\psi(B_{i_0+1}),\ldots,\psi(B_{i_1-1})}\subseteq F$ and $\psi(B_j)\neq \psi(B_{j'})$ for all distinct $j,j' \in \set{i_0,\ldots,i_1-1}$, we obtain $i_1-i_0 \le |F|$.
    It follows that $|B\setminus B_{i_1}|\leq \sum_{i=i_0+1}^{i_1}|B_i\setminus B_{i-1}|= i_1-i_0 \leq |F|$. As $\psi(B_{i_1})\notin F$, the statement follows.
\begin{case0}
    $\psi(B_{i})\in F$ holds for all $i \in \set{0,\ldots,\ell}$.
\end{case0}

As $\psi(B_i)\neq \psi(B_j)$ for all distinct $i,j \in \set{0,\ldots,\ell}$, we have $\ell \leq |F|$. Let $p$ be the smallest integer such that $g_p \notin F$ and $q$ be the smallest integer such that $g_{m-q} \notin F$. Observe that $p$ and $q$ are well-defined as $M$ has an $F$-avoiding basis. As $\set{g_0,\ldots,g_{p-1}}\cup \set{g_{m-q+1},\ldots,g_m}\subseteq F$ and $\set{g_0,\ldots,g_{p-1}}$ and $\set{g_{m-q+1},\ldots,g_m}$ are disjoint, we obtain $p+q \leq |F|$. Next, by \cref{thm:lemos}\ref{it:kano}, there exists a basis $B^-$ with $\psi(B^-)=g_p$ and $|B_\ell \setminus B^-|\leq p$. Further, applying \cref{thm:lemos}\ref{it:kano} to $B_\ell$ and $-\psi$, we obtain that there exists a basis $B^+$ of $M$ with $\psi(B^+)=g_{m-q}$ and $|B_0 \setminus B^+|\leq q$. If $|B \setminus B^-|>|F|$ and $|B \setminus B^+|>|F|$ hold, we obtain
\begin{align*}
    2|F|&<|B \setminus B^-|+|B \setminus B^+|\\
    &\leq |B \setminus B_\ell|+|B_\ell \setminus B^-|+|B \setminus B_0|+|B_0 \setminus B^+|\\
    &\leq (\ell-i_0)+p+i_0+q\\
    &=p+q+\ell\\
    &\leq 2|F|,
\end{align*}
a contradiction. We hence obtain that one of $|B \setminus B^-|\leq |F|$ and $|B \setminus B^+|\leq |F|$ holds. As $\psi(B^-)\notin F$ and $\psi(B^+)\notin F$ hold, the statement follows.
\end{proof}

\subsection{Strongly Base Orderable Matroids and Relaxations}
In this section, we introduce a relaxed notion of strong base-orderability, called $(\alpha,k)$-weak base orderability, where $\alpha$ and $k$ are positive integers. In \cref{subsubsec:alphak}, we define this notion and show its relation to strong base-orderability and group-restricted bases. In \cref{subsubsec:reprweak} and \cref{sec:graph}, we conclude results for matroids representable over fixed finite fields and graphic matroids, respectively.
\subsubsection{\texorpdfstring{$(\alpha, k)$}{(α,k)}-Weakly Base Orderable Matroids}\label{subsubsec:alphak}

A matroid is called \emph{strongly base orderable} if for any two bases $B_1, B_2$, there exists a bijection $\varphi\colon B_1 \setminus B_2 \to B_2 \setminus B_1$ such that $(B_1\setminus Z) \cup \varphi(Z)$ is a basis for each $Z \subseteq B_1 \setminus B_2$.
For some positive integer $k$, we say that the ordered basis pair $(B_1,B_2)$ has the {\it $k$-exchange property} if there exist pairwise disjoint nonempty subsets $X_1,\dots, X_k \subseteq B_1 \setminus B_2$ and $Y_1, \dots, Y_k \subseteq B_2 \setminus B_1$ such that $\left(B_1 \setminus \bigcup_{i\in Z} X_i\right) \cup \bigcup_{i \in Z} Y_i$ is a basis for each $Z \subseteq [k]$. For positive integers $\alpha$ and $k$, we define a matroid $M$ to be \emph{weakly $(\alpha,k)$-base orderable} if the ordered pair $(B_1,B_2)$ has the $k$-exchange property for any two bases $B_1,B_2$ of $M$ with $|B_1\setminus B_2|\geq \alpha$.   
We note that $(\alpha, k)$-weakly base orderability is a relaxation of \emph{$k$-base orderability} defined by Bonin and Savitsky~\cite{bonin2016infinite}, and our definition of the $k$-exchange property differs from their definition of $k$-exchange-ordering. 
Observe that strongly base orderable matroids are precisely the matroids that are $(k,k)$-weakly base orderable for each $k \ge 1$.

For a matroid $M$ and two disjoint bases $B_1,B_2$ of $M$ with $B_1\cup B_2=E(M)$, we say that $(B_1,B_2)$ is a \emph{basis partition} of $M$. For a basis $B$ of a matroid $M$, we say that a minor $M'$ of $M$ is a \textit{$B$-minor} if it is obtained by contracting some elements of $B$ and deleting some elements of $E(M)\setminus B$. We use the following simple observation later.

\begin{restatable}{lemma}{redtominor} \label{lem:redtominor}
    Let $B_1$ and $B_2$ be two bases of a matroid $M$. Further, let $M'$ be a $B_1$-minor of $M$ such that $(B_1',B_2')$ is a basis partition of $M'$ and has the $k$-exchange property for some $k \in \N$, where $B_i'\coloneqq B_i \cap E(M')$ for $i=1,2$. Then, $(B_1,B_2)$ has the $k$-exchange property in $M$.
\end{restatable}

\begin{proof}
    As $(B_1',B_2')$ has the $k$-exchange property in $M'$, there exist pairwise disjoint nonempty subsets $X_1,\dots, X_k \subseteq B'_1$ and $Y_1, \dots, Y_k \subseteq B'_2$ such that $\left(B_1' \setminus \bigcup_{i\in Z} X_i\right) \cup \bigcup_{i \in Z} Y_i$ is a basis for each $Z \subseteq [k]$. For $Z \subseteq [k]$, let $B_Z\coloneqq \left(B_1 \setminus \bigcup_{i\in Z} X_i\right) \cup \bigcup_{i \in Z} Y_i$ and $B'_Z\coloneqq \left(B_1' \setminus \bigcup_{i\in Z} X_i\right) \cup \bigcup_{i \in Z} Y_i$. As $B_1'\subseteq B_1\setminus B_2$ and $B_2'\subseteq B_2\setminus B_1$, it suffices to prove that $B_Z$ is a basis of $M$ for every $Z \subseteq [k]$.
    This follows from $B'_Z$ being a basis of $M'$ using that $B_Z = B'_Z \cup (B_1 \setminus B'_1)$ and $M'$ is obtained from $M$ by contracting $B_1 \setminus B'_1$ and deleting a subset of $E(M)\setminus B_1$.
\end{proof}

The following result is our main motivation to consider weak base orderability. It establishes a connection between weak base orderability and \cref{conj:fk}.

\begin{theorem}\label{thm:weakly}
    Let $M$ be a matroid, $\psi\colon E(M) \to \Gamma$ a group labeling, and $F \subseteq \Gamma$ a finite collection of forbidden labels.
    If $M$ is $(\alpha, |F|+1)$-weakly base orderable, then for each basis $B$, there exists an $F$-avoiding basis $B^*$ with $|B \setminus B^*| \le \alpha-1$, provided that at least one $F$-avoiding basis exists.
\end{theorem}
For the proof, we need the following result which is most likely routine.

\begin{restatable}{proposition}{grouptrivi}\label{prop:grouptrivi}
    Let $S$ be a finite set, $\psi\colon S \to \Gamma$ a group labeling, and $0 \notin F \subseteq \Gamma$ satisfying $|F| \le |S|-1$. Then, there exists some nonempty $S' \subseteq S$ with $\psi(S')\notin F$.
\end{restatable}

\begin{proof}
    Let $k \coloneqq |S|-1$ and let $s_1,\ldots,s_{k+1}$ be an arbitrary ordering of $S$. If $\psi(\set{s_1,\ldots,s_i})=\psi(\set{s_1,\ldots,s_j})$ for some $1 \le i < j \le  k+1$, then $\psi(\set{s_{i+1},\ldots,s_j})=0 \notin F$. Otherwise, $\set{\psi(s_1),\psi(\set{s_1,s_2}),\ldots,\psi(S)}$ contains $k+1$ distinct values, hence at least one of them is not in $F$.
\end{proof}

\begin{proof}[Proof of \cref{thm:weakly}]
    Let $k \coloneqq |F|$ and let $B$ be a basis and $B'$ an $F$-avoiding basis minimizing $|B'\setminus B|$. 
    If $|B'\setminus B|\leq \alpha-1$, there is nothing to prove.
    We may hence suppose that $|B'\setminus B|\geq \alpha$. 
    Then, as $M$ is $(\alpha,k+1)$-weakly base orderable, there exist pairwise disjoint nonempty subsets $X_1, \dots, X_{k+1}\subseteq B'\setminus B$ and $Y_1, \dots, Y_{k+1} \subseteq B\setminus B'$ such that $\left(B'\setminus \bigcup_{i \in Z} X_i\right) \cup \bigcup_{i \in Z} Y_i$ is a basis for each $Z \subseteq [k+1]$.
    We define $\psi'\colon [k+1] \rightarrow \Gamma$ by $\psi'(i)=\psi(Y_i)-\psi(X_i)$ for all $i \in [k+1]$.
    Observe that $0 \notin F'\coloneqq \Set{f-\psi(B')}{f \in F}$, as $B'$ is an $F$-avoiding basis. 
    It hence follows from \cref{prop:grouptrivi} that there exists some nonempty $Z \subseteq [k+1]$ with $\psi'(Z) \notin F'$. Let $B''\coloneqq \left(B'\setminus \bigcup_{i \in Z} X_i\right)\cup\bigcup_{i \in Z}Y_i$. By the definition of $X_1, \dots, X_{k+1}$ and $Y_1,\dots, Y_{k+1}$, we obtain that $B''$ is a basis of $M$. Further, we have $\psi(B'')=\psi(B')+\psi'(Z)\notin F$. Finally, we have $|B''\setminus B|<|B'\setminus B|$ since $Z$ is nonempty. This contradicts the choice of $B'$.
\end{proof}

As strongly base orderable matroids are $(k, k)$-weakly base orderable for any $k \ge 1$, we also get the following.

\begin{corollary}\label{cor:sbo} Strongly base orderable matroids satisfy \cref{conj:k}.
\end{corollary}

\begin{remark}

    Note that the proofs of \cref{thm:weakly} and \cref{prop:grouptrivi} also show that \cref{conj:k} holds for a matroid $M$ if for any two bases $B, B'$ of $M$ there exist an ordering $b_1, \dots, b_r$ of $B$ and an ordering $b'_1,\dots, b'_r$ of $B'$ such that $\set{b_1,\dots, b_i, b'_{i+1},\dots, b'_j,b_{j+1}, \dots, b_r}$ is a basis of $M$ for any $0 \le i \le j \le r$. Baumgart~\cite{baumgart2009ranking} called such an ordering a \emph{subsequence-interchangeable base ordering} (SIBO for short), and conjectured that any basis pair of a graphic matroid has a SIBO. Baumgart showed that the conjecture holds for a pair of disjoint spanning trees of a wheel graph.
    We are not aware of any basis pair of a matroid not having a SIBO. However, the existence of an SIBO for any basis pair of a matroid would imply a long-standing conjecture of Gabow~\cite{gabow1976decomposing} on the serial symmetric exchange property of matroids.
\end{remark}

\subsubsection{Matroids Representable over Finite Fields} \label{subsubsec:reprweak}

In this section, we prove that the concept of weakly base orderability allows us to deal with a large class of matroids, namely all those which are representable over a fixed finite field. More precisely, we prove the following result.

\begin{theorem}\label{thm:reprweak}
    There is a function $f\colon \N\times \N \rightarrow \N$ such that for every prime power $q$, every $\GF(q)$-representable matroid is weakly $(f(q,k),k)$-orderable for any positive integer $k$.
\end{theorem}

On a high level, the proof works in the following way. First, relying on results of \cite{ding1996unavoidable} on the existence of certain submatrices of large matrices over finite fields, we show that every $\GF(q)$-representable matroid has a certain substructure. We then show that this substructure has the desired property. From this, we can conclude the theorem.

In order to find this substructure, we deal with the matrices representing the matroids in consideration. We first need some notation for these matrices. For two matrices $A$ and $A'$, we say that $A$ contains $A'$ as a {\it permuted submatrix} if $A'$ can be obtained from $A$ by deleting and permuting rows and columns. For a square matrix, we refer by its \emph{size} to its number of rows. Let $q$ be a prime power. We say that a triple $(\alpha,\beta,\gamma)$ of elements of $\GF(q)$ is {\it feasible} if $\alpha\neq \beta$ and at least one of $\beta \neq 0$ and $\gamma \neq 0$ hold. For a triple $(\alpha,\beta,\gamma)$ and a positive integer $t$, the \emph{$(\alpha,\beta,\gamma)$-diagonal matrix} of size $t$ is the $t \times t$ matrix $A = (A_{ij})$ such that 
\[A_{ij}=\begin{cases}
    \alpha & \text{if $i< j$,}\\
    \beta & \text{if $i= j$,}\\
    \gamma & \text{if $i > j$}
\end{cases}\]
for $i,j\in [t]$.
We now collect some properties of $(\alpha,\beta,\gamma)$-diagonal matrices.
We first need the following result due to Ding, Oporowski, Oxley, and Vertigan~\cite{ding1996unavoidable} showing that $(\alpha,\beta,\gamma)$-diagonal matrices can always be found in sufficiently large matrices over a fixed finite field. 

\begin{theorem}[Ding--Oporowski--Oxley--Vertigan {\cite[Theorem~2.3]{ding1996unavoidable}}] \label{thm:submatrix}
    There is a computable function $f_0\colon \N \times \N \to \N$ with the following property: Let $q$ be a prime power, $t$ a positive integer and $A$ a matrix over $\GF(q)$ having at least $f_0(q,t)$ columns no two of which are identical. Then, $A$ contains a permuted square submatrix $A'$ of size $t$ which is $(\alpha,\beta,\gamma)$-diagonal for a triple $(\alpha,\beta,\gamma)$ with $\alpha \neq \beta$.
\end{theorem}

We actually need a slight strengthening of \cref{thm:submatrix} which follows easily.

\begin{restatable}{proposition} {submatrixfeasible}\label{prop:submatrix2}
    There is a computable function $f_1\colon \N\times \N \to \N$ with the following property: Let $q$ be a prime power, $t$ a positive integer and $A$ a matrix over $\GF(q)$ having at least $f_1(q,t)$ columns no two of which are identical. Then, $A$ contains a permuted square submatrix $A'$ of size $t$ which is $(\alpha,\beta,\gamma)$-diagonal for a feasible triple $(\alpha,\beta,\gamma)$.
\end{restatable}

\begin{proof}
    Let $f_0$ be the function provided by \cref{thm:submatrix}. We prove the statement for the function $f_1$ defined by $f_1(q,t)\coloneqq f_0(q,t+1)$ for each $t \in \N$ and prime power $q$. Let $A$ be a matrix over $\GF(q)$ having at least $f_1(q,t)$ columns no two of which are identical. By \cref{thm:submatrix}, we obtain that $A$ has a permuted square submatrix $A'$ of size $t+1$ which is $(\alpha,\beta,\gamma)$-diagonal for some triple $(\alpha,\beta,\gamma)$ with $\alpha \neq \beta$. If $\beta \neq 0$ or $\gamma \neq 0$ holds, observe that $(\alpha,\beta,\gamma)$ is feasible. As the submatrix of $A'$ obtained by deleting the first row and the first column is an $(\alpha,\beta,\gamma)$-diagonal matrix of size $t$, the statement follows. We may hence suppose that $\beta=\gamma=0$. We obtain by assumption that $\alpha \neq 0$. Let $A''$ be obtained from $A'$ by deleting the first column and the last row and then reversing the order of the rows and the columns. Then, $A''$ is a permuted submatrix of $A'$ and hence $A$. Further, $A''$ is an $(0,\alpha,\alpha)$-diagonal matrix of size $t$. As $\alpha \neq 0$, we obtain that $(0,\alpha,\alpha)$ is feasible. Hence the statement follows.
\end{proof} 

We are now ready to give the following result showing that every sufficiently large matroid that is representable over a fixed finite field has a certain substructure. The approach is to choose a matrix representing the matroid and find a particular submatrix in this matroid using \cref{prop:submatrix2}. After, we show that a minor represented by this matrix can be obtained by applying certain deletions and contractions. 

\begin{restatable}{lemma}{minorexist}\label{lem:minorexist}
    There is a computable function $f_1\colon \N\times \N\to \N$ with the following properties: Let $q$ be a prime power, $k$ a positive integer, $M$ a $\GF(q)$-representable matroid of rank at least $f_1(q,k)$ and $(B_1,B_2)$ a basis partition of $M$. Then, there exists a $B_1$-minor $M'$ of $M$ that can be represented by a matrix of the form $[I_k ~ A]$, where $I_k$ is the identity matrix of size $k$ and $A$ is an $(\alpha,\beta,\gamma)$-diagonal matrix for a feasible triple $(\alpha,\beta,\gamma)$, and the columns of $I_k$ correspond to the elements of $B_1'$ and those of $A$ correspond to the elements of $B_2'$, where $B_i'\coloneqq B_i \cap E(M')$ for $i =1,2$.
\end{restatable}

\begin{proof}
    Let $f_1$ be the function provided by \cref{prop:submatrix2}. As $(B_1,B_2)$ is a basis partition of $M$ and $M$ is $\GF(q)$-representable, we obtain that $M$ can be represented by a matrix of the form $[I ~ A_0]$ over $\GF(q)$, where $I$ is an identity matrix of size at least $f_1(q,k)$, $A_0$ is a nonsingular matrix of the same size, and the columns of $I$ correspond to $B_1$ and the columns of $A_0$ correspond to $B_2$. As $A_0$ is nonsingular, in particular, it does not contain two identical columns. It now follows from \cref{prop:submatrix2} that there exists a permuted submatrix $A$ of $A_0$ which is an $(\alpha,\beta,\gamma)$-diagonal matrix of size $k$ for a feasible triple $(\alpha,\beta,\gamma)$. Let $X \subseteq B_1$ be the set of elements of $B_1$ whose corresponding columns only contain 0's in all rows which contribute to $A$ and let $Y\subseteq B_2$ be the set of elements of $B_2$ whose corresponding columns do not contribute to $A$. Let $M'\coloneqq M/X\backslash Y$. First observe that $M'$ is a $B_1$-minor of $M$.  For $i\in 1,2$, let $B_i'\coloneqq B_i \cap E(M')$. Finally, observe that $M'$ is represented by $[I_k ~ A]$ where the columns of $I_k$ correspond to the elements of $B_1'$ and the columns of $A$ correspond to the elements of $B_2'$ if an appropriate labeling of the columns is chosen.
\end{proof}

We will prove \cref{thm:reprweak} by showing that matroids representable by a very specific class of matrices satisfy its conclusion. For this, we need a statement showing that certain matrices are nonsingular, which we derive from an explicit formula for the determinants of $(\alpha,\beta,\gamma)$-triangular matrices due to Efimov~\cite{efimov2021determinant}.

\begin{lemma}[Efimov~\cite{efimov2021determinant}]
    Let $q$ be a prime power, $\alpha, \beta, \gamma \in \GF(q)$, and let $A$ denote the $(\alpha, \beta, \gamma)$-diagonal matrix of size $t$.  Then,
    \[\det(A) = \begin{cases}\frac{\alpha{(\beta-\gamma)}^t-\gamma{(\beta-\alpha)}^t}{\alpha -\gamma} & \text{if $\alpha \ne \gamma$,} \\ {(\beta-\alpha)}^{t-1}((t-1)\cdot\alpha + \beta) & \text{if $\alpha = \gamma$,} \end{cases} \]
    where in the last expression $(t-1)\cdot\alpha$ is shorthand for $\alpha+\dotsb+\alpha$ ($t-1$ times).
\end{lemma}

\begin{restatable}{proposition}{nonsing} \label{prop:nonsing}
    Let $q$ be a prime power, $(\alpha,\beta,\gamma)$ a feasible triple, $t$ a multiple of $q(q-1)$, and $A$ the $(\alpha,\beta,\gamma)$-diagonal matrix of size $t$. Then, $A$ is nonsingular.
\end{restatable}

\begin{proof}
    As $(\alpha, \beta, \gamma)$ is feasible, $\alpha \neq \beta$ holds.
    First suppose that $\alpha \neq \gamma$. If $\beta = \gamma$, we obtain $\beta=\gamma\ne 0$ by the feasibility of $(\alpha, \beta, \gamma)$, hence $\det(A)=\frac{-\gamma{(\beta-\alpha)}^t}{\alpha -\gamma} \neq 0$. Otherwise, if $\beta \ne \gamma$, then using that $t$ is divisible by $q-1$, we obtain 
    \begin{align*}
        \det(A)=\frac{\alpha{(\beta-\gamma)}^t-\gamma{(\beta-\alpha)}^t}{\alpha -\gamma}
        =\frac{\alpha -\gamma}{\alpha -\gamma}
        \neq 0.
    \end{align*}
    Now suppose that $\alpha=\gamma$. As $\alpha \neq \beta$ and $t$ is divisible by $q$, we obtain that 
    \[
       \det(A)={(\beta-\alpha)}^{t-1}((t-1)\cdot\alpha + \beta)={(\beta-\alpha)}^{t}\neq 0. \qedhere
    \]
\end{proof}

We are now ready to conclude the result for the specific class of matroids.

\begin{lemma}\label{lem:triisk}
    There is a computable function $f_2\colon \N \times \N \to \N$ with the following properties: Let $q$ be a prime power, $k$ a positive integer, and $M$ a matroid that can be represented by $[I ~ A]$ over $\GF(q)$, where $I$ is an identity matrix of size $f_2(q,k)$ and $A$ is an $(\alpha,\beta,\gamma)$-diagonal matrix of the same size for a feasible triple $(\alpha,\beta,\gamma)$. Next, let $B_1$ and $B_2$ be the subsets of $E(M)$ corresponding to $I$ and $A$, respectively. Then, $(B_1,B_2)$ is a basis partition of $M$ and has the $k$-exchange property.
\end{lemma}
\begin{proof}
    Let $f_2$ be the function defined by $f_2(q,k)\coloneqq q(q-1)k$ for all positive integers $q$ and $k$. By \cref{prop:nonsing}, we have that $A$ is nonsingular and hence $(B_1,B_2)$ is a basis partition of $M$. For $i \in [k]$, let $X_i$ be the subset of $B_1$ and $Y_i$ be the subset of $B_2$ that corresponds to the columns of indices $q(q-1)(i-1)+1$ to $q(q-1)i$ of $I$ and $A$, respectively.  For $Z \subseteq [k]$, let $B_Z\coloneqq \left(B_1 \setminus \bigcup_{i\in Z} X_i\right) \cup \bigcup_{i \in Z} Y_i$. It suffices to prove that $B_Z$ is a basis of $M$ for every $Z \subseteq [k]$. To this end, consider some fixed $Z \subseteq [k]$. Observe that the matrix obtained from restricting $[I ~ A]$ to the columns corresponding to $B_Z$ can be transformed into a matrix of the form $\begin{bsmallmatrix} I' & A_1 \\ O & A'\end{bsmallmatrix}$ by exchanging rows and columns. Here, $I'$ is the identity matrix of size $q(q-1)(k-|Z|)$, $O$ is a zero matrix, $A'$ is an $(\alpha,\beta,\gamma)$-diagonal matrix of size $q(q-1)|Z|$, and $A_1$ is an arbitrary matrix. As the size of $A'$ is divisible by $q(q-1)$, we obtain by \cref{prop:nonsing} that $A'$ is nonsingular. It follows that $\begin{bsmallmatrix} I' & A_1 \\ O & A'\end{bsmallmatrix}$ is nonsingular and hence $M \mathbin{|} B_Z$ is the free matroid. As $|B_Z|=|B_1|$ by construction, we obtain that $B_Z$ is a basis of $M$. This finishes the proof.
\end{proof}

Finally, we combine \cref{lem:minorexist,lem:triisk,lem:redtominor} to conclude \cref{thm:reprweak}.

\begin{proof}[Proof of \cref{thm:reprweak}]
    We prove the statement for the function $f\colon \N \times \N \to \N$  defined by $f(q,k)\coloneqq f_1(f_2(q,k),k)$ for each $k \in \N$ and prime power $q$. Let $B_1$ and $B_2$ be bases of a $\GF(q)$-representable matroid $M$ with $|B_1 \setminus B_2| \ge f(q,k)$. We need to prove that $(B_1,B_2)$ has the $k$-exchange property.
    Let $M'\coloneqq M \mathbin{/} (B_1\cap B_2) \mathbin{\backslash} (E(M)\setminus (B_1 \cup B_2))$. Further, for $i=1,2$, let $B_i'\coloneqq B_i\cap E(M')$ and observe that $(B_1',B_2')$ is a basis partition of $M'$. It follows from \cref{lem:minorexist} that there exists a $B_1'$-minor $M''$ of $M'$ that can be represented by a matrix of the form $[I ~ A]$, where $I$ is the identity matrix of size $f_2(q,k)$, $A$ is an $(\alpha,\beta,\gamma)$-diagonal matrix of size $f_2(q,k)$ for a feasible triple $(\alpha,\beta,\gamma)$ and the columns of $I$ and $A$ correspond to the elements of $B''_1$ and $B''_2$, respectively, where $B_i''\coloneqq B_i' \cap E(M'')$ for $i = 1,2$. We now obtain from \cref{lem:triisk} that $(B_1'',B_2'')$ is a basis partition of $M''$ and has the $k$-exchange property in $M''$. As $M''$ is a $B_1$-minor of $M$, we now obtain from \cref{lem:redtominor} that $(B_1,B_2)$ has the $k$-exchange property in $M$. 
\end{proof}

Combining \cref{thm:reprweak,thm:weakly}, \cref{lem:cons}, and \cref{thm:liuxu}, we get the following. 

\begin{corollary} \label{cor:zero-basis-fpt}
    Let $q$ be a prime power, $M$ a $\GF(q)$-representable matroid, $\psi\colon E \to \Gamma$ a group labeling and $F \subseteq \Gamma$ a finite set of forbidden labels. When $|F|$ is fixed, then \textscup{$F$-avoiding Basis} is solvable in polynomial time. Moreover, if $|\Gamma|$ is finite, then \textscup{Zero Basis} is FPT when parameterized by $|\Gamma|$. 
\end{corollary}

\subsubsection{Binary and Regular Matroids} \label{sec:regular}

As a strengthening of the $k$-exchange property, we say that the basis pair $(B_1, B_2)$ of a matroid has the \emph{elementary $k$-exchange property} if there exist $k$-element subsets $X\subseteq B_1 \setminus B_2$ and $Y \subseteq B_2 \setminus B_1$ and a bijection $\varphi\colon X \to Y$ such that $(B_1 \setminus Z) \cup \varphi(Z)$ is a basis for each $Z \subseteq X$. Note that this is equivalent to requiring $|X_i|=|Y_i|=1$ for each $i \in [k]$ in the definition of the $k$-exchange property.  We define a matroid $M$ to be \emph{elementarily $(\alpha, k)$-weakly base orderable} if $(B_1, B_2)$ has the elementary $k$-exchange property for any pair of basis $B_1$ and $B_2$ with $|B_1 \setminus B_2| \ge \alpha$. 

It turns out that all regular matroids are elementarily $(f(k), k)$-weakly base orderable for some large function $f\colon \N \to \N$, while the same is not true for binary matroids. While these results do not have immediate algorithmic applications, we believe that they are interesting in their own respect.

The following result shows that the statement of \cref{thm:reprweak}  does not hold if we require elementary $(f(k), k)$-weak base orderability instead of $(f(k), k)$-weak  base orderability, not even for binary matroids.

\begin{restatable}{theorem}{spike} \label{prop:spike}
    Let $r \ge 4$ be an even integer and let $M$ denote the binary matroid represented by $[I_r ~ J_r-I_r]$ where $I_r$ is the identity matrix and $J_r$ is the square all-one matrix of size $r$. Then, $M$ is not elementarily $(r, 3)$-weakly base orderable.
\end{restatable}

\begin{proof}
    Let $x_1,\dots, x_r$ and $y_1,\dots, y_r$ denote the elements of $E(M)$ corresponding to the first and last $r$ columns of $[I_r ~ J_r-I_r]$, respectively. Let $B_1 \coloneqq \set{x_1,\dots, x_r}$ and $B_2 \coloneqq \set{y_1,\dots,y_r}$.
    Since $r$ is even, both $B_1$ and $B_2$ are bases of $M$.  Assume that $M$ is elementarily $(r,3)$-weakly base orderable, then there exist 3-element subsets $X \subseteq B_1$ and $Y \subseteq B_2$ and a bijection $\varphi\colon X \to Y$ such that $(B_1 \setminus Z) \cup \varphi(Z)$ is a basis for each $Z \subseteq X$.  We may assume by symmetry that $X=\set{x_1,x_2,x_3}$. Since $(B_1\setminus X)\cup\set{y_1,y_2,y_3}$ is not a basis, $Y\ne \set{y_1,y_2,y_3}$, thus we may assume that $\varphi(x_1)=y_4$. For $i \in \set{2,3}$, $(B_1\setminus \set{x_1, x_i})\cup\set{y_4, \varphi(x_i)}$ being a basis implies that $\varphi(x_i) \in \set{y_1,y_i}$, while $B_1-x_i+\varphi(y_i)$ being a basis implies that $\varphi(x_i) \ne y_i$. This shows that $\varphi(y_i) = y_1$ for $i \in \{2,3\}$, which is a contradiction.
\end{proof}

By using that the binary matroid of $[I_4~J_4-I_4]$ is not regular, we can show the following statement strenghtening of \cref{thm:reprweak} holds for regular matroids. The proof is a modification of that of \cref{thm:reprweak}.

\begin{restatable}{theorem}{regular}\label{thm:regular}
There is a computable function $f'\colon \N \to \N$ such that every regular matroid is elementarily $(f'(k), k)$-weakly base orderable for any positive integer $k$.
\end{restatable} 

\begin{proof}
    Let $f'(k) \coloneqq f_1(2, \max\set{k, 4})$ for $k \in  \N$ where $f_1$ is the function provided by \cref{lem:minorexist}. Let $B_1$ and $B_2$ be bases of a regular matroid $M$ with $|B_1 \setminus B_2| \ge f'(k)$, we need to show that $(B_1, B_2)$ has the elementary $k$-exchange property. We may assume that $k \ge 4$. 
    Let $M'\coloneqq M/(B_1 \cap B_2) \backslash (E(M)\setminus (B_1\cup B_2))$, and $B'_i \coloneqq B_i \cap E(M')$ for $i=1,2$. Observe that $(B'_1, B'_2)$ is a basis partition of $M'$.
    It follows from \cref{lem:minorexist} that there exists a $B'_1$-minor $M''$ of $M'$ that can be represented over $GF(2)$ by a matrix of the form $[I_k~A]$ where $I_k$ is the identity matrix of size $k$, $A$ is an $(\alpha, \beta, \gamma)$-diagonal matrix of size $k$ for a feasible triple $(\alpha. \beta, \gamma)$ and the columns of $I_k$ and $A$ correspond to the elements of $B''_1$ and $B''_2$, respectively where $B''_i \coloneqq B'_i \cap E(M'')$ for $i=1,2$. By the feasibility of $(\alpha, \beta, \gamma)$, we have $(\alpha, \beta, \gamma) \in \set{(0,1,0), (0,1,1), (1,0,1)}$. 
    If $(\alpha, \beta, \gamma) = (1,0,1)$, then $A=J_k-I_k$ where $J_k$ is the all-ones matrix of size $k$. Then, as $k \ge 4$, $M''$ has a minor that can represented by $[I_4 ~ J_4-I_4]$ over $GF(2)$. This is a contradiction, as this matroid is not regular (see e.g.~\cite[page 645]{oxley2011matroid}), and the class of regular matroid is minor-closed. We concluded that $(\alpha, \beta, \gamma) \in \set{(0,1,0), (0,1,1)}$.
    Let $\varphi\colon B''_1 \to B''_2$ denote the bijection mapping the element of $B''_1$ corresponding to the $i$th column if $I_k$ to the element of $B''_2$ corresponding to the $i$th column of $A$ for $i \in [k]$. Fix a subset $Z \subseteq B''_1$. 
    As in the proof of \cref{lem:triisk}, the columns of $[I~A]$ corresponding to $(B''_1\setminus Z)\cup \varphi(Z)$ can be transformed into a matrix of the form $\begin{bsmallmatrix} I_{|Z|} & A_1 \\ O & A'\end{bsmallmatrix}$ where $I_{|Z|}$ is the identity matrix of size $|Z|$, $A'$ is the $(\alpha, \beta, \gamma)$-diagonal matrix of size $k-|Z|$, $O$ is an all-zeros matrix and $A_1$ is an arbitrary matrix. This matrix is nonsingular, as its determinant equals to $1$ by $(\alpha, \beta, \gamma) \in \set{(0,1,0), (0,1,1)}$. It follows that $(B''_1 \setminus Z) \cup \varphi(Z)$ is a basis for each $Z \subseteq B''_1$, thus $(B''_1, B''_2)$ has the elementary $k$-exchange property. Using \cref{lem:redtominor2}, we conclude that $M$ is elementarily $(f'(k), k)$-weakly base orderable.
\end{proof}

\subsubsection{Graphic Matroids}\label{sec:graph}

As graphic matroids are regular, they are elementarily $(f(k), k)$-weakly base orderable by \cref{thm:regular} for some large function $f\colon \N \to \N$. We give an independent proof from the proofs of \cref{thm:reprweak,thm:regular}, which shows that for graphic matroids there exists such a function $f$ satisfying $f(k)=O(k^3)$.

\begin{restatable}{theorem}{graphicwbo}\label{thm:graphicwbo}
    Graphic matroids are elementarily $(3k^3, k)$-weakly base orderable for any $k \ge 1$.
\end{restatable}

For the proof, we will use the following statement which can be proved analogously to \cref{lem:redtominor}.
\begin{lemma} \label{lem:redtominor2}
    Let $B_1$ and $B_2$ be two bases of a matroid $M$. Further, let $M'$ be a $B_1$-minor of $M$ such that $(B'_1, B'_2)$ is a basis partition of $M'$ and has the elementary $k$-exchange property for some $k \ge 1$, where $B'_i \coloneqq B_i \cap E(M')$ for $i=1,2$. Then, $(B_1, B_2)$ has the elementary $k$-exchange property in $M$.
\end{lemma}

Next, we introduce some notation from graph theory. For a graph $G$, we use the notations $V(G)$ and $E(G)$ to denote its set of vertices and edges, respectively. Throughout this section, for a graph $G$, a spanning subgraph $H$ of $G$, $F \subseteq E(H)$ and $F'\subseteq E(G)\setminus E(H)$, we use $H-F$ for $(V(G),E(H)\setminus F)$ and $H+F'$ for $(V(G),E(H)\cup F')$. For disjoint subsets $U_1, U_2 \subseteq V(G)$ we use the notations $\delta_G(U_1, U_2) \coloneqq \Set{u_1u_2 \in E(G)}{u_1 \in U_1, u_2 \in U_2}$ and $\delta_G(U_1)\coloneqq \delta_G(U_1, V(G)\setminus U_1)$.
We use the notation $\delta_G(v) = \delta_G(\set{v})$ for a vertex $v \in V(G)$, and $d_G(v) = |\delta_G(v)|$ denotes the degree of $v$.
For a subset $U \subseteq V(G)$ we say that an edge $uv$ is \emph{induced} by $U$ if $u,v \in U$, and $G[U]$ denotes the graph on vertex set $U$ whose edges are the edges of $G$ induced by $U$.
An edge $uv$ \emph{links} the (not necessarily disjoint) subsets $U_1, U_2 \subseteq V(G)$ if $u \in U_1$ and $v \in U_2$.
The following statement follows from Euler's theorem on Eulerian walks, and can also be found in \cite{harary1972evolution}.

\begin{proposition} 
    Let $T$ be a tree with $\ell$ vertices of odd degree. Then, the minimum number of edge-disjoint paths covering the edges of $T$ equals $\lceil \ell /2 \rceil$.
\end{proposition}
\begin{corollary} \label{cor:decomp}
    Let $T$ be a tree and let $k$ denote the number of its leaves. Then, the edges of $T$ can be covered by $k-1$ edge-disjoint paths.
\end{corollary}
\begin{proof}
If $T$ has $n$ vertices from which $\ell$ have odd degrees, then $2n-2 = \sum_{v \in V} d_T(v) \ge k + (\ell-k)\cdot 3 + (n-\ell)\cdot 2$, thus $\ell/2 \le k-1$.
\end{proof}

We will also use the following observation.

\begin{lemma}\label{lem:cutcycles}
    Let $G$ be a graph containing a Hamiltonian path $P=v_1,\dots,v_n$ and let $Z \subseteq E(G)\setminus E(P)$ with $|Z|\geq k \ell$ for two integers $k,\ell\ge 0$. Then, either there exists some $q \in [n]$ such that $|Z \cap \delta_G(\set{v_1,\ldots,v_{q-1}},\set{v_{q},\ldots,v_{n}})|\geq k$ or there exists a collection of $\ell$ edge-disjoint cycles in $G$ each of which contains exactly one edge of $Z$.
\end{lemma}
\begin{proof}
    We prove the statement by induction on $\ell$. Let $i$ be the smallest integer in $[n]$ such that $\set{v_1,\ldots,v_i}$ induces an edge of $Z$. Observe that there exists a cycle $C$ with $V(C)\subseteq \set{v_1,\ldots,v_i}$ that contains exactly one edge of $Z$. If $|\delta_G(\set{v_1,\ldots,v_{i-1}},\set{v_i,\ldots,v_n})\cap Z|\geq k$, there is nothing to prove. We may hence suppose that this is not the case, so $G-\set{v_1,\ldots,v_{i-1}}$ contains at least $(\ell-1)k$ edges of $Z$. We obtain by induction that either there exists some $q \in \set{i+1,\ldots,n}$ such that $|Z \cap \delta_G(\set{v_i,\ldots,v_{q-1}},\set{v_{q},\ldots,v_{n}})|\geq k$ or there exists a collection $\mathcal{C}$ of $\ell-1$ edge-disjoint cycles in $G-\set{v_1,\ldots,v_{i-1}}$ each of which contains exactly one edge of $Z$. In the first case, we obtain $|Z \cap \delta_G(\set{v_1,\ldots,v_{q-1}},\set{v_{q},\ldots,v_{n}})|\geq |Z \cap \delta_G(\set{v_i,\ldots,v_{q-1}},\set{v_{q},\ldots,v_{n}})|\geq k$ and in the second case, we obtain the desired collection of cycles by adding $C$ to $\mathcal{C}$.
\end{proof}

We are now ready to prove \cref{thm:graphicwbo}. Vaguely speaking, we use the fact that every large tree either contains a lot of leaves or contains long paths. Either of these structures will prove beneficial for us.

\begin{proof}[Proof of \cref{thm:graphicwbo}]
    Assume that $\alpha\ge k \ge 1$ are integers such that there exists a graphic matroid which is not elementarily $(\alpha, k)$-weakly base orderable. Let $G$ be a connected graph representing such a matroid with the minimal number of edges possible. There exist spanning trees $T_1$ and $T_2$ of $G$ not satisfying the elementary $k$-exchange property with $|E(T_1) \setminus E(T_2)| \ge \alpha$. By \cref{lem:redtominor2}, contracting $E(T_1) \cap E(T_2)$ and deleting $E(G) \setminus (E(T_1) \cup E(T_2))$ in $G$ gives a graph which is also not elementarily $(\alpha, k)$-weakly base orderable. Hence,  using that $G$ has the minimum number of edges, we obtain $E(T_1) \cap E(T_2)=\emptyset$ and $E(T_1) \cup E(T_2)=E(G)$. 
    This means that $|E(T_1)|\ge \alpha$ and for any choice of $k$-element subsets $X \subseteq E(T_1)$ and $Y \subseteq E(T_2) $ and any bijection $\varphi\colon X \rightarrow Y$, there exists some $Z \subseteq X$ such that $T_1-Z+\varphi(Z)$ is not a spanning tree of $G$.

    \begin{claim} \label{cl:fewleaves}
        $T_1$ has less than $k$ leaves.
    \end{claim}
    \begin{claimproof}
        Suppose otherwise, and let $L$ be a set of $k$ leaves of $T_1$.
        For every $v \in L$, let $x_v$ be the unique edge in $\delta_{T_1}(v)$, and let $X\coloneqq \Set{x_v}{v \in L}$.
        By \cref{lem:bijection}, there exists a bijection $\phi\colon E(T_1) \to E(T_2)$ such that $T_1-e+\phi(e)$ is a spanning tree for each $e \in E(T_1)$.
        Let $y_v \coloneqq \phi(x_v)$ for $v \in L$, define $Y \coloneqq \Set{y_v}{v\in L}$ and let $\varphi$ denote the restriction of $\phi$ to $X$. 
        Observe that $|X|=|Y|=|L|=k$.
        By the choice of $G$, $T_1$ and $T_2$, there exists some $Z \subseteq X$ such that $T_1-Z+\varphi(Z)$ contains a cycle $C$.
        Note that $Z$ is nonempty since $T_1$ is a spanning tree.
        Let $v_1,\dots, v_\ell \in L$ such that $E(C)\cap \varphi(Z) = \set{y_{v_1}, \dots, y_{v_\ell}}$.
        Since $T_1-x_{v_i}+y_{v_i}$ is a spanning tree for $i \in [\ell]$, the edge $y_{v_i}$ is incident to $v_i$, that is, $y_{v_i} = \set{u_i, v_i}$ for some $u_i \in V(G)$.
        As $v_i$ is isolated in $T_1-Z$,    
        both edges incident to $v_i$ in $C$ are in $E(C) \cap \varphi(Z) = \set{y_{v_1}, \dots, y_{v_\ell}}$, thus $v_i \in \set{u_1, \dots, u_\ell}$. 
        We showed that $\set{v_1,\dots, v_\ell} \subseteq \set{u_1,\dots, u_\ell}$, hence  $\set{v_1,\dots, v_\ell} = \set{u_1,\dots, u_\ell}$ and each of $v_1,\dots, v_\ell$ has degree two in $\set{y_{v_1},\dots,y_{v_\ell}}$. This contradicts that $T_2$ is a spanning tree.
    \end{claimproof}

    \begin{claim} \label{cl:twopaths}
        $T_1$ does not contain edge-disjoint paths $P_1$ and $P_2$ such that $V(P_1)$ and $V(P_2)$ are linked by at least $2k$ edges.
    \end{claim}
    \begin{claimproof}
        Suppose otherwise, so there exist edge-disjoint paths $P_1$ and $P_2$ and a set $E'$ of at least $2k$ edges in $E(T_2)$ linking $V(P_1)$ and $V(P_2)$.
        For $i \in [2]$, let $W_i$ be the set of vertices in $V(P_i)$ incident to an edge in $E'$.
        As $T_2[W_1\cup W_2]$ is a forest, we obtain $|W_1|+|W_2|\geq |E'|+1 \ge 2k+1$.
        By symmetry, we may suppose that $|W_1|\geq \lceil (|W_1|+|W_2|)/2 \rceil \geq k+1$.
        Let $P_1=v_1,\ldots, v_q$, and let $q_0 \in [q]$ be the unique index such that $v_{q_0}$ is in the same component of $T_1-E(P_1)$ as $V(P_2)$.
        As $|W_1| \ge k+1$, there exist indices $1 \le i_1 < \dots < i_k \le q$ such that $\set{v_{i_1},\dots, v_{i_k}} \subseteq W_1-v_{q_0}$.
        Note that $\set{v_{i_1},\dots, v_{i_k}} \subseteq V(P_1) \setminus V(P_2)$, as if $P_1$ and $P_2$ have a common vertex then $V(P_1) \cap V(P_2) = \set{v_{q_0}}$. 
        For $j \in [k]$, let $y_j$ be an edge in $E'$ incident to $v_{i_j}$, let $x_j \coloneqq v_{i_j}v_{i_j+1}$ if $i_j < q_0$, and let $x_j \coloneqq v_{i_j}v_{i_j-1}$ if $i_j > q_0$.
        Let $X\coloneqq \Set{x_j}{j \in [k]}$, $Y\coloneqq \Set{y_j}{j \in [k]}$ and let $\varphi\colon X \to Y$ be defined by $\varphi(x_j) = y_j$ for $j \in [k]$.
        We prove that $T_1-Z+\varphi(Z)$ is a spanning tree for each $Z \subseteq X$ by showing that each vertex of $P_1$ is contained in the same component of $T_1-Z+\varphi(Z)$ as $V(P_2)$.
        By symmetry, it suffices to prove the statement for vertices $v_s$ with $s \in [q]$ and $s \ge q_0$.
        We prove it by induction on $s$, the statement holds for $s=q_0$ by the definition of $q_0$.
        Let $s \ge q_0+1$.
        If $v_s \not \in \set{v_{i_1},\dots, v_{i_k}}$, or $s=i_j$ for some $j \in [k]$ and $x_j=v_{i_j}v_{i_{j}-1} \not \in Z$, then $T_1-Z+\varphi(Z)$ contains the edge $v_sv_{s-1}$, thus the statement follows by induction.
        Otherwise, $s=i_j$ for some $j \in [k]$ with $x_j \in Z$, thus $T-Z+\varphi(Z)$ contains the edge $y_j$ which links $v_{i_j}$ to $V(P_2)$.
        This shows that $T_1-Z+\varphi(Z)$ is a spanning tree for each $Z \subseteq X$, which is a contradiction by $|X|=|Y|=k$.
    \end{claimproof}
    
    \begin{claim}\label{cl:disjcycles}
        $G$ does not contain a collection of $k$ edge-disjoint cycles each of which contains exactly one edge of $T_2$.
    \end{claim}
    \begin{claimproof}
        Suppose otherwise and let $\mathcal{C}$ be a collection of such cycles. 
        For every $C \in \mathcal{C}$, let $y_C$ be the unique edge in $E(C)\cap E(T_2)$ and let $x_C$ be an arbitrary edge in $E(C)\cap E(T_1)$.
        Let $X \coloneqq \Set{x_C}{C \in \mathcal{C}}, Y\coloneqq \Set{y_C}{C \in \mathcal{C}}$, and let $\varphi\colon X \rightarrow Y$ be the mapping with $\varphi(x_C)=y_C$ for every $C \in \mathcal{C}$. Observe that $|X|=|Y|=k$. 
        We prove by induction on $|Z|$ that $T_1-Z+\varphi(Z)$ is a spanning tree for each $Z \subseteq X$.
        If $|Z| \ge 1$, then let $x_C \in Z$ be an arbitrary element.
        $T \coloneqq T_1-(Z-x_C)+\varphi(Z-x_C)$ is a spanning tree by the induction hypothesis.
        The fundamental cycle of $y_C$ with respect to $T$ is $C$, in particular it contains $x_C$, thus $T_1-Z+\varphi(Z) = T-x_C+y_C$ is a spanning tree. This contradicts the choice of $G$, $T_1$, and $T_2$. 
    \end{claimproof}

    By \cref{cl:fewleaves} and \cref{cor:decomp}, we obtain that there exists a collection $\mathcal{P}$ of at most $k-2$ edge-disjoint paths in $T_1$ such that $\bigcup_{P \in \mathcal{P}}E(P)=E(T_1)$.
    By \cref{cl:twopaths}, for any two distinct paths $P_1, P_2 \in \mathcal{P}$ at most $2k-1$ edges link $V(P_1)$ to $V(P_2)$.
    If $V(P)$ induces at least $2k^2$ edges of $T_2$ for some $P \in \mathcal{P}$, then \cref{lem:cutcycles} implies that either there exists a set of $2k$ edges of $E(T_2)$ linking two edge-disjoint subpaths of $P$ or there exists a collection of $k$ edge-disjoint cycles in $G$ each of which contains exactly one edge of $E(T_2)$.  We obtain contradictions to Claims~\ref{cl:twopaths} and \ref{cl:disjcycles}, respectively. Therefore, $V(P)$ induces at most $2k^2-1$ edges of $T_2$ for each $P \in \mathcal{P}$, hence
    \[|E(T_2)| \le \binom{k-2}{2} \cdot (2k-1) + (k-2)\cdot (2k^2-1) = \frac{6k^3-19k^2+15k-2}{2}< 3k^3.\]
    We got that $\alpha \le |E(T_2)| < 3k^3$ for any $\alpha \ge k$ such that some graphic matroid is not elementarily $(\alpha, k)$-weakly base orderable, that is, graphic matroids are elementarily $(3k^3, k)$-weakly base orderable.  
\end{proof}

\subsection{Two forbidden labels}
The objective of this section is to prove the following restatement of the case $|F|=2$ of \cref{conj:k}. 
Vaguely speaking, we first reduce the problem to matroids on six elements and then combine some earlier results with a particular treatment for the cycle matroid of $K_4$.

In order to prove the result, we need to prove that an $F$-avoiding basis can be found by exchanging at most two elements which we do by means of a contradiction. The following two results allow us to only consider the case when the matroid has rank three. This also follows from the proof of \cite[Theorem~2]{liu2023congruency}, but we give proofs to keep the paper self-contained. 

\begin{lemma}\label{lem:distance-F-reduction}
    Let $M$ be a matroid, $\psi\colon E(M)\rightarrow \Gamma$ a group labeling and $F\subseteq \Gamma$ a finite set of labels. 
    Suppose that there exists at least one $F$-avoiding basis of $M$.
    If there exists a basis $A$ such that $|A\setminus A'|\geq |F|+1$ for all $F$-avoiding bases $A'$ of $M$, then there also exists a basis $B$ and an $F$-avoiding basis $B^*$ such that $|B\setminus B^*|=|F|+1$ and $|B\setminus B'|\geq |F|+1$ for all $F$-avoiding bases $B'$.
\end{lemma}
\begin{proof}
    Let $B$ be a basis of $M$ for which the minimum of $|B\setminus B'|$ over all $F$-avoiding bases $B'$ is minimum among all bases that satisfy $|B\setminus B'|\geq |F|+1$ for all $F$-avoiding bases $B'$ and let $B^*$ be an $F$-avoiding basis for which this minimum is attained. Observe that $B$ is well-defined by assumption and if $|B\setminus B^*|=|F|+1$, there is nothing to prove. We may hence suppose that $|B\setminus B^*|\geq |F|+2$. By the basis exchange axiom, there exist $x \in B\setminus B^*$ and $y \in B^*\setminus B$ such that $B_1\coloneqq B-x+y$ is a basis. Clearly, by $|B\setminus B_1|=1$ and the choice of $B$, we have $\psi(B_1)\in F$. If there exists an $F$-avoiding basis $B_1^*$ with $|B_1\setminus B_1^*|\leq |F|$, we obtain $|B\setminus B_1^*|\leq |B\setminus B_1|+|B_1\setminus B_1^*|\leq |F|+1 $, a contradiction. We hence obtain $|B_1\setminus B'|\geq |F|+1$ for every $F$-avoiding basis $B'$. Further, we have $|B_1\setminus B^*|=|B\setminus B^*|-1$, a contradiction to the choice of $B$.
  \end{proof}

  \begin{proposition}\label{prop:shrink}
      Let $M$ be a matroid, $\psi\colon E(M)\rightarrow \Gamma$ a group labeling, $F \subseteq \Gamma$ a finite subset of labels and let $B$ be a basis such that $|B \setminus \bar{B}|\geq  |F|+1$ for every $F$-avoiding basis $\bar{B}$. Let $M'$ be obtained from $M$ by contracting a set $X \subseteq B$, let $B'\coloneqq B\setminus X$ and let $F'\coloneqq \Set{f -\psi(X)}{f \in F}$. Then, $|B' \setminus \bar{B'}|\geq  |F|+1$ for every basis $F'$-avoiding basis $\bar{B'}$ of $M'$.
  \end{proposition}
  \begin{proof}
      Suppose otherwise, so there exists a basis $B_1'$ of $M'$ which is $F'$-avoiding and satisfies $|B'\setminus B_1'|\leq |F|$. Let $B_1\coloneqq B_1'\cup X$, then $B_1$ is a basis of $M'=M/X$ by the definition of contraction. Observe that $\psi(B_1)=\psi(B_1')+\psi(X)\notin F$ as $\psi(B_1')\notin F'$. Further, by construction and as $B'$ and $B_1'$ are bases of $M'$, we have $|B_1|=|B|-|B'|+|B_1'|=|B|$. Hence $B_1$ is an $F$-avoiding basis of $M$ with $|B\setminus B_1|=|B'\setminus B_1'|\leq |F|$, a contradiction.
  \end{proof}

It will turn out that in the forthcoming proof, the cycle matroid of $K_4$ will play a crucial role. It needs to be treated separately due to the following result. Mayhew and Royle \cite{Mayhew2008nine} showed by the implementation experiment that the number of rank-3 matroids on six elements which is not strongly base orderable is one. Since $M(K_4)$ is not strongly base orderable, we can see that a rank-3 matroid on six elements is either isomorphic to $M(K_4)$ or strongly base orderable.  

\begin{proposition}[Mayhew--Royle \cite{Mayhew2008nine}]\label{prop:6bo}
    Let $M$ be a rank-3 matroid on six elements that is not isomorphic to $M(K_4)$. Then, $M$ is strongly base orderable.
\end{proposition}

We now give our main result for the special case that the matroid in consideration is the cycle matroid of $K_4$.
\begin{lemma}\label{lem:k4}
    Let $M\coloneqq M(K_4)$, let $\psi\colon E(M)\to \Gamma$ be a group labeling and let $F \subseteq \Gamma$ be a 2-element set such that $M$ has at least one $F$-avoiding basis. Then, for any basis $B$, there exists an allowed basis $B^*$ with $|B\setminus B^*|\leq 2$.
\end{lemma}
\begin{proof}
    Suppose otherwise, so there exist two bases $B$ and $B^*$ of $M$ such that $|B\setminus B'|\geq 3$ for every allowed basis $B'$ and $B^*$ is an allowed basis. By symmetry, we may suppose that $B=\set{a_1,a_2,a_3}$ and $B^*=\set{b_1,b_2,b_3}$ where for $i \in [3]$, the edges of $K_4$ corresponding to  $a_i$ and $b_i$ are marked in \cref{k4}.
    \begin{figure}[h!]
    \centering
        \includegraphics[width=.2\textwidth]{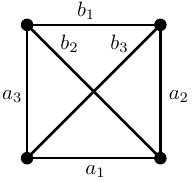}
        \caption{A drawing of $K_4$ with the edge labeling used in the proof of \cref{lem:k4}.}\label{k4}
\end{figure}Let $F=\set{f_1,f_2}$. By symmetry, we may suppose that $\psi(B)=f_1$.

    \begin{claim}\label{cl:two-forbidden-simul}
    For any $\emptyset \ne X \subsetneq B$ and $\emptyset \ne Y \subsetneq B^*$ such that $B - X + Y$ and $B^* + X - Y$ are bases, we have $\psi(B - X + Y) = \psi(B^* + X - Y) = f_2$.
  \end{claim}
   \begin{claimproof}
    Let $B_1=B - X + Y$ and $B_2=B^* + X - Y$. For $i \in [2]$, the label of $B_i$ is $f_1$ or $f_2$ since otherwise $B_i$ is an allowed basis and is closer to $B$ than $B^*$.
    In addition, if $B_i$ has label $f_1$ for some $i \in [2]$, we obtain
    \begin{align*}\label{eq:two-forbidden-eq}
        \psi(B_{3-i})=\psi(B) + \psi(B^*)-\psi(B_i) = f_1+\psi(B^*)-f_1=\psi(B^*),
    \end{align*}
    a contradiction to the previous observation.
  \end{claimproof}
  As $B-a_1+b_1$ and $B^*-b_1+a_1$ are bases of $M$, we obtain $\psi(B-a_1+b_1)=f_2$ by \cref{cl:two-forbidden-simul} and hence $\psi(b_1)-\psi(a_1)=\psi(B-a_1+b_1)-\psi(B)=f_2-f_1$ . Similar arguments show $\psi(b_2)-\psi(a_1)=\psi(b_3)-\psi(a_1)=\psi(b_1)-\psi(a_2)=\psi(b_1)-\psi(a_3)=f_2-f_1$. In particular, we obtain $\psi(a_1)=\psi(a_2)=\psi(a_3)$ and $\psi(b_1)=\psi(b_2)=\psi(b_3)$, so $\psi(b_i)-\psi(a_j)=f_2-f_1$ for all $i,j \in [3].$ As $B-a_1+b_1$ and $B^*-b_1+a_1$ are bases of $M$, we obtain $\psi(B^*-b_1+a_1)=f_2$  by \cref{cl:two-forbidden-simul}.

  This yields
   \begin{align*}
        f_2=\psi(B^*-b_1+a_1)=\psi(B)+(\psi(b_2)-\psi(a_2))+(\psi(b_3)-\psi(a_3))=f_1+2(f_2-f_1),
    \end{align*}
    so $f_1=f_2$, which is a contradiction because $f_1$ and $f_2$ are distinct labels. 
 \end{proof}
 
We are now ready to prove the main result of this section. 

\begin{restatable}{theorem}{forb}\label{thm:two-forbidden}
    Let $M$ be a matroid, $\psi\colon E(M)\rightarrow \Gamma$ a group labeling, and let $F$ be a 2-element subset of $\Gamma$. For any basis $B$, there exists an $F$-avoiding basis $B^*$ such that $|B \setminus B^*| \le 2$, provided that there exists at least one $F$-avoiding basis.
\end{restatable}

\begin{proof}

Suppose otherwise. Then, by \cref{lem:distance-F-reduction}, there exists a basis $B$ with $|B \setminus B'| \ge 3$ for all $F$-avoiding bases $B'$ and an $F$-avoiding basis $B_1$ with $|B \setminus B_1| = 3$. Let $M'\coloneqq M/(B \cap B_1)\setminus (E(M)\setminus(B \cup B_1))$, $B'\coloneqq B\setminus B_1$, $B_1'\coloneqq B\setminus B_1$, and $F'\coloneqq\Set{f -\psi(X)}{f \in F}$.  We obtain by \cref{prop:shrink} that $|B' \setminus \bar{B'}|\geq  3$ for every $F'$-avoiding basis $\bar{B'}$ of $M'$.  Further, we have $\psi(B_1')=\psi(B_1)-\psi(X)\notin F'$. Observe that $M'$ is a matroid on 6 elements. By \cref{prop:6bo}, we obtain that $M'$ is either strongly base orderable or isomorphic to $M(K_4)$. If $M'$ is strongly base orderable, we obtain a contradiction to \cref{cor:sbo}. If $M'$ is isomorphic to $M(K_4)$, we obtain a contradiction to \cref{lem:k4}.
\end{proof}

\section{Hardness and Negative Results} \label{sec:hard}

In this section, we give the algorithmic hardness results and counterexamples contained in this work. \cref{hard1,hard2}  contain algorithmic intractability results and \cref{sec:counter} contains a counterexample to a conjecture of Liu and Xu~\cite{liu2023congruency}.

\subsection{Hardness of Non-zero Common Basis with
\texorpdfstring{$\Z_2 \le \Gamma$}{Z2 <= Γ}}\label{hard1}
We here show that \textsc{Non-Zero Common Basis} is intractable for any group $\Gamma$ such that $\Z_2 \le \Gamma$. This implies that the condition on $\Gamma$ in \cref{thm:non-zero-common-basis-poly} is crucial indeed. Then, we explain the relation of this result to common base packing and dual lattices.

\begin{restatable}{theorem}{nonzerocommonbasishard}
\label{thm:non-zero-common-basis-hard}
    \textscup{Non-Zero Common Basis} requires an exponential number of independence queries for any fixed group $\Gamma$ such that $\Z_2 \le \Gamma$.    
\end{restatable}

We will use the following statement which can be derived from e.g.\ \cite[Theorem~5.3.5]{frank2011connections}.
The matroids arising in the form as in \cref{lem:sp} coincide with \emph{sparse paving matroids} of rank at least two. 

\begin{lemma} \label{lem:sp}
Let $r \ge 2$ be an integer, $E$ a ground set of size at least $r$, and $\cH \subseteq \binom{E}{r}$ such that $|H \cap H'| \le r-2$ for each distinct $H_1, H_2 \in \cH$. Then, $\binom{E}{r}\setminus \cH$ forms the family of bases of a matroid of rank $r$.
\end{lemma}

Using this construction and an information-theoretic argument, we prove \cref{thm:non-zero-common-basis-hard}. We note that the construction of the sparse paving matroid used in our proof is closely related to the notion of $\Pi$-matroids introduced in \cite{doronarad2023tight}.

\begin{proof}[Proof of \cref{thm:non-zero-common-basis-hard}]
    As $\Gamma$ contains $\Z_2$ as a subgroup, we may suppose  that $\Gamma = \Z_2$. Let $r \ge 2$, $A=\set{a_1,\dots, a_r}$ and $B=\set{b_1,\dots, b_r}$ be disjoint sets of size $r$, and $E\coloneqq A \cup B$. Consider \[\cH \coloneqq \Set{H \subseteq E}{\text{$|H \cap \set{a_1,b_1}| = \dots = |H\cap \set{a_r, b_r}| = 1$, $|H \cap A|$ is odd}}.\]
    Observe that $|H_1 \cap H_2| \le r-2$ holds for any two distinct $H_1, H_2 \in \cH$.
    Indeed, if $|H_1 \cap A| = |H_2 \cap A|$, then $H_1\setminus H_2$ contains at least one element of $A$ and at least one element of $B$, otherwise $|H_1 \cap A|$ and $|H_2 \cap A|$ differ by at least two since both of them are odd. 
    Then, by \cref{lem:sp}, $\binom{E}{r}\setminus \cH$ forms the basis family of a sparse paving matroid $M$. Let $N$ denote the unitary partition matroid defined by partition classes $\set{a_1, b_1}, \dots, \set{a_r, b_r}$, and $\psi\colon E \to \Z_2$ the labeling defined by $\psi(a_1)=\dots = \psi(a_r) = 1$ and $\psi(b_1) = \dots = \psi(b_r) = 0$. The common bases of $M$ and $N$ are exactly the sets containing one element from each of the pairs $\set{a_1, b_1}, \dots, \set{a_r, b_r}$ and an even number of elements from $A$, hence all common bases of $M$ and $N$ have label zero.
    For any $H \in \cH$ let $M_H$ denote the sparse paving matroid with family of bases $\binom{E}{r}\setminus (\cH \setminus \set{H})$. Observe that $M_H$ is well-defined for all $H \in \cH$  by \cref{lem:sp}. Since $H$ is a common basis of $M_H$ and $N$ having label 1, any algorithm for \textsc{Non-Zero Common Basis} must distinguish between the instance $M, N, \psi$ and the instances $M_H, N, \psi$ for $H \in \cH$. Since the only way to distinguish $M$ and an $M_H$ is to query the independence of $H$, any algorithm solving \textsc{Non-Zero Common Basis} must query the independence of each $H \in \cH$, and thus uses at least $|\cH| = 2^{r-1}$ independence queries.
\end{proof}

\begin{remark}
     Bérczi and Schwarcz~\cite{berczi2021complexity} showed that any algorithm which decides if the common ground set of two matroids can be partitioned into two common bases requires an exponential number of oracle queries. 
     Observe that if $r$ is odd in the proof of \cref{thm:non-zero-common-basis-hard}, then $E$ cannot be partitioned into two common bases of $M$ and $N$, while $(H, E\setminus H)$ is a partition into two common bases of $M_H$ and $N$ for any $H \in \cH$. 
     This gives a new and simpler proof of the result of \cite{berczi2021complexity}.
\end{remark}

The previous connection of a partitioning problem and a certain non-zero problem is no coincidence: relaxations of partitioning problems are related to certain non-zero problems for the group $\R/ \Z \simeq [0,1)$ via \emph{dual lattices}. The \emph{dual} of a lattice $L\subseteq \R^n$ is the set $L^* \coloneqq \Set{y \in \R^n}{y^T x \in \Z ~ (x \in L)}$. For a set family $\mathcal{F} \subseteq 2^E$, the dual of $\lat(\mathcal{F})$ is $\lat^*(\mathcal{F}) \coloneqq \Set{y \in \R^E}{y(F) \in \Z ~ (F\in \mathcal{F})}$ by letting $y(F)\coloneqq \sum_{e \in F} y_e$ for $F \subseteq E$. Therefore, deciding membership in the dual lattice $\lat^*(\mathcal{F})$ is equivalent to deciding the existence of a non-zero member of $\mathcal{F}$ for a labeling to the group $\R/\Z \cong [0, 1)$. By Hermite's theorem~\cite[Theorem~4.1.23]{frank2011connections}, $\lat(\mathcal{F})$ is related to its dual by $\lat(\mathcal{F}) = \Set{x \in \R^E}{\nexists y \in \lat^*(\mathcal{F}) \text{ with } y^Tx \not \in \Z}$.

The problem of partitioning $E$ into members of $\mathcal{F}$ is equivalent to finding values $x_F \in \{0,1\}$ for each $F \in \mathcal{F}$ such that $\sum_{e \in F} x_F = 1$ for each $e \in E$. By relaxing the condition $x_F \in \{0,1\}$ to $x_F \in \Z$ for each $F \in \mathcal{F}$, we get a relaxation which is equivalent to deciding the membership of the all-one vector $\ones$ in $\lat(\mathcal{F})$. By the above relation of $\lat(\mathcal{F})$ to its dual, $\ones \notin \lat(\mathcal{F})$ is equivalent to the existence of $y \in \lat^*(\mathcal{F})$ with $y^T \ones \notin \Z$, that is, the existence of $y \in \R^E$ such that $y(F) \in \Z$ for all $F \in \mathcal{F}$ and $y(E) \notin \Z$. 
This can be formulated as the existence of a labeling $\psi\colon E \to \R/\Z$ such that $\psi(E) \ne 0$ and all members of $\mathcal{F}$ have label zero.

\begin{remark}
The proof of \cref{thm:non-zero-common-basis-hard} also shows the hardness of the previous relaxation in case $\mathcal{F}$ is the family of common bases of two matroids, that is, any algorithm deciding $\ones \in \lat(\mathcal{F})$ requires an exponential number of oracle queries. Indeed, if $r$ is odd, then $\ones \in \lat(\cB(M_H) \cap \cB(N))$ for any $H \in \mathcal{H}$ as $(H, E\setminus H)$ is a partition into common bases of $M_H$ and $N$, while $\ones \not \in \lat(\cB(M) \cap \cB(N))$ as the labeling $\psi$ given in the proof satisfies $\psi(E) \ne 0$ and $\psi(B)=0$ for all $B \in \cB(M)\cap \cB(N)$.
\end{remark}

\subsection{Hardness of Zero Basis and Zero Common Basis}\label{hard2}
In this section, we show three hardness results for \textscup{Zero Basis} and \textscup{Zero Common Basis}. First, we show in \cref{thm:zerobasishard} that \textscup{Zero Basis} is \NPhard even if $M$ is a uniform matroid and $\Gamma=\Z$. Next, in \cref{thm:zero-basis-infinite}, we give a result showing that any algorithm solving \textscup{Zero Basis} for a matroid given by an independence oracle and a fixed infinite group $\Gamma$ requires an exponential number of independence queries. Finally, in \cref{thm:zero-common-base-fixed}, we show a similar result for \textscup{Zero Common Basis} even if the group is arbitrary nontrivial.

We now start by proving that \textscup{Zero Basis} is hard even for uniform matroids by using the hardness of the well-known \textsc{Subset Sum} problem. 
Recall that \textsc{Subset Sum} is the problem of determining whether there exists a subset $F$ of a given set of integers $S \subseteq \Z$ whose elements sum up to a given target value $t \in \Z$.
The following well-known result can be found in \cite[(8.23)]{kleinberg2006}.

\begin{proposition} \label{prop:subsetsum}
    \textscup{Subset Sum} is \NPhard even for instances $(S,t)$ with $s>0$ for all $s \in S$.
\end{proposition}

\cref{prop:subsetsum} implies hardness of \textsc{Zero Basis} as follows.

\begin{restatable}{theorem}{zerobasishard} \label{thm:zerobasishard}
    \textscup{Zero Basis} is \NPhard for a uniform matroid and $\Gamma = \Z$.
\end{restatable}

\begin{proof}
    Let $(S,t)$ be an instance of \textsc{Subset Sum} with $s>0$ for all $s \in S$. Let $k \coloneqq |S|$ and $E=\Set{x_s}{s \in S} \cup \set{y_1,\dots, y_k} \cup \set{z}$ a set of size $2k+1$. Let $M$ be the uniform matroid on $E$ having rank $k+1$. Define $\psi\colon E\to \Z$ by $\psi(x_s)=s$ for all $s \in S$, $\psi(y_i)=0$ for all $i \in [k]$, and $\psi(z)=-t$.

    We now show that $(S,t)$ is a positive instance of \textsc{Subset Sum} if and only if $M$ has a zero basis. First suppose that $(S,t)$ is a positive instance of \textsc{Subset Sum}, so there exists a set $F \subseteq S$ with $\sum_{f \in F}f=t$. Then, $B\coloneqq \Set{x_f}{f \in F}\cup \set{y_1,\ldots,y_{k-|F|}, z}$ is a zero basis of $M$.
    Now suppose that $M$ contains a zero basis $B$. As $|F|>|Y|$, we obtain that $B$ contains at least one element of $X \cup \set{z}$. Hence, if $z$ is not contained $B$, we obtain $\psi(B)=\psi(B \cap X)>0$, a contradiction. We hence have $z \in B$. Let $F\coloneqq \Set{s \in S}{x_s \in B}$.
    We obtain $\sum_{s \in F} s =\sum_{x \in X \cap B}\psi(x)=\psi(B)-\psi(z)=-\psi(z)=t$. Hence $(S,t)$ is a positive instance of \textsc{Subset Sum}.
\end{proof}

It can be derived from \cite[Theorem~1.4]{doronarad2023tight} that \textsc{Zero Basis} is also hard for finite groups given as operation tables.

\begin{theorem}[see Doron-Arad--Kulik--Shachnai {\cite[Theorem~1.4]{doronarad2023tight}}]\label{thm:zero-base-table}
    \textscup{Zero Basis} requires an exponential number of independence queries for a finite group $\Gamma$ given as an operation table.
\end{theorem}

We next show that \textsc{Zero Basis} is intractable for any fixed infinite group.
To this end, we use the following structural theorem on abelian groups with finite exponents.
Here, the \emph{exponent} of a group means the least common multiple of the orders of the elements.

\begin{theorem}[{first Prüfer theorem; see~\cite[Theorem~5.2]{fuchs2015abelian}}]\label{thm:first-prufer}
    Any abelian group with a finite exponent is isomorphic to a direct sum of cyclic groups.
\end{theorem}

\begin{lemma}\label{lem:infinite-group-set}
    Let $\Gamma$ be an infinite abelian group and $n \in \N$.
    Then, there exists a subset $G \subseteq \Gamma$ of cardinality $n$ such that $\sum_{g \in H} s_g g \ne 0$ for any nonempty $H \subseteq G$ and $s_g \in \set{-1, +1}$ for $g \in H$.
\end{lemma}

\begin{proof}
    First, suppose that $\Gamma$ has an infinite exponent.
    Then, we can take an element $g \in \Gamma$ with the order at least $2^n$.
    Thus, $G = \set{2^0g, 2^1g, \dotsc, 2^{n-1}g}$ satisfies the desired property as any signed sum of a nonempty subset of $G$ is in the form of $ag$ for some non-zero integer $a$ with $|a| \le 2^n - 1$.
    If $\Gamma$ has a finite exponent, it is decomposed as $\Gamma \simeq \bigoplus_{\lambda \in \Lambda} \Z_{m(\lambda)}$ by \cref{thm:first-prufer}, where $m$ is a map from an index set $\Lambda$ to an integer greater than one.
    Specifically, $|\Lambda| = \infty$ follows from $|\Gamma| = \infty$.
    Taking $n$ distinct indices $\lambda_1, \dotsc, \lambda_n \in \Lambda$, define $G = \set{g_1, \dotsc, g_n} \subseteq \Gamma$ as follows: each $g_i$ is an element of $\Gamma$, identified with an (infinite) tuple of elements of cyclic groups indexed by $\Lambda$, such that the $\lambda_i$th component is $1 \in \Z_{m(\lambda_i)}$ and the others are zero.
    Then, $G$ generates a subgroup of $\Gamma$ isomorphic to $\bigoplus_{i=1}^n \Z_{m(\lambda_i)}$ and any signed sum of a nonempty subset of $G$ cannot be zero.
\end{proof}

Our hardness proof for \textsc{Zero Basis} with an infinite group is based on a reduction from \emph{matroid parity}, which is the following problem: given a matroid whose ground set is partitioned into pairs, we find a basis consisting of pairs.
Matroid parity is known to be intractable for general matroids.

\begin{proposition}[{\cite{jensen1982complexity,lovasz1980matroid}}]\label{prop:parity-hard}
    Solving matroid parity requires an exponential number of independence oracle queries.
\end{proposition}

\begin{theorem}\label{thm:zero-basis-infinite}
    \textscup{Zero Basis} requires exponentially many independence oracle queries for any fixed infinite group $\Gamma$. 
\end{theorem}

\begin{proof}
    We show the claim by reducing matroid parity to \textsc{Zero Basis} with $\Gamma$.
    Let $M$ be a matroid of $2n$ elements and $\set{a_1, b_1}, \dotsc, \set{a_n, b_n}$ a partition of $E(M)$ into pairs.
    Let $G = \set{g_1, \dotsc, g_n} \subseteq \Gamma$ be a set of group elements provided by \cref{lem:infinite-group-set}.
    We construct a labeling $\psi\colon E(M) \to \Gamma$ as $\psi(a_i) = g_i$ and $\psi(b_i) = -g_i$ for $i \in [n]$.
    Then, a basis $B$ of $M$ is zero if and only if $B$ consists of pairs, proving the claim from \cref{prop:parity-hard}.
\end{proof}

Recall that \textsc{(Weighted) Zero Basis} is solvable if $\Gamma$ is a fixed, finite group~\cite{liu2023congruency}. In contrast, \cref{thm:non-zero-common-basis-hard} implies that \textsc{Zero Common Basis} is hard for any fixed group $\Gamma$ such that $\Z_2 \le \Gamma$. By modifying that proof, the hardness of \textsc{Zero Common Basis} follows even when the assumption $\Z_2 \le \Gamma$ is dropped. 

\begin{restatable}{theorem}{zerocommonbasefixed}\label{thm:zero-common-base-fixed}
    \textscup{Zero Common Basis} requires an exponential number of independence queries for any nontrivial fixed group $\Gamma$.
\end{restatable}

\begin{proof}  
    Let $r \ge 2$, $A=\set{a_1,\dots, a_r}$ and $B=\set{b_1,\dots, b_r}$ be disjoint sets of size $r$, and $E\coloneqq A \cup B$. 
    Let $g \in \Gamma$ be an arbitrary non-zero element.
    Consider the labeling $\psi\colon E\to \Gamma$ defined by $\psi(a_1)=\psi(a_2)= \cdots =\psi(a_{\lfloor r/2 \rfloor})=g$, $\psi(a_{\lfloor r/2 \rfloor+1})=\psi(a_{\lfloor r/2 \rfloor+2})= \cdots = \psi(a_r)=-g$ and $\psi(b_1)=\psi(b_2)=\cdots=\psi(b_r)=0$. Let 
    \[
    \cH \coloneqq \Set*{H \subseteq E}{|H \cap \set{a_1,b_1}| = \dots = |H\cap \set{a_r, b_r}| = 1, \psi(H)=0}. \]
    Observe that $|H_1 \cap H_2| \leq r-2$ for all $H_1, H_2 \in \cH, H_1 \neq H_2$, as otherwise $H_1 \symdif H_2=\set{a_i, b_i}$ for some $1 \leq i \leq r$ and $\psi(a_i) \neq \psi(b_i)$. Hence, by \cref{lem:sp}, $\binom{E}{r}\setminus \cH$ forms the family of bases of a sparse paving matroid $M$. Let $N$ denote the unitary partition matroid defined by partition classes $\set{a_1, b_1}, \dots, \set{a_r, b_r}$. The common bases of $M$ and $N$ are exactly the sets containing one element from each of the pairs $\set{a_1, b_1}, \dots, \set{a_r, b_r}$ and having a non-zero label. For any $H \in \cH$, by \cref{lem:sp}, $\binom{E}{r}\setminus (\cH\setminus \set{H})$ forms the family of bases of sparse paving matroid $M_H$. Since $H$ is a common basis of $M_H$ and $N$ having label zero, any algorithm for \textsc{Zero Common Basis} must distinguish between the instance $M, N, \psi$ and the instances $M_H, N, \psi$ for $H \in \cH$. Since the only way to distinguish $M$ and an $M_H$ is to query the independence of $H$, any algorithm solving \textsc{Zero Common Basis} must query the independence of each $H \in \cH$, and thus uses at least $|\cH|$ independence queries. Note that for any $K \subseteq \set{1, 2, \ldots, {\lfloor r/2 \rfloor}}$ and $L \subseteq \set{\lfloor r/2 \rfloor+1, \lfloor r/2 \rfloor+2, \ldots , r}$ with $|K|=|L|$ the set \[\Set{a_i}{i \in K \cup L} \cup \Set{b_i}{i \in [r] \setminus (K \cup L)}\] 
    is in $\cH$, hence $|\cH|$ is indeed exponential.
\end{proof}

\subsection{Counterexample to a Conjecture of Liu and Xu}\label{sec:counter}

Liu and Xu~\cite{liu2023congruency} proposed a conjecture which is even stronger than the implications from Conjectures~\ref{conj:k} and \ref{conj:fk}. In order to state their conjecture, we need the following definition. For a finite abelian group $\Gamma$ its \emph{Davenport constant} $D(\Gamma)$ is defined as the minimum value such that every sequence of elements from $\Gamma$ of length $D(\Gamma)$ contains a nonempty subsequence with sum 0. Liu and Xu proposed the following conjecture.

\begin{conjecture}[Liu--Xu~\cite{liu2023congruency}]\label{conj:xuliu}
    Let $\Gamma$  be a finite abelian group. Then, $\Gamma$ is $(D(\Gamma)-1)$-close.
\end{conjecture}

We provide a counterexample for \cref{conj:xuliu}. More precisely, we prove the following result.

\begin{restatable}{theorem}{thmcounter}\label{thmcounter} 
    Let $\Gamma=\Z_2^d$ for some $d \geq 4$. Then, $\Gamma$ is not $(D(\Gamma)-1)$-close.
\end{restatable}

We first need the following result from Olson \cite{olson1969combinatorial} on the Davenport constant.
\begin{proposition}[Olson~\cite{olson1969combinatorial}] \label{prop:olson} 
    Let $d \geq 1$ be an integer. Then, $D(\Z_2^d)=d+1$.
\end{proposition}

We now prove a result on zero bases with respect to a labeling to $\Z_2^d$ in certain matroids.
\begin{lemma} \label{prop:counterexample}
    For any positive integer $d$, there exists a matroid $M$, a labeling $\psi\colon E(M) \to \Z_2^d$ and a basis $B$ such that $M$ has a zero basis and $|B\setminus B^*| \ge 2^{d-1}-1$ holds for each zero basis $B^*$. 
\end{lemma}
\begin{proof}
    Consider a ground set $E$ of size $2^d$, and label each element of $E$ with a distinct label from $\Z_2^d$. 
    Let $r = 2^{d-1}-1$ and define $\mathcal{H} = \Set{H \subseteq E}{|H|=r, \psi(H)=0}$. Note that $\cH$ is nonempty, e.g., it contains $\Set{(g_1,\dots, g_{d-1}, 0)}{g_1,\dots, g_{d-1} \in \Z_2}$. 
    Since $\psi$ takes distinct values on the elements of $E$, $|H \cap H'| \le r-2$ holds for any $H, H' \in \cH$ with $H \ne H'$. 
    Let $B^*$ be any member of $\cH$.
    By \cref{lem:sp}, there exists a sparse paving matroid $M$ on ground set $E$ with family of bases $\binom{E}{r}\setminus (\cH \setminus \set{B^*})$.
    Then, $B^*$ is the only zero basis of $M$ by the definition of $\cH$. 
    Since $|E\setminus B^*| = 2^{d-1}+1 > r$ and  $|H \cap H'| \le r-2$ holds for $H, H' \in \cH$ with $H \ne H'$, there exists a subset $B \subseteq E \setminus B^*$ such that $|B| = r$ and $B \not \in \cH$.
    Then, $B$ is a basis of $M$ such that $|B \setminus B^*| = r = 2^{d-1}-1$ holds for the only zero basis $B^*$ of $M$.
\end{proof}

\cref{prop:counterexample} implies that the group $\Z_2^d$ is not $(2^{d-1}-2)$-close for any $d \geq 1$. For any $d \geq 4$, we have $2^{d-1}-2>d+1$ and hence \cref{thmcounter} follows by \cref{prop:olson}.

\section{Conclusion}\label{sec:conclusion}

In this work, we have treated several problem settings on finding bases of group-labeled matroids whose labels satisfy certain conditions. Many questions remain open. 

In \cref{subsec:non-zero-common-basis}, we deal with  \textsc{Weighted Non-Zero Common Basis} for groups $\Gamma$ with $\Z_2 \not\le \Gamma$ and give an approximation algorithm and exact algorithms for some special cases. However, the general complexity of \textsc{Weighted Non-Zero Common Basis} for $\Z_2 \not\le \Gamma$ remains open. In the setting where $\Z_2 \le \Gamma$ is allowed, it would be interesting to see whether \cref{thm:non-zero-exact-base} can be extended to the case that one matroid is a partition matroid with an arbitrary fixed number of partition classes.

In \cref{sec:algebraic}, randomized algebraic algorithms turn out to be a powerful tool for finding bases and common bases of certain labels. It would be interesting to see whether more of the problems that can be solved by these randomized algorithms can also be solved deterministically. For example, one could consider \textsc{Non-Zero Common Basis} for arbitrary groups when one of the matroids is graphic and the other one is a partition matroid.

In \cref{sec:fixedF}, we studied Conjectures~\ref{conj:k} and \ref{conj:fk} related to \textsc{$F$-avoiding Basis}. Recall that \cref{conj:fk} implies that \textsc{$F$-avoiding Basis} can be solved in polynomial time if $|F|$ is fixed. As a quest for an even stronger algorithmic result, we can ask whether \textsc{$F$-avoiding Basis} is FPT with respect to $|F|$, i.e., whether it can be solved in time $f(|F|)n^{O(1)}$ for some computable function $f\colon \mathbb{Z}_{\geq 0}\rightarrow \mathbb{Z}_{\geq 0}$. The problem \textsc{Weighted $F$-avoiding Basis} is also of interest. While Conjectures \ref{conj:k} and \ref{conj:fk} remain wide open, the following stronger conjecture can be formulated analogously to the notion of strong $k$-closeness introduced by Liu and Xu~\cite{liu2023congruency}. 
\begin{conjecture} \label{conj:k_weighted} Let $M$ be a matroid on a ground set $E$, $\psi\colon E \to \Gamma$  a group labeling, $F \subseteq \Gamma$ a finite subset, and $w\colon E \to \R$ a weight function.  Suppose that $M$ has an $F$-avoiding basis. Then, for any minimum weight basis $B$, there exists a minimum weight $F$-avoiding basis $B^*$ such that $|B \setminus B^*| \le |F|$.
\end{conjecture}
Note that \cref{conj:k_weighted} would imply the polynomial solvability of \textsc{Weighted $F$-avoiding Basis} if $|F|$ is fixed. The conjecture holds for $|F|=1$ by \cref{lem:weighted-non-zero-base}, and it can also be shown that it holds for strongly base orderable matroids. 

\paragraph{Acknowledgments}
The authors are grateful to the organizers of the 14th Emléktábla Workshop where the collaboration of the authors started. The authors thank Naonori Kakimura, Kevin Long, and Tomohiko Yokoyama for several discussions during the workshop, and Kristóf Bérczi, András Frank, and András Sebő for pointing out the connections to lattices and dual lattices. 

András Imolay was supported by the Rényi Doctoral Fellowship of the Rényi Institute. Ryuhei Mizutani was supported by Grant-in-Aid for JSPS Fellows Grant Number JP23KJ0379 and JST SPRING Grant Number JPMJSP2108. Taihei Oki was supported by JSPS KAKENHI Grant Number JP22K17853 and JST ERATO Grant Number JPMJER1903.
Tamás Schwarcz was supported by the \'{U}NKP-23-3 New National Excellence Program of the Ministry for Culture and Innovation from the source of the National Research, Development and Innovation Fund. This research has been implemented with the support provided by the Lend\"ulet Programme of the Hungarian Academy of Sciences -- grant number LP2021-1/2021.

\bibliographystyle{plainurl}
\bibliography{biblio.bib}

\begin{thebibliography}{10}

\bibitem{Artmann2017}
S.~Artmann, R.~Weismantel, and R.~Zenklusen.
\newblock A strongly polynomial algorithm for bimodular integer linear programming.
\newblock In {\em Proceedings of the 49th Annual ACM SIGACT Symposium on Theory of Computing (STOC '17)}. ACM, 2017.

\bibitem{BangJensen2009digraphs}
J.~Bang-Jensen and G.~Gutin.
\newblock {\em Digraphs}.
\newblock Springer, London, 2009.

\bibitem{barahona1987exact}
F.~Barahona and W.R. Pulleyblank.
\newblock Exact arborescences, matchings and cycles.
\newblock {\em Discrete Applied Mathematics}, 16(2):91--99, 1987.

\bibitem{baumgart2009ranking}
M.~Baumgart.
\newblock {\em Ranking and ordering problems of spanning trees}.
\newblock PhD thesis, Technische Universit{\"a}t M{\"u}nchen, 2009.

\bibitem{berczi2023reconfiguration}
K.~B{\'e}rczi, B.~M{\'a}trav{\"o}lgyi, and T.~Schwarcz.
\newblock Reconfiguration of basis pairs in regular matroids.
\newblock {\em arXiv preprint arXiv:2311.07130}, 2023.

\bibitem{bhatnagar}
N.~Bhatnagar, D.~Randall, V.~Vazirani, and E.~Vigoda.
\newblock Random bichromatic matchings.
\newblock {\em Algorithmica}, 50:418--445, 2008.

\bibitem{bonin2016infinite}
J.E. Bonin and T.J. Savitsky.
\newblock An infinite family of excluded minors for strong base-orderability.
\newblock {\em Linear Algebra and its Applications}, 488(1):396--429, 2016.

\bibitem{bouchet1987greedy}
A.~Bouchet.
\newblock Greedy algorithm and symmetric matroids.
\newblock {\em Mathematical Programming}, 38(2):147--159, 1987.

\bibitem{brezovec1986two}
C.~Brezovec, G.~Cornu{\'{e}}jols, and F.~Glover.
\newblock Two algorithms for weighted matroid intersection.
\newblock {\em Mathematical Programming}, 36(1):39--53, 1986.

\bibitem{brualdi1969comments}
R.A. Brualdi.
\newblock Comments on bases in dependence structures.
\newblock {\em Bulletin of the Australian Mathematical Society}, 1(2):161--167, 1969.

\bibitem{berczi2023complexity}
K.~Bérczi, G.~Csáji, and T.~Király.
\newblock On the complexity of packing rainbow spanning trees.
\newblock {\em Discrete Mathematics}, 346(4):113297, 2023.

\bibitem{berczi2021complexity}
K.~Bérczi and T.~Schwarcz.
\newblock Complexity of packing common bases in matroids.
\newblock {\em Mathematical Programming}, 188(1):1--18, 2021.

\bibitem{camerini1992exact}
P.M. Camerini, G.~Galbiati, and F.~Maffioli.
\newblock Random pseudo-polynomial algorithms for exact matroid problems.
\newblock {\em Journal of Algorithms}, 13(2):258--273, 1992.

\bibitem{chudnovsky2006apath}
M.~Chudnovsky, J.~Geelen, B.~Gerards, L.~Goddyn, M.~Lohman, and P.D. Seymour.
\newblock {Packing non-zero $A$-paths in group-labelled graphs}.
\newblock {\em Combinatorica}, 26(5):521--532, 2006.

\bibitem{cunningham1984testing}
W.H. Cunningham.
\newblock Testing membership in matroid polyhedra.
\newblock {\em Journal of Combinatorial Theory, Series B}, 36(2):161--188, 1984.

\bibitem{devos2009generalization}
M.~DeVos, L.~Goddyn, and B.~Mohar.
\newblock A generalization of {K}neser’s {A}ddition {T}heorem.
\newblock {\em Advances in Mathematics}, 220(5):1531--1548, 2009.

\bibitem{diestel2017graph}
R.~Diestel.
\newblock {\em Graph Theory}, volume 173 of {\em Graduate Texts in Mathematics}.
\newblock Springer, Berlin, fifth edition, 2017.

\bibitem{ding1996unavoidable}
G.~Ding, B.~Oporowski, J.~Oxley, and D.~Vertigan.
\newblock Unavoidable minors of large 3-connected binary matroids.
\newblock {\em Journal of Combinatorial Theory, Series B}, 66(2):334--360, 1996.

\bibitem{doronarad2023tight}
I.~Doron-Arad, A.~Kulik, and H.~Shachnai.
\newblock Tight lower bounds for weighted matroid problems.
\newblock {\em arXiv preprint arXiv:2307.07773}, 2023.

\bibitem{EDMONDS197939}
J.~Edmonds.
\newblock Matroid intersection.
\newblock In P.L. Hammer, E.L. Johnson, and B.H. Korte, editors, {\em Discrete Optimization I}, volume~4 of {\em Annals of Discrete Mathematics}, pages 39--49. Elsevier, 1979.

\bibitem{efimov2021determinant}
D.~Efimov.
\newblock Determinant of three-layer {T}oeplitz matrices.
\newblock {\em Journal of Integer Sequences}, 24(2):3, 2021.

\bibitem{elmaalouly2023exact}
N.~El~Maalouly, R.~Steiner, and L.~Wulf.
\newblock {Exact matching: correct parity and FPT parameterized by independence number}.
\newblock In {\em Proceedings of the 34th International Symposium on Algorithms and Computation (ISAAC '23)}, volume 283 of {\em Leibniz International Proceedings in Informatics}, pages 28:1--28:18, Dagstuhl, 2023.

\bibitem{frank1981weighted}
A.~Frank.
\newblock A weighted matroid intersection algorithm.
\newblock {\em Journal of Algorithms}, 2(4):328--336, 1981.

\bibitem{frank2011connections}
A.~Frank.
\newblock {\em Connections in Combinatorial Optimization}, volume~38 of {\em Oxford Lecture Series in Mathematics and its Applications}.
\newblock Oxford University Press, Oxford, 2011.

\bibitem{fuchs2015abelian}
L.~Fuchs.
\newblock {\em Abelian Groups}.
\newblock Springer Monographs in Mathematics. Springer, Cham, 2015.

\bibitem{fujishige1977primal}
S.~Fujishige.
\newblock A primal approach to the independent assignment problem.
\newblock {\em Journal of the Operations Research Society of Japan}, 20(1):1--15, 1977.

\bibitem{gabow1976decomposing}
H.~Gabow.
\newblock Decomposing symmetric exchanges in matroid bases.
\newblock {\em Mathematical Programming}, 10(1):271--276, 1976.

\bibitem{Galluccio1998}
A.~Galluccio and M.~Loebl.
\newblock On the theory of pfaffian orientations. i. perfect matchings and permanents.
\newblock {\em The Electronic Journal of Combinatorics}, 6(1), 1998.

\bibitem{harary1972evolution}
F.~Harary and A.J. Schwenk.
\newblock Evolution of the path number of a graph: covering and packing in graphs, {II}.
\newblock In R.C. Read, editor, {\em Graph Theory and Computing}, pages 39--45. Academic Press, New York, 1972.

\bibitem{harvey2011disjoint}
N.J.A. Harvey, T.~Király, and L.C. Lau.
\newblock On disjoint common bases in two matroids.
\newblock {\em SIAM Journal on Discrete Mathematics}, 25(4):1792--1803, 2011.

\bibitem{horsch2022rainbow}
F.~H{\"o}rsch, T.~Kaiser, and M.~Kriesell.
\newblock Rainbow bases in matroids.
\newblock {\em arXiv preprint arXiv:2206.10322}, 2022.

\bibitem{iwata2022finding}
Y.~Iwata and Y.~Yamaguchi.
\newblock Finding a shortest non-zero path in group-labeled graphs.
\newblock {\em Combinatorica}, 42(S2):1253--1282, 2022.

\bibitem{jensen1982complexity}
P.M. Jensen and B.~Korte.
\newblock Complexity of matroid property algorithms.
\newblock {\em SIAM Journal on Computing}, 11(1):184--190, 1982.

\bibitem{jia2023exact}
X.~Jia, O.~Svensson, and W.~Yuan.
\newblock The exact bipartite matching polytope has exponential extension complexity.
\newblock In {\em Proceedings of the 34th Annual ACM-SIAM Symposium on Discrete Algorithms (SODA '23)}, pages 1635--1654. SIAM, 2023.

\bibitem{kaltofen2005determinant}
E.~Kaltofen and G.~Villard.
\newblock On the complexity of computing determinants.
\newblock {\em Computational Complexity}, 13(3-4):91--130, 2005.

\bibitem{kano1987maximum}
M.~Kano.
\newblock Maximum and $k$-th maximal spanning trees of a weighted graph.
\newblock {\em Combinatorica}, 7(2):205--214, 1987.

\bibitem{kawamoto1978second}
T.~Kawamoto, Y.~Kajitani, and S.~Shinoda.
\newblock On the second maximal spanning trees of a weighted graph (in {J}apanese).
\newblock {\em Transactions of the IECE of Japan}, 61-A:988--995, 1978.

\bibitem{kawase2020twoforbidpath}
Y.~Kawase, Y.~Kobayashi, and Y.~Yamaguchi.
\newblock Finding a path with two labels forbidden in group-labeled graphs.
\newblock {\em Journal of Combinatorial Theory, Series B}, 143:65--122, 2020.

\bibitem{kim2023gamma}
D.~Kim, D.~Lee, and S.~Oum.
\newblock {$\Gamma$}-graphic delta-matroids and their applications.
\newblock {\em Combinatorica}, 43(5):963--983, 2023.

\bibitem{kleinberg2006}
J.~Kleinberg and E.~Tardos.
\newblock {\em Algorithm Design}.
\newblock Pearson Deutschland, 2006.

\bibitem{kobayashi2023rcb}
Y.~Kobayashi, R.~Mahara, and T.~Schwarcz.
\newblock Reconfiguration of the union of arborescences.
\newblock {\em arXiv preprint arXiv:2304.13217}, 2023.

\bibitem{krogdahl1974combinatorial}
S.~Krogdahl.
\newblock A combinatorial base for some optimal matroid intersection algorithms.
\newblock Technical Report STAN-CS-74-468, Computer Science Department, Stanford University, Stanford, 1974.

\bibitem{krogdahl1976combinatorial}
S.~Krogdahl.
\newblock A combinatorial proof for a weighted matroid intersection algorithm.
\newblock Technical Report Computer Science Report 17, Institute of Mathematical and Physical Sciences, University of Tromso, Tromso, 1976.

\bibitem{Krogdahl1977dependence}
S.~Krogdahl.
\newblock The dependence graph for bases in matroids.
\newblock {\em Discrete Mathematics}, 19(1):47--59, 1977.

\bibitem{lemos2006weight}
M.~Lemos.
\newblock Weight distribution of the bases of a matroid.
\newblock {\em Graphs and Combinatorics}, 22(1):69--82, 2006.

\bibitem{liu2023congruency}
S.~Liu and C.~Xu.
\newblock On the congruency-constrained matroid base.
\newblock In {\em Proceedings of the 25th Conference on Integer Programming and Combinatorial Optimization (IPCO '24)}, 2024.
\newblock To appear.

\bibitem{lovasz1979rand}
L.~Lov{\'a}sz.
\newblock On determinants, matchings, and random algorithms.
\newblock In Lothar Budach, editor, {\em Fundamentals of Computation Theory, {FCT} '79, Proceedings of the Conference on Algebraic, Arthmetic, and Categorial Methods in Computation Theory}, pages 565--574. Akademie-Verlag, Berlin, 1979.

\bibitem{lovasz1980matroid}
L.~Lov{\'a}sz.
\newblock Matroid matching and some applications.
\newblock {\em Journal of Combinatorial Theory. Series B}, 28(2):208--236, 1980.

\bibitem{lovasz1985some}
L.~Lov{\'a}sz.
\newblock Some algorithmic problems on lattices.
\newblock In L.~Lov{\'a}sz and E.~Szemer{\'e}di, editors, {\em Theory of algorithms}, volume~44 of {\em Colloquia Mathematica Societatis J{\'a}nos Bolyai}, pages 323--337. North-Holland, Amsterdam, 1985.

\bibitem{lovasz1987matching}
L.~Lov{\'a}sz.
\newblock Matching structure and the matching lattice.
\newblock {\em Journal of Combinatorial Theory, Series B}, 43(2):187--222, 1987.

\bibitem{matoya2022pfaffian}
K.~Matoya and T.~Oki.
\newblock Pfaffian pairs and parities: counting on linear matroid intersection and parity problems.
\newblock {\em SIAM Journal on Discrete Mathematics}, pages 2121--2158, 2022.

\bibitem{Mayhew2008nine}
D.~Mayhew and G.F. Royle.
\newblock Matroids with nine elements.
\newblock {\em Journal of Combinatorial Theory, Series B}, 98(2):415--431, 2008.

\bibitem{mayr1989spanning}
E.W. Mayr and C.G. Plaxton.
\newblock On the spanning trees of weighted graphs.
\newblock In Jan van Leeuwen, editor, {\em Graph-Theoretic Concepts in Computer Science}, pages 394--405, Berlin, 1989. Springer.

\bibitem{mulmuley1987matching}
K.~Mulmuley, U.V. Vazirani, and V.V. Vazirani.
\newblock Matching is as easy as matrix inversion.
\newblock {\em Combinatorica}, 7(1):105--113, 1987.

\bibitem{nagele2023ccc}
M.~N{\"a}gele, R.~Santiago, and R.~Zenklusen.
\newblock Congruency-constrained {TU} problems beyond the bimodular case.
\newblock {\em Mathematics of Operations Research}, 2023.

\bibitem{nagele2019submodular}
M.~N{\"a}gele, B.~Sudakov, and R.~Zenklusen.
\newblock Submodular minimization under congruency constraints.
\newblock {\em Combinatorica}, 39(6):1351--1386, 2019.

\bibitem{nagele2020new}
M.~N{\"a}gele and R.~Zenklusen.
\newblock A new contraction technique with applications to congruency-constrained cuts.
\newblock {\em Mathematical Programming}, 183(1):455--481, 2020.

\bibitem{olson1969combinatorial}
J.E. Olson.
\newblock A combinatorial problem on finite abelian groups, {I}.
\newblock {\em Journal of Number Theory}, 1(1):8--10, 1969.

\bibitem{oxley2011matroid}
J.~Oxley.
\newblock {\em Matroid Theory}, volume~21 of {\em Oxford Graduate Texts in Mathematics}.
\newblock Oxford University Press, Oxford, second edition, 2011.

\bibitem{papadimitriou1982exact}
C.H. Papadimitriou and M.~Yannakakis.
\newblock The complexity of restricted spanning tree problems.
\newblock {\em Journal of the ACM}, 29(2):285--309, 1982.

\bibitem{rieder1991lattices}
J.~Rieder.
\newblock The lattices of matroid bases and exact matroid bases.
\newblock {\em Archiv der Mathematik}, 56(6):616--623, 1991.

\bibitem{schrijver2003combinatorial}
A.~Schrijver.
\newblock {\em Combinatorial Optimization: Polyhedra and Efficiency}.
\newblock Springer, Berlin, 2003.

\bibitem{schrijver1990spanning}
A.~Schrijver and P.D. Seymour.
\newblock Spanning trees of different weights.
\newblock In W.J. Cook and P.D. Seymour, editors, {\em Polyhedral Combinatorics}, volume~1 of {\em DIMACS Series in Discrete Mathematics and Theoretical Computer Science}, pages 281--288. AMS-ACM, 1990.

\bibitem{schwartz1980pit}
J.T. Schwartz.
\newblock Fast probabilistic algorithms for verification of polynomial identities.
\newblock {\em Journal of the ACM}, 27(4):701--717, 1980.

\bibitem{svensson2017}
O.~Svensson and J.~Tarnawski.
\newblock The matching problem in general graphs is in quasi-nc.
\newblock In {\em 2017 IEEE 58th Annual Symposium on Foundations of Computer Science (FOCS)}. IEEE, 2017.

\bibitem{ito2023reconfiguring}
I.~Takehiro, I.~Yuni, K.~Yasuaki, N.~Yu, O.~Yota, and W.~Kunihiro.
\newblock Reconfiguring (non-spanning) arborescences.
\newblock {\em Theoretical Computer Science}, 943:131--141, 2023.

\bibitem{tomizawa1974axb}
N.~Tomizawa and M.~Iri.
\newblock {An algorithm for determining the rank of a triple matrix product $AXB$ with application to the problem of discerning the existence of the unique solution in a network (in Japanese)}.
\newblock {\em Electronics and Communications in Japan}, 57(11):50--57, 1974.

\bibitem{vondergathen2013modern}
J.~von~zur Gathen and J.~Gerhard.
\newblock {\em Modern Computer Algebra}.
\newblock Cambridge University Press, Cambridge, third edition, 2013.

\bibitem{webb2004paffian}
K.~Webb.
\newblock {\em Counting Bases}.
\newblock PhD thesis, University of Waterloo, Waterloo, 2004.

\bibitem{yuster}
R.~Yuster.
\newblock Almost exact matchings.
\newblock {\em Algorithmica}, 63:39--50, 2012.

\bibitem{zippel1979pit}
R.~Zippel.
\newblock Probabilistic algorithms for sparse polynomials.
\newblock In E.W. Ng, editor, {\em Symbolic and Algebraic Computation}, volume~72 of {\em Lecture Notes in Computer Science}, pages 216--226. Springer, Berlin, 1979.

\end{thebibliography}

\clearpage
\appendix
\section{Appendix} \label{sec:appendix}

This appendix contains the omitted proof of \cref{thm:graphalg}, that is, we give an algorithm for finding a shortest non-zero directed cycle in a digraph. 

First, we show how to compute shortest non-zero walks. Our algorithm is a modification of Dijkstra's shortest path algorithm. In order to simplify notation, we suppose in the following that the digraphs in consideration do not contain parallel arcs. The proofs can easily be adapted for graphs with parallel arcs.

We first introduce some notation. Given a digraph $D$, a {\it walk} in a digraph $D$ is a sequence $P=v_1,\ldots,v_t$ of vertices of $D$ such that $v_iv_{i+1}\in A(D)$ for all $i \in [t-1]$. We also say that $P$ is a {\it $v_1v_t$-walk.} If $v_1=v_t$, we say that $P$ is a {\it closed} walk. We further define $A(P)=\{v_1v_2,\ldots,v_{t-1}v_t\}$. For some $i \in [t]$, we say that $v_1,\ldots,v_i$ is a {\it prefix} of $P$.

 \begin{theorem} \label{thm:walks}
    Let $D$ be a digraph, $s \in V(D)$ a vertex, $\psi\colon A(D) \to \Gamma$ a group labeling, and $w\colon A(D) \to \R$ a conservative weight function. There is a polynomial-time algorithm that, for each vertex $v$, computes a shortest non-zero $sv$-walk or correctly reports that no such walk exists. 
\end{theorem}
\begin{proof}
    We may assume that each vertex is reachable from $s$. Moreover, we may also assume that $w$ is nonnegative on each arc, as we can compute a potential $p$ in polynomial time by the conservativity of $w$ and increase the weight of each arc $uv$ by $p(u)-p(v)$, see \cite[Section 3.1.1]{frank2011connections} for details. This increases the weight of each $sv$-walk by $p(s)-p(v)$, thus does not change whether a non-zero $sv$-walk is the shortest.

    We compute a shortest paths tree $T$ rooted at $s$ with respect to the weight function $w$, that is, a spanning arborescence with root $s$ such that the unique $sv$-path $T_v$ in $T$ is a shortest $sv$-path in $D$. We call an $sv$-walk $P$ a \emph{distinct-label $sv$-walk} if $\psi(P) \ne \psi(T_v)$. Observe that it is sufficient to compute a shortest distinct-label $sv$-walk for each vertex $v$. Indeed, if $\psi(T_v)\ne0$, then $T_v$ is a shortest non-zero $sv$-walk, otherwise distinct-label $sv$-walks and non-zero $sv$-walks coincide by definition.

    In what follows we give an algorithm that computes a shortest distinct-label $sv$-walk or correctly reports that no such walk exists. The crucial insight is that any prefix of a shortest distinct-label walk is either a shortest walk or a shortest distinct-label walk. 
    We will maintain and gradually increase a set $S \subseteq V(D)$ such that for all $v \in S$, we know a shortest distinct-label $sv$-walk $Q_v$. We start with $S=\emptyset$ and then increase $S$ through several iterations. We now describe one of the iterations. We refer to $S$ as the set assigned before this iteration. For each arc $uv\in A(D)$, we set 
      \[\alpha(uv) \coloneqq \begin{cases}
        w(T_u)+w(uv) & \text{if $\psi(T_u)+\psi(uv) \ne \psi(T_v)$}, \\ 
        w(Q_u)+w(uv) & \text{if $u \in S$ and $\psi(T_u)+\psi(uv) = \psi(T_v)$}, \\ 
        \infty & \text{otherwise}. \end{cases}\]
    Let $\alpha(v) \coloneqq \min\Set{\alpha(uv)}{uv \in A(D)}$ for each $v \in V(D)\setminus S$.  If $\alpha(v) = \infty$ for each $v \in V(D)\setminus S$, we stop the algorithm and report that $D$ does not contain a distinct-label $sv$-walk for any $v \in V(D) \setminus S$. Otherwise, take $v^* \in \argmin \Set{\alpha(v)}{v \in V(D) \setminus S}$.
    Let $u^*v^*$ be an arc entering $v^*$ such that $\alpha(v^*) = \alpha(u^*v^*)$ and either $\psi(T_{u^*})+\psi(u^*v^*) \ne \psi(T_{v^*})$, or $u \in S$ and $\psi(Q_{u^*})+\psi(u^*v^*) = \psi(T_{v^*})$. Let $Q_{v^*}$ be obtained by concatenating  $T_{u^*}$ in the former case and $Q_{u^*}$ in the latter case with the arc $u^*v^*$. Finally, we let $S \coloneqq S\cup \set{v^*}$.

    We prove the correctness of the algorithm. Assume that at the start of an iteration $Q_v$ is a shortest distinct-label $sv$-walk for each $v \in S$. The next claim shows that we correctly terminate the algorithm if $\alpha(v) = \infty$ for each $v \in V(D) \setminus S$. 
    
    \begin{claim} \label{cl:end}
         If $\alpha(v)=\infty$ for each $v \in V(D)\setminus S$, then $D$ does not contain a distinct-label $sv$-walk for any $v \in V(D)\setminus S$.
    \end{claim}
    \begin{claimproof}
        Suppose otherwise and let $v \in V(D)\setminus S$ such that $D$ contains a distinct-label $sv$-walk $P_v$. Let $w$ be the first vertex on $P_v$ such that $w \in V(D) \setminus S$ and $\psi(P_w) \ne \psi(T_w)$ for some  $sw$-prefix $P_w$ of $P_v$. Let $uw$ be the last arc of $P_w$, and $P_u$ the subwalk of $P_w$ obtained by deleting $w$. Then, either $u \in S$, or $\psi(P_u) = \psi(T_u)$ and thus $\psi(T_u)+\psi(uw) = \psi(P_u)+\psi(uw) = \psi(P_w) \ne \psi(T_w)$. In both cases, $\alpha(uw) < \infty$, hence $\alpha(w) < \infty$, a contradiction.   
    \end{claimproof}

    The next claim shows that if there is a vertex $v\in V(D) \setminus S$ with $\alpha(v) < \infty$, then $Q_{v^*}$ is a shortest distinct-label $sv^*$-walk.
    
\begin{claim} \label{cl:add}
        $Q_{v^*}$ is a shortest distinct-label $sv^*$-walk.
    \end{claim}
    \begin{claimproof}
        It follows by  definition that $Q_{v^*}$ is a distinct-label $sv^*$-walk.
        Let $P_{v^*}$ be a shortest  distinct-label $sv^*$-walk and let $P_u$ be a shortest prefix of $P_{v^*}$ that ends in a vertex $u$ such that $u \notin S$ and  $P_u$ is a distinct-label $su$-walk. Let $u'$ be the predecessor of $u$ in $P_u$ and let $P_{u'}=P_u-u$. If $\psi(T_{u'})+\psi(u'u) \ne \psi(T_u)$, we obtain $\alpha(u)\le w(T_{u'})+w(u'u)\le w(P_{u'})+w(u'u)=w(P_u)$. If $\psi(T_{u'})+\psi(u'u) = \psi(T_u)$, we obtain by the choice of $u$ that $\psi(P_{u'})=\psi(P_{u})-\psi(uu')\neq \psi(T_{u})-\psi(uu')=\psi(T_{u'})$, so $P_{u'}$ is a distinct-label $su'$-walk. It follows by the choice of $u$ that we have $u' \in S$. This yields $\alpha(u)\le w(Q_{u'})+w(u'u)\le w(P_{u'})+w(u'u)=w(P_u)$. In either case, we obtain that $\alpha(u)\le w(P_{u})$.  By the nonnegativity of $w$, we obtain $w(Q_{v^*}) = \alpha(v^*) \le \alpha(u) \le w(P_{u}) \leq w(P_{v^*}) \le w(Q_{v^*})$. Hence equality holds throughout and $Q_{v^*}$ is a shortest distinct-label $sv^*$-walk in $D$.
    \end{claimproof}
    \cref{cl:end} and \cref{cl:add} imply the correctness of the algorithm.
\end{proof}

We will combine \cref{thm:walks} with the next observation which shows that instead of a shortest non-zero directed cycle, it is enough to find a shortest non-zero closed walk. The statement and proof are analogous to that of \cite[Lemma~34]{elmaalouly2023exact}.

\begin{lemma}[see \cite{elmaalouly2023exact}] \label{lem:walkcycle}
    Let $D$ be a digraph, $\psi\colon A(D) \to \Gamma$ a group labeling, and $w \colon A(D) \to \R$ a conservative weight function. There is a polynomial-time algorithm that, given as input a non-zero closed walk $P$, outputs a non-zero directed cycle with $w(C) \le w(P)$.
\end{lemma}
\begin{proof}
    It is well-known that we can decompose $P$ in linear time into a collection of directed cycles $C_1,\ldots,C_t$ such that each arc of $D$ appears in $C_1,\ldots,C_t$ in total as many times as it appears in $P$.
    Since $w$ is conservative, $0 \le w(C_i)$ for each $i \in [t]$.  
    As $w(P) = w(C_1)+\dots+w(C_t)$, this implies that $w(C_i) \le w(P)$ for each $i \in [t]$. By $0 \ne \psi(P) = \psi(C_1)+\dots+\psi(C_t)$, there exists $i \in [t]$ such that $\psi(C_i) \ne 0$. We output $C_i$.
\end{proof}

\cref{thm:walks} and \cref{lem:walkcycle} together give a polynomial-time algorithm for finding shortest non-zero directed cycles given a conservative weight function. We now give the proof of \cref{thm:graphalg} which we restate here.
\nzdircycle*
\begin{proof}
   For each vertex $v$, we compute a shortest non-zero closed walk $P_v$ from $v$ with \cref{thm:walks}. 
    Let $V'$ denote the set of vertices $v$ for which a non-zero closed walk starting from $v$ exists.
    Using \cref{lem:walkcycle}, for each $v \in V'$ we obtain a non-zero directed cycle $C_v$ with $w(C_v) \le w(P_v)$. Then, there is no non-zero directed cycle if $V'$ is empty, otherwise  
    $C \coloneqq \argmin_{v \in V'} w(C_v)$ is a shortest non-zero closed walk by $w(C) = \min_{v \in V'} w(C_v) \le \min_{v \in V'} w(P_v)$, in particular, it is a shortest non-zero directed cycle.
\end{proof}

\end{document}